\documentclass[11pt]{article}

\title{Distributed Symmetry Breaking on Power Graphs via Sparsification}
\newif\ifdraft
\draftfalse
\nonstopmode

\usepackage{graphicx}
\usepackage{floatrow}
\usepackage{verbatim}
\usepackage{amsmath} 
\usepackage{amssymb}
\usepackage{amsthm}
\usepackage{subfiles}
\usepackage{fullpage}
\usepackage{url}
\usepackage{scalerel}  % scaleto for subscripts
\usepackage[dvipsnames, table,xcdraw]{xcolor}  % colorbox
\usepackage{enumitem}  % \begin{enumerate}[label=(\alph*)]
\usepackage{pdfpages}
\usepackage{thmtools,thm-restate}
\usepackage{empheq}  % system of equations brace
\usepackage[ruled,vlined,noend]{algorithm2e}  % pseudo code
\usepackage{dsfont}  % \mathds{1} gives a nice looking characteristic function
\usepackage{xcolor}
\usepackage{tcolorbox}
\usepackage[unicode, hidelinks]{hyperref}
\usepackage{cleveref}
\usepackage{ulem}  % strikethrough command sout
\usepackage{tablefootnote}
\usepackage{todonotes}

\newcommand{\pr}{\mathbb{P}}
\newcommand{\ev}{\mathbb{E}}
\newcommand{\One}{\mathds{1}}

\newcommand*{\QED}{\null\nobreak\hfill\ensuremath{\square}}

\newcommand{\ra}{\rightarrow}

\DeclareMathOperator{\poly}{poly}

\DeclareMathOperator*{\argmin}{argmin}

\newcommand{\IS}{\mathcal{I}}

\newcommand{\congest}{$\mathsf{CONGEST}$\xspace}
\newcommand{\mpc}{$\mathsf{MPC}$\xspace}
\newcommand{\local}{$\mathsf{LOCAL}$\xspace}
\newcommand{\diam}{\mathsf{diam}} 
\newcommand{\dist}{\operatorname{dist}}
\newcommand{\ID}{\operatorname{ID}}
\newcommand{\ball}{\operatorname{Ball}}

\newcommand{\LOCAL}{$\mathsf{LOCAL}$\xspace}

\newcommand{\CONGEST}{$\mathsf{CONGEST}$\xspace}

\newtheorem{theorem}{Theorem}[section]
\newtheorem{lemma}[theorem]{Lemma}
\newtheorem{corollary}[theorem]{Corollary}
\newtheorem{definition}[theorem]{Definition}

\newtheorem{claim}[theorem]{Claim}

\newtheorem*{lemma*}{Lemma}

\newcommand{\maxdeg}{\widehat{{\Delta}}} 
\newcommand{\maxactivedeg}{\Delta_A}
\newcommand{\detsparsification}{DetSparsification\xspace}
\newcommand{\bandwidth}{\mathsf{bandwidth}}

\newcommand{\border}{\text{Border}}
\newcommand{\lovasz}{Lov\'{a}sz\xspace}

\newcommand{\myemail}[1]{\,$\cdot$\, {\small #1}}
\newcommand{\myaff}[1]{\,$\cdot$\, {\small #1}\par\medskip}

\newenvironment{myabstract}
{\list{}{\listparindent 1.5em%
        \itemindent    \listparindent
        \leftmargin    1cm
        \rightmargin   1cm
        \parsep        0pt}%
    \item\relax}
{\endlist}

\newenvironment{mycover}
{\list{}{\listparindent 0pt
        \itemindent    \listparindent
        \leftmargin    1cm
        \rightmargin   1cm
        \parsep        0pt}%
    \raggedright
    \item\relax}
{\endlist}

\begin{document}

\begin{mycover}
{\huge\bfseries\boldmath Distributed Symmetry Breaking on Power Graphs via Sparsification \par}
\bigskip
\bigskip
\bigskip
\textbf{Yannic Maus}
\myemail{yannic.maus@ist.tugraz.at}
\myaff{TU Graz}

\textbf{Saku Peltonen}
\myemail{saku.peltonen@gmail.com}
\myaff{Aalto University}

\textbf{Jara Uitto}
\myemail{jara.uitto@aalto.fi}
\myaff{Aalto University}

\bigskip

\end{mycover}
\begin{myabstract}
\noindent\textbf{Abstract.}
In this paper we present efficient distributed algorithms for classical symmetry breaking problems, maximal independent sets (MIS) and ruling sets, in power graphs. 
We work in the standard \congest model of distributed message passing, where the communication network is abstracted as a graph $G$.
Typically, the problem instance in \congest is is identical to the communication network $G$, that is, we perform the symmetry breaking in $G$.
In this work, we consider a setting where the problem instance corresponds to a power graph $G^k$, where each node of the communication network $G$ is connected to all of its $k$-hop neighbors.

A $\beta$-ruling set is a set of non-adjacent nodes such that each node in $G$ has a ruling neighbor within $\beta$ hops; a natural generalization of an MIS.
On top of being a natural family of problems, ruling sets (in power graphs) are well-motivated through their applications in the powerful \emph{shattering} framework [BEPS JACM'16, Ghaffari SODA'19] (and others).
We present randomized algorithms for computing maximal independent sets and ruling sets of $G^k$ in essentially the same time as they can be computed in $G$. 
Our main contribution is a deterministic polylogarithmic time  algorithm for computing  $k$-ruling sets of $G^k$, which (for $k>1$) improves \emph{exponentially} on the current state-of-the-art runtimes.  
Our main technical ingredient for this result is a deterministic sparsification procedure which may be of independent interest. 

 We also revisit the shattering algorithm for MIS [BEPS J'ACM'16] and present different approaches for the  post-shattering phase. Our solutions are algorithmically and analytically simpler (also in the LOCAL model) than existing solutions and obtain the same runtime as [Ghaffari SODA'16].  
\end{myabstract}

\thispagestyle{empty}

\clearpage
\thispagestyle{empty}
\tableofcontents
\clearpage
\setcounter{page}{1}
\section{Introduction}
In this paper we provide efficient deterministic and randomized algorithms for symmetry breaking problems on power graphs, that is, we compute maximal independent sets and ruling sets on the power graph $G^k$ for an integer $k$, where $G$ is the input graph. 
To illustrate the setting, let us define the central problem  of our work. 
A \emph{maximal independent set (MIS)}  $S$ of a graph $G$ is a set of non-adjacent nodes such that every node $v$ of $G$ is \emph{dominated} by a node in $S$, that is, there is a node in $S$ with distance at most $1$ from $v$. \emph{Ruling sets} generalize this notion by relaxing the distance of domination. In the power graph $G^k$, any two nodes nodes of $G$ are connected if they are at most $k$ hops apart. Hence, in an MIS of $G^k$ any two nodes have distance at least $k+1$ and every node of $G$ is dominated by a node of $S$ with distance at most $k$. 
MIS and ruling sets have been studied extensively in the classic message passing models of distributed computing, i.e., the \LOCAL and the \CONGEST model, e.g., see \cite{Luby:1986ub,alon86,BEPS16,awerbuch1989network,ghaffari16_MIS,Gha19,SEW13,KMW18,BBKO22}.  

Understanding symmetry breaking on power graphs is crucial as they appear naturally in various settings. 
A classic example is given by the frequency assignment problem. In order to avoid interference in a network of wireless transmitters, one wants to assign frequencies to the nodes of a communication network such that all neighbors of each node receive different frequencies. The problem is a vertex coloring problem on the power graph $G^2$ \cite{KMR01,HKM20,HKMN20}. 
Problems on power graphs also appear as subroutines when solving problems for $G$. One example is the state-of-the-art randomized algorithm to compute an MIS (of $G$) that relies on the computation of ruling sets of power graphs, both in the \LOCAL model \cite{BEPS16,ghaffari16_MIS} and in the \CONGEST model \cite{Gha19}. A second example is the current state-of-the-art deterministic algorithms for computing MIS \cite{FGG22}. A third example where such ruling sets appear as subroutines is the algorithm to compute spanners in \cite{EM19}. In \Cref{sec:reasons}, we provide further examples and further motivate (symmetry breaking) problems on power graphs. The main model of our work is the \CONGEST model of distributed computing. 

\vspace{-1mm}
\paragraph{The challenges of working with power graphs in the \CONGEST model.}  In the \LOCAL and the \CONGEST model, a communication network is abstracted as an $n$-node graph $G$ with nodes representing computing entities and edges communication links \cite{linial92,peleg2000distributed}. Nodes are equipped with $O(\log n)$-bit IDs. In order to solve some problem in the network, the nodes  communicate with each other in synchronous rounds. In each round, nodes are allowed to perform arbitrary local computations and send messages  to each of their neighbors in $G$. In the \LOCAL model messages can be of \emph{unbounded} size, while in the \CONGEST model message sizes are restricted to $O(\log n)$ bits.  The \emph{time complexity} of an algorithm is the number of rounds until each node has computed its own part of the solution, e.g., whether it is contained in an independent set or not.  Classically in the literature, an algorithm for a problem---think of the MIS problem--- assumes that the communication network is also the problem instance. In contrast, in our work the power graph $G^k$ serves as the problem instance while $G$ remains the communication network. In the \LOCAL model, the setting does not yield any major difficulties as an algorithm designed for $G$ (e.g., to compute an MIS of $G$), can be run on $G^k$ with a multiplicative overhead of $k$ rounds (to compute an MIS of $G^k$). However, such a statement is not true in the \CONGEST model. In fact, in the \CONGEST model, a node does not even know its degree in the problem instance $G^k$ and if many vertices want to send different messages to their neighbors in $G^k$, congestion appears in the communication network and the messages cannot be delivered efficiently. 
This seemingly small difference has huge effects and makes it much more challenging to construct algorithms.

\subsection{Our contributions.} 
An \emph{$r$-ruling set} $S$ of a graph $G$ is a set of non-adjacent nodes such that for every node $v$ of $G$ there is a node in $S$ with distance at most $r$ from $v$. Hence, in an $r$-ruling set of $G^k$ any two nodes have distance at least $k+1$ and the domination distance is $r\cdot k$.  In the literature this often appears as an $(\alpha,\beta)$-ruling set where $\alpha$ specifies the minimum distance between nodes and $\beta$ the domination distance. As ruling sets relax the domination guarantee of maximal independent sets, they are usually easier to compute. 
Often, ruling sets are sufficiently powerful to replace MIS computations as subroutines in algorithmic applications. 
Our main  contribution is the first efficient deterministic algorithm to compute $k$-ruling sets for $G^k$. Throughout the paper $\widetilde{O}(x)$ omits factors that are logarithmic in $x$ and $\Delta$ refers to the maximum degree of the graph $G$, even when we are solving a problem on the power graph $G^k$.

\begin{restatable}[$k$-ruling set of $G^k$]{theorem}{corDetRulingFinal}
    \label{corollary:kksquaredRulingSet}
    Let $k \ge 1$ be an integer (potentially a function of $n$). There is a deterministic distributed algorithm that computes a $k$-ruling set of $G^k$ in polylogarithmic time in the \CONGEST model. More detailed, the round complexity is $\widetilde{O}(k^2 \cdot \log^4 n \cdot \log \Delta)$ rounds. 
\end{restatable}
For a constant $k>1$, \Cref{corollary:kksquaredRulingSet} improves \emph{exponentially} on the previous state of the art that required $O(n^{1/k})$ rounds \cite{EM19}.  This previous algorithm is an extension to $G^k$ of a  classic deterministic algorithm for computing $O(\log_B n)$-ruling sets (of $G$) in $O(B\cdot \log_B n)$ rounds \cite{awerbuch1989network,SEW13,henzinger2016deterministic,KMW18,EM19}, where $B$ is a parameter that can trade domination for runtime.

For $k=1$, \Cref{corollary:kksquaredRulingSet} computes a $(2,1)$ ruling set (aka an MIS), for $k=2$, it computes a $(3,4)$ ruling set and for $k=3$ it is a $(4,9)$ ruling set. 
 Our main technical contribution in order to obtain \Cref{corollary:kksquaredRulingSet}  is a novel sparsification procedure. On a very high level, we turn $G^k$ into a much sparser representation $\bar{G}$ such that an MIS of $\bar{G}$ is a good ruling set on the original $G^k$. The benefit is that we can communicate more efficiently on $\bar{G}$. 

 \smallskip

\noindent \textbf{Randomized symmetry breaking  in power graphs.}
Our second result is an algorithm for MIS of $G^k$ with (essentially) the same runtime as the state of the art for computing an MIS of $G$.

\begin{restatable}{theorem}{thmMISPower}
    %\label[theorem]{thm:MISPower}
    \label{thm:MISPower}
    There is a randomized distributed algorithm that computes a maximal independent set of $G^k$ in $\widetilde{O}(k^2\log \Delta \cdot \log\log n + k^4\log^5 \log n)$ rounds of the \congest model, with high probability.
\end{restatable}
The best result that can be achieved with previous work is an $O(k\log n)$-round algorithm based on Luby's algorithm for computing an MIS (see \Cref{sec:lubyGk} for details) \cite{Luby:1986ub,alon86}. 
\Cref{thm:MISPower} compares favorably to the state of art for computing an MIS of $G$ in $O(\log\Delta\log\log n + \poly\log\log n)$ rounds in the \CONGEST model \cite{Gha19}. 
The randomized complexity of MIS on $G^k$ in \LOCAL is $O(k^2\log\Delta +k\poly\log \log n)$, simply by running the best known algorithm for MIS on $G$ on the power graph \cite{ghaffari16_MIS}. 
By combining \Cref{thm:MISPower} with the sparsification methods of \cite{KP12} and \cite{BKP14}  we obtain the following corollary. 
\begin{restatable}{corollary}{thmkkBetaRulingRand}
    \label{cor:kkBetaRulingRand}
    There is a randomized distributed algorithm that computes a $\beta$-ruling set of $G^k$ in $\widetilde{O}(\beta \cdot k^{1 + 1/(\beta-1)}  (\log \Delta)^{1/(\beta-1)} + \beta \cdot k \log \log n +  k^{4}\log^5 \log n)$ rounds. 
\end{restatable}
The only previous randomized algorithm for computing a ruling set of $G^k$ is a  $O(k\log\log n)$-ruling set algorithm that works in $O(k^2\log\log n)$ rounds \cite{Gha19}. The technique can inherently not compute $O(1)$-ruling sets of $G^k$.  

\newpage

\noindent \textbf{Simplifying shattering for MIS and Ruling Sets on $G$ (in \LOCAL and \CONGEST).}  Ruling sets of power graphs are essential in the the current state-of-the-art randomized algorithms to compute an MIS in \LOCAL \cite{BEPS16,ghaffari16_MIS} and the \CONGEST model \cite{Gha19}. In our last contribution, we revisit\footnote{We revisited the details of these algorithms as they rely on the shattering framework that we also use to obtain  \Cref{thm:MISPower}. However, the setting in \Cref{thm:MISPower} is more involved. }  the \emph{current} state-of-the-art algorithms for computing an MIS in $G$ in \LOCAL \cite{BEPS16,ghaffari16_MIS} and  \CONGEST  \cite{Gha19}.  These results are based on the  by-now-standard \emph{shattering technique} that uses a random process to solve the respective problem on most of the graph such that the remaining unsolved parts only consists of  (many) \emph{small connected components}. 
The components are small in the sense that a certain power graph ruling set of a single component has at most logarithmically many nodes. This fact is exploited to solve the remaining problem on these small components efficiently. 
The challenge here is that to algorithmically exploit the fact that these ruling sets are of a small size, one also has to  compute them. Intuitively, one wants to just run a classic $(5,O(\log\log n))$-ruling set algorithm on each of the remaining connected components in parallel. The crux  is that the distance between ruling set nodes needs to be measured in $G$ (for the ruling set to be small) and not just within the small component. But a node cannot tell efficiently whether a node in distance $5$ in $G$ is in the same small component as itself or not. Hence, the nodes cannot easily determine whether they can both be contained in the ruling set or not and one cannot readily treat different small components independently. 

Thus, the arXiv version  of \cite{BEPS16}  presents an involved two-phase shattering procedure with an involved analysis. It turns out that the presented analysis of bounding the size of the ruling set has a small (but crucial) mistake (see Sections~\ref{ssec:nutshell} and \ref{sec:revisedShattering}  for details). 
The journal version of the respective work presents a fix through an even more involved analysis \cite{BEPS16}. 
The later works \cite{ghaffari16_MIS,Gha19} build upon the arXiv version of \cite{BEPS16}. While we believe that both works can be adapted to build upon the arguments in the journal version of \cite{BEPS16}, this cannot be done in a black-box manner. Both works use a different ruling set algorithm than \cite{BEPS16} and the internals of the ruling set algorithm are crucial for the arguments in \cite{BEPS16}.

As our last contribution, we revisit the shattering framework and present two different algorithmic and analytical solutions for computing an MIS that work in the \LOCAL and the \CONGEST model. 
While these do not improve upon the complexity of \cite{BEPS16,ghaffari16_MIS,Gha19}, our approach is simpler and provides a fix for results that build upon the arXiv version of \cite{BEPS16}.

\begin{restatable}{theorem}{thmMISrepaired}
\label{thm:MISrepaired}
There are randomized \LOCAL and \CONGEST algorithms that w.h.p.~compute a maximal independent set  on any $n$-node graph with maximum degree $\Delta$ and that run in $O(\log\Delta)+\poly\log\log n$ rounds and in $O(\log\Delta\log\log n)+\poly\log\log n$ rounds, respectively. 
\end{restatable}

Even though \Cref{thm:MISrepaired} deals with the classic setting of computing an MIS of the graph $G$ (and not its power graph), the flaw in the current argument is closely related to properly dealing with  ruling sets  in power graphs. Some other works in the literature build upon the ideas of the arXiv version of \cite{BEPS16}, e.g.,  algorithms for $\Delta$-coloring \cite{GHKM18} or the \lovasz Local Lemma \cite{FG17}, and thus undergo a similar issue. However, we note that the aforementioned power graph ruling sets are only relevant when $\Delta\geq\poly\log n$ as otherwise one can use simpler approaches to exploit the smallness of the components. 
But for $\Delta\geq\poly\log n$ it is not necessary to correcting the flaw in the two papers. 
For the $\Delta$-coloring problem a very recent paper provided a genuinely different algorithm to solve the problem for $\Delta\geq \poly\log n$, making a fix unnecessary \cite{FHM22}. For the \lovasz Local Lemma problem, one can use the algorithm by Chung, Pettie and Su  to solve the problem in $O(\log n)$ rounds \cite{CPS17}, which is faster than the $O(\Delta^2)+\poly\log\log n$ rounds of \cite{FG17} when $\Delta\geq \poly\log n$.  Last but not least, let us remark that the state-of-the-art randomized $(\Delta+1)$-coloring algorithms are not impacted by the flaw as they only rely on shattering when $\Delta\leq \poly\log n$ \cite{CLP20,HN21}. 
For more details on the background of \Cref{thm:MISrepaired} and our solutions, see \Cref{ssec:nutshell} and \Cref{sec:revisedShattering}. 
Note that \Cref{thm:MISrepaired} is also central to countless results in the literature as the algorithm belongs to the most used subroutines in the area. Also \Cref{thm:MISPower} and \Cref{cor:kkBetaRulingRand} rely on extensions of the algorithmic ideas used for \Cref{thm:MISrepaired}. 

\subsection{Why should we care about problems on power graphs?}
\label{sec:reasons}
 We have already mentioned that ruling sets of power graphs are important subroutines.  In this section, we provide various further reasons for studying \CONGEST algorithms on power graphs.

\noindent \textbf{Derandomization and learning distant information.}
The main challenge when working with power graphs is that  even though we aim at solving problems on $G^k$, the communication network is just $G$ itself and a node cannot immediately communicate with its neighbors in $G^k$. For an example, consider the simple task of learning some small individual piece information from each of your neighbors. In $G$, this problem is trivial and can clearly be achieved in a single \CONGEST round. 
However, as soon as we turn to $G^2$, we suffer a huge overhead in the number of rounds. Solving this task requires $\Omega(\Delta)$ rounds in the worst case, where $\Delta$ is the maximum degree of the graph which may be very large. 

But learning large amounts of information that is not stored at immediate neighbors of $G$ is an important ingredient of recent efficient \LOCAL algorithms. 
Prime examples of this behavior are the general derandomization results \cite{SLOCAL17,GHK17,RG20} in the \LOCAL model. Note that besides the aspect of learning large amounts of information, another crucial ingredient to these results is the computation of so  called network decompositions of the power graph $G^k$, usually for a non-constant $k$. 
It is a major open problem of the area to determine for which type of problems  such derandomizations can be obtained in the \CONGEST model.

On the positive side, for selected problems on $G$, there are efficient \CONGEST algorithms  that are all based on limiting the amount of information that one has to aggregate for an efficient derandomization of a (simple) randomized procedure.  Examples are given by algorithms for $(\Delta+1)$-coloring \cite{BKM19,GK20}, MIS \cite{censor2017derandomizing,FGG22}, or minimum dominating set approximations \cite{DKM19,FGG22}. 
Our sparsification results are also based on derandomization (in a more extreme setting) and hence add to the class of problems that can be efficiently derandomized in the \CONGEST model. 
In \Cref{ssec:nutshell} we detail on how we circumvent the necessity to learn large amounts of information used by the sparsification process in the \LOCAL model.

\noindent \textbf{Power graph and virtual graph problems  as subroutines.} Another prevalent ingredient to many recent results is that they solve intermediate problems on virtual graphs. 
For example, in the state of the art deterministic algorithm for MIS in \LOCAL and \CONGEST, we must simulate an algorithm (for some intermediate problem) on $G^2$ \cite{FGG22}.
Hence, despite the fact that we consider the setting where the communication network equals the input graph, intermediate steps require solving problems on power graphs. 
In \local, the unbounded message sizes allow us to communicate in $G^2$ with a constant overhead and hence, algorithms for the intermediate steps in the virtual graph $G^2$ are straight-forward to implement with a constant overhead in the runtime.
In contrast, handling these intermediate steps is much more involved in the \congest model.
In general, we  believe that the study of power graph problems contributes to the general important theme of detaching the input graph from the communication network.  Problems on power graphs serve as a clean abstraction to develop tools that can be used whenever the problem instance and the communication network are not exactly the same.

\noindent \textbf{Robustness of algorithms and techniques.} 
More broadly, studying problems on power graphs also serves as a clean abstraction to develop  algorithmic techniques that are \emph{robust} in the sense that they can work in settings with stronger communication restrictions. In the long run we also expect this kind of research to lead to algorithms that are more model independent. 
For example, many recent \mpc algorithms are fast implementations of  communication efficient \LOCAL or \CONGEST algorithms~\cite{GU19,CDP19,Chang2019b,Czumaj2021,Balliu2022}. 
An extreme case of communication efficiency  are algorithms in the beeping model which is similar to the \CONGEST model with $1$-bit messages. 
An excellent example is the pre-shattering phase of Ghaffari's MIS algorithm that when computing an MIS for $G$ works in the beeping model~\cite{Ghaffari2017}. 
Due to its robust design, it has been used (with slight adaptations) to obtain the state-of-the-art in other settings, such as \mpc and LCA~\cite{GU19,Ghaffari2022}.
\Cref{thm:MISPower} is also based on an extension of this result to $G^k$ that works in a stronger version of the beeping model, allowing $O(\log\log n)$-sized messages. See \Cref{sec:randomizedPowerGraphs}  for details. 

Another aspect that makes algorithms for power graphs more robust is that they have to operate in the setting where nodes do not know their degree (in $G^k$).  Even though there are ways to remove the necessity of knowing your degree from an algorithm \cite{KSV13}, these cannot be applied in a black-box manner in all settings.

\subsection{Further related work}

\label{app:tableRelatedWork}
\begin{table}[p]
    \centering
    \renewcommand{\arraystretch}{1.4}
    \hspace*{-1.4cm} %% USE WITH AUTHORYEAR CITATION STYLE
    \begin{tabular}{|l|l|l|l|l|l|} \hline
    \textbf{Paper}      & \textbf{Result} & & \textbf{Model} & \textbf{Running time} \\ \hline\hline
    \cite{RG20} & MIS & det. & \LOCAL & $O(\log^7 n)$  \\ \hline
    \cite{censorhillel2016derandomizing} + \cite{RG20} & MIS & det. & \congest & $O(\log^8 n)$ \\ \hline
    \cite{FGG22} & MIS & det. & \congest & $\widetilde{O}(\log n \cdot \log^2 \Delta)$  \\ \hline
    \cite{awerbuch1989network}, \cite{henzinger2016deterministic} & $(k+1, k\cdot \log n)$ & det. & \congest & $O(k \cdot \log n)$ \\ \hline
    \cite{SEW13}, \cite{KMW18}$^*$  & $(k+1,k \lceil \log_B n \rceil)$ & det. & \congest & $O(k \cdot B \cdot \log_B n)$ \\ \hline
    \cite{SEW13}, \cite{KMW18} & $(k+1,kc)$ & det. & \congest & $O(k \cdot c \cdot n^{1/c})$ \\ \hline
    \rowcolor[HTML]{EFEFEF} 
    \textbf{New} & $(k+1, k^2)$ & det. & \congest & $\widetilde{O}(k^2  \log^4 n \cdot \log \Delta)$  \\
    \hline\hline
    \cite{Luby:1986ub} & MIS of $G^k$ & rand. & \congest & $O(k \cdot \log n)$  \\ \hline
    \cite{BEPS16} + \cite{RG20} & MIS & rand. & \local & $O(\log^2 \Delta + \poly \log \log n)$\\ \hline
    \cite{ghaffari16_MIS} + \cite{RG20}&MIS& rand. & \local & $O(\log \Delta  + \poly\log\log n)$  \\ \hline
    \cite{Gha19} + \cite{GGR20} & MIS & rand. & \congest & $O(\log \Delta \cdot \log \log n + \poly\log\log n)$  \\ \hline
    \rowcolor[HTML]{EFEFEF} 
    \textbf{New} & MIS of $G^k$ & rand. & \congest & $\widetilde{O}(k^2\log \Delta \cdot \log\log n + k^4\log^5 \log n)$ \\ 
    \hline\hline
    \cite{GV07} & $(-, O(\log \log n))^\dagger$ & rand. & \congest & $O(\log\log n)$ \\ \hline
    \cite{KP12} & $(-,2)^\ddagger$ & rand. & \congest & $O\left(\log \Delta / (\log n)^\epsilon\right)$ \\ \hline
    \cite{BKP14}+\cite{ghaffari16_MIS}+\cite{RG20} & $(2,\beta)$ & rand. & \local & $O(\beta \cdot (\log \Delta)^{1/\beta} + \poly \log \log n)$ \\ \hline
    \cite{BKP14}+\cite{Gha19}+\cite{GGR20} & $(2,\beta)$ & rand. & \congest & $O(\beta \cdot (\log \Delta)^{1/(\beta-1)} + \poly \log \log n)$\\ \hline
    \cite{Gha19} & $(k+1, O(k^2 \log \log n))$ & rand. & \congest & $O(k^2 \log \log n)$ \\ \hline 
    \rowcolor[HTML]{EFEFEF} &&&& \\[-6mm]
    \rowcolor[HTML]{EFEFEF} 
    \textbf{New} & $(k+1, k\beta)$ & rand. & \congest & $\begin{aligned} 
    &O\big(\beta k^{1 + 1/(\beta-1)}  (\log \Delta)^{1/(\beta-1)}\big) \\ &+\widetilde{O}\big(\beta k \log \log n +  k^{4}\log^5 \log n\big) \end{aligned}$  \\ \hline
    \end{tabular}
    \caption{
        Summary of related work on the MIS and ruling set problem in the \congest and \local models. The \congest algorithms that can be used to solve the corresponding problem on $G^k$ for any $k \ge 1$ are explicitly specified. In \local, all algorithms can trivially be used for $G^k$ with a slowdown factor of $\poly k$ (not shown in the runtimes). 
        (*) Additionally, given a $\gamma$-coloring of $G^k$, there is a deterministic $(k+1,k \lceil \log_B \gamma \rceil)$-ruling set algorithm running in $O(k \cdot B \cdot \log_B \gamma)$ rounds. $(\dagger, \ddagger)$  The respective algorithms do not satisfy the independence condition of ruling sets but compute bounded degree subgraphs with the respective domination property.  The degree bounds are ($\dagger$) $O(\log^5 n)$ and  $O(\log n \cdot 2^{(\log n)^\epsilon})$ for $(\ddagger)$.}
    \label{table:relatedwork}
\end{table}

 Recent years have seen several results for problems in power graphs in the \CONGEST model, reaching from verifying solutions efficiently (or showing that this is not possible) \cite{FHN20}, and answering several computational questions in the settings, e.g., for the already discussed problem of distance-$2$ coloring \cite{HKM20,HKMN20} or optimization problems \cite{BCMPP20}.

\paragraph{Ruling sets.} We first focus on results in \LOCAL. 
We have already discussed the deterministic ruling set algorithm of \cite{SEW13,henzinger2016deterministic,KMW18}. Actually, the result is more general and it provides an $O(\log_B C)$-ruling set of $G$ in $O(B\cdot \log_BC)$ time if the graph is equipped with a vertex coloring with $C$ colors. 
By choosing the parameter $B$  appropriately and combining it with Linial's algorithm that computes a $O(\Delta^2)$-coloring in $O(\log^* n)$ rounds, one can compute a $\beta$-ruling set in $O(\beta \cdot \Delta^{2/\beta}+\log^* n)$ rounds. 

Gfeller and Vicari~\cite{GV07} provide a randomized sparsification algorithm to compute a $O(\log \log n)$-dominating set\footnote{The definition of ruling sets is different in~\cite{GV07}: what we refer to as a $t$-dominating set is called a $t$-ruling set in \cite{GV07}, while a $(2,t)$-ruling set is called an independent $t$-ruling set.} $S \subseteq V$ in $O(\log \log n)$ rounds, such that the maximum degree of $G[S]$ is $O(\log^5 n)$. This can be combined with the aforementioned approach to obtain a randomized  an $O(\log\log n)$-ruling set algorithm with round complexity $O(\log\log n)$. Ghaffari extends this approach to $G^k$ in the \CONGEST model and obtains an $O(k\log\log n)$-ruling set in $O(k^2\log\log n)$ rounds \cite{Gha19}. 

Kothapalli and Pemmaraju \cite{KP12} provide another sparsification method that computes a dominating set with degree $O(\Delta'\log n)$ in $O(\log\Delta/\log \Delta')$ rounds. Ghaffari uses multiple iterations of this sparsification method to compute a $\beta$-ruling set in $O(\beta \log^{1/\beta}\Delta)+\poly\log\log n$ rounds in \LOCAL \cite{ghaffari16_MIS}. He also provides similar but slightly weaker results in \CONGEST \cite{Gha19}.

The recently developed powerful round elimination technique \cite{B19}  has been used to prove lower bounds for the computation of ruling sets in \LOCAL \cite{BBO22,BBKO22}. Parameterized by the number of nodes the lower bounds for $\beta$-ruling sets is $\Omega(\log n/(\beta \log \log n))$ for deterministic algorithms and $\Omega(\log \log n/(\beta \log \log \log n))$ for randomized algorithms, as long as $\beta$ is at most $\approx~\sqrt{\log n/\log \log n}$ and $\approx\sqrt{\log \log n/\log \log \log n}$, respectively. 
As a function of the maximum degree $\Delta$, the lower bound is $\Omega(\beta \cdot \Delta^{1/\beta})$. 
In \Cref{table:relatedwork}, we provide an overview of known ruling set and MIS algorithms and contrast them with our results. 

\paragraph{Related work for graph sparsification.}
A fundamental building block of our most involved result (\Cref{corollary:kksquaredRulingSet}) is the method of graph sparsification~\cite{Ahn12, Hegeman2015, Ghaffari2017, GU19, Ass19, CDP19, Czumaj2021}.
Roughly speaking, the idea is to turn the input graph into a sparser representation that still allows us to solve the given problem.
Sparsification has been used in many different settings and computational models and often the exact properties we want from the sparsification vary depending on the setting.

\textit{Global sparsification.} In models with $\widetilde{O}(n)$ memory, such as streaming, sketching, congested clique, and linear memory MPC, the goal is to reduce the total (global) number of edges in the graph.
If the sparser representation has roughly $n$ \emph{edges}, then it can be processed locally in constant time.
For example, the current state-of-the-art for connectivity/MST~\cite{Jurdzinski2018, Nowicki2021}, MIS~\cite{Ahn2015, Ghaffari2018}, and $(\Delta + 1)$-vertex-coloring (and list-coloring)~\cite{Chang2019b, Ass19, Alon2020, Czumaj2021} are based on this method.

\textit{Local sparsification.} Another sparsification approach is to sparsify the graph \emph{locally}.
To get an intuition, suppose that the output of each node $v$ depends only on some local information, i.e., the output of node $v$ can be decided by examining the graph in the small $T$-hop neighborhood around node $v$.
Then, a local sparsification reduces the number of edges in that neighborhood and, at the same time, preserves the property that a correct output for $v$ can be determined from the local neighborhood of $v$ in the sparser graph.
Combined with the memory-hungry graph exponentiation technique~\cite{Lenzen2010}, this approach has been successfully used in sublinear models such as the low-space MPC model.
For example, the current state-of-the-art for MIS~\cite{GU19, CDP19} and $(\Delta + 1)$-vertex-coloring~\cite{Chang2019b, Czumaj2021} are based on this method.

In the context of the \congest model, locally sparsifying the communication graph is a natural way to avoid congestion.
In one previous work, this line of thinking was used for vertex coloring $G^2$~\cite{HKM20}.
However, in the context of coloring, one can split the \emph{problem} into independent instances with disjoint color palettes.
This property creates a fundamental difference between coloring and other symmetry breaking problems, such as MIS and ruling sets, which do not enjoy the luxury of splitting into disjoint instances.
Sparsified graphs are also of independent interest, for example in the context of finding distance preserving spanners, both in cases of linear memory models~\cite{dinitz2020, Dory2021} and local sparsification~\cite{forster2021}.

\section{Notation and \texorpdfstring{$k$}{k}-wise independent random variables}
\label{sec:preliminaries}
%\noindent\textbf{Notation.} 
We always use $G=(V,E)$ to refer to the original input graph. Similarly, the degree $d(v)$ and (non-inclusive) neighborhood $N(v)$ of a vertex do not change throughout our algorithms. For $s \ge 0$, the \textit{power graph} $G^s$ is the graph where $V(G^s) = V$ and $E(G^s) = \{ \{v, w\} \in V \times V : \dist_G(v,w) \le s \}$. 
A subgraph of the power graph induced by a set of nodes $X$ is denoted $G^s[X]$. Note that this is not the same as $(G[X])^s$: the latter only contains edges formed by paths only using nodes in $X$. 
$N^s(v)$ is called the \textit{distance-$s$ neighborhood of $v$}, which is the neighborhood of $v$ in $G^s$. Let $d^s(v) := |N^s(v)|$. For any $X \subseteq V$, let $N^s(X) := \cup_{v \in X}N^s(v)$. We use \textit{distance-$s$ $X$-neighborhood} to refer to $N^s(v, X) := N^s(v) \cap X$. \textit{Distance-$s$ $X$-degree} is defined as $d^s(v, X) := |N^s(v, X)|$. 

$\IS \subseteq V$ is $\alpha$-independent in $G$, if for all distinct $v, w \in \IS$, $\dist_G(v,w) \ge \alpha$. For $\alpha=2$, we simply say that $\IS$ is \textit{independent} in $G$. For any $S \subseteq V$, $Q \subseteq S$ is a $\beta$-dominating set of $S$, if for all $u \in S$, there exists some $v \in Q$ such that $\dist_G(u,v) \le \beta$. When $S=V$, we say that $Q$ is a $\beta$-dominating set. An $(\alpha, \beta)$-\textit{ruling set} of a graph $G=(V,E)$ is a subset $Q \subseteq V$, such that $Q$ is $\alpha$-independent and $\beta$-dominating. $\alpha$ and $\beta$ are generally referred to as the \textit{independence} and \textit{domination} parameters, respectively. Note that a $(2,1)$-ruling set is a maximal independent set of $G$ and a $(k+1,k)$-ruling set is a maximal independent set of $G^k$. A set $S \subseteq V$ is $k$-connected in $G$ if for all $S' \subset S: \dist_G(S', S \setminus S) \le k$. Equivalently, $S$ is $k$-connected if $G^k[S]$ is connected.

\textit{A breadth-first search tree} (BFS-tree) $T \subseteq G$ with root $r \in V$ and depth $s \ge 0$ is a tree such that $V(T) = N^s(v) \cup \{v\}$ and $\forall v \in V(T): \dist_T(v,r) = \dist_G(v,r)$. We say that a BFS tree $T\subseteq G$ is known in the distributed setting, if each $v \in V(T)$ knows its immediate \textit{ancestor}$(T,v) \in N(v)$, \textit{descendants}$(T,v) \subseteq N(v)$ and the ID of the root node \textit{root}$(T) \in V$. A depth-$s$ BFS tree includes all nodes in the distance-$s$ neighborhood of the root node. BFS trees are non-unique whenever the underlying graph is not a tree. We say that a BFS tree $T$ is \textit{spanning} when $V(T) = V(G)$.

\vspace{3mm}

\begin{definition}[Network decomposition]
    \label{def:networkDecomp}
        A $(c,d)$-\textit{network decomposition} of $G$ is a partition of $V$ into disjoint $d$-diameter clusters. There is a $c$-coloring of the clusters, such that adjacent clusters have different colors. The cluster diameter can be measured in two ways. A cluster $C \subseteq V$ has \textit{strong-diameter} $d$, if any two vertices of $C$ have distance at most $d$ in $G[C]$. In a \textit{weak-diameter} cluster, any two vertices of $C$ have distance at most $d$ in $G$. For communication between the nodes in a cluster, we associate $C$ with a \textit{Steiner tree} $T_C$, which is a tree subgraph of $G$ with a set of \textit{terminal} and \textit{non-terminal} nodes. Terminal nodes are equal to the nodes in the cluster, while non-terminal nodes may be any nodes in the graph. The diameter of the Steiner tree is $O(d)$. The Steiner trees of clusters may overlap with each other. We say that a network decomposition has \textit{congestion} $\tau$, if each edge is in at most $\tau$ Steiner trees for clusters of any single color. In power graphs, a \textit{network decomposition of $G^k$} clusters the vertices with a \textit{separation} of $k+1$. For any two clusters $C,C'$ of the same color, it is required that $\dist_G(C,C') > k$.
\end{definition}

\paragraph{k-wise independent variables.} A collection of discrete random variables $X_1, \dots, X_n$ is $k$-\textit{wise independent} if for any $I \subseteq [n]$ with $|I|\le k$ and any values $x_i$, we have $\pr(\wedge_{i \in I} X_i = x_i) = \prod_{i \in I} \pr(X_i = x_i)$. We can simulate such variables by picking a random hash function from a family of $k$-wise independent hash functions: 

\begin{definition}
    For $N, L, k \in \mathbb{N}$ such that $k \leq N$, a family of functions $\mathcal{H}=\{h:[N] \rightarrow[L]\}$ is $k$-wise independent if for all distinct $x_{1}, \ldots, x_{k} \in[N]$, the random variables $h(x_{1}), \ldots, h(x_{k})$ are independent and uniformly distributed in $[L]$ when $h$ is chosen uniformly at random from $\mathcal{H}$. 
\end{definition}

\begin{lemma}[Corollary 3.34 in \cite{vadhan2012pseudorandomness}]
    \label{lem:hashfamily}
    For every $a, b, k$, there is a family of $k$-wise independent hash functions $\mathcal{H}=\{h:$ $\left.\{0,1\}^{a} \rightarrow\{0,1\}^{b}\right\}$ such that choosing a random function from $\mathcal{H}$ takes $k \cdot \max \{a, b\}$ random bits, and evaluating a function from $\mathcal{H}$ takes time $\poly(a, b, k) .$
\end{lemma}

\section{Technical overview }
\label{ssec:nutshell}
 
\paragraph{Deterministic Sparsification.}
Our main technical ingredient is a deterministic sparsification procedure. 
 From a high level point of view, we gradually sparsify the input graph, i.e., we compute a sequence of sparser node sets $V \supseteq Q_0\supseteq Q_1\supseteq \ldots \supseteq Q_k$, while ensuring that every node in $V$ remains within  constant distance of the set.
More precisely, for each $1 \le s \le k$, the distance from any node in $Q_{s-1}$ to $Q_s$ is at most $2$ in~$G^s$ (at most $2s$ in $G$), and the set $Q_{s}$ is sparse in $G^s$, that is, every node has a bounded number of $Q_s$-neighbors in $G^s$. We use the sparsity of $Q_s$ to efficiently limit the congestion when computing the sparser set $Q_{s+1}$. 
At the end $Q_{k-1}$ is sparse enough to efficiently simulate any algorithm on $G^{k}[Q_{k-1}]$, i.e., the subgraph of $G^{k}$ induced by $Q_{k-1}$. For example, to compute a $k$-ruling set of $G^k$, we can use sparsification to find $Q_{k-1}$ and then compute an MIS of $Q_{k-1}$ on the power graph $G^k$, where only nodes in $Q_{k-1}$ are allowed to enter the MIS.

We believe that this sparsification procedure is of independent interest and may be helpful for other problems and in other models of computation. We summarize it in the next lemma. 
Recall, the distance-$k$ $Q$-degree of a node $v$ is the number of neighbors in $N^k(v)\cap Q$. 
\begin{restatable}[Sparsification in Power Graphs]{lemma}{lemSparsification} \label{lem:sparsification}
    Let $k \ge 1$ (potentially a function of $n$). There is a deterministic distributed algorithm that, given %a graph $G=(V,E)$ with $n$ nodes, maximum degree~$\Delta$, and 
    a subset $Q_0 \subseteq V$, finds a set of vertices $Q \subseteq Q_0$ such that for all $v \in V$, 
    \begin{itemize}
        \item (bounded distance-$k$ $Q$-degree): $d^k(v,Q) \le 72\log n = O(\log n)$
        \item (domination): $\dist_{G}(v, Q) \le k^2 + k + \dist_{G}(v, Q_0)$
    \end{itemize}
    The algorithm runs in
    $O(\diam(G) \cdot k \cdot \log^2 n \cdot \log \Delta + k^2 \cdot \log \Delta)$ rounds in the \congest model.
\end{restatable}

Using a network decomposition , we can replace the diameter factor in \Cref{lem:sparsification} with a term that is polylogarithmic in $n$ (see \Cref{app:noDiameter} for more details and in particular \Cref{lem:ndSparsification} for the precise statement).  
For $k=1$ (and if initialized with $Q_0=V$), we obtain a polylogarithmic-round algorithm to compute a set with domination distance $2$ such that each $v \in V$ has at most $O(\log n)$ neighbors in $Q$. Note that the degree bound is a \emph{local property}. If one instead wants to compute a set with the same domination distance and \emph{globally} minimize its size, it is known that even constant approximation algorithms require near-quadratic time in  \CONGEST  \cite{BCMPP20}. 

\textbf{How does the sparsification work?} The core idea of our deterministic sparsification is to derandomize the following randomized sampling algorithm: First assume that the graph $G^k$ is $\Delta^k$-regular. 
We can sample  every node with probability $O(\log n/\Delta^k)$ into a set $Q$. 
Then every node in $V$ gets $\Theta(\log n)$ distance-$k$ $Q$-neighbors, with high probability. 
Even though this is a trivial $0$-round algorithm, it is non-trivial to derandomize it, as the constraints imposed by the nodes depend on the decision of nodes in distance $\Theta(k)$. 
Generally, there is little hope that one can directly (and deterministically) find good random bits (aka good conditions) in order to fulfill the constraints of all nodes.
Instead, we perform a more fine-grained sparsification method, in which we gradually sparsify the graph, each time slightly losing in the domination property. This can be viewed as first computing the set $Q_k$ for $k=1$, then for $k=2$, and so on.

Continuing in the randomized setting, we explain our approach for finding $Q_1$ in $G$. The sampling algorithm has two objectives. For all $v \in V$, the number of neighbors in $Q_1$ should be at most $O(\log n)$, while the distance to $Q_1$ should be at most 2 (or at most some constant). Since degrees are not uniform, it is difficult to find the right sampling probability that satisfies the two objectives for all nodes. 
Hence, the process is slowed down to $O(\log \Delta)$ \textit{stages}. All nodes in $Q_0$ start as \textit{active} ($Q_0=V$ in our applications). In each stage, we sample active nodes to $Q_1$. Initially, the sampling probability is low, so that high degree nodes get at least one sampled neighbor, while not exceeding the $O(\log n)$ bound. At the end of each stage, the distance-2 neighborhood of sampled nodes is deactivated. This guarantees that high degree nodes do not get more sampled neighbors in later stages, effectively decreasing the maximum \textit{active degree}. The decrease in active degree means that we can sample in the next stage with a slightly higher probability. In the end, only nodes with low active degree remain, which can be included in the sampled set. The result $Q_1$ is a 2-dominating set, while all $v \in V$ have at most $O(\log n)$ $Q_1$-neighbors.

\textbf{Sparsifying $G$.} Our deterministic sparsification for $G$ is based on derandomizing the previous sampling algorithm. The randomized analysis works with $O(\log n)$-wise independence. This allows simulating the randomness in one stage with a $O(\log^2 n)$-bit random seed. Derandomization is done stage by stage, using the method of conditional expectations. In each stage, we need to guarantee two events: (1) each node gets at most $O(\log n)$ sampled neighbors, while (2) each high active degree node gets at least one of its neighbors sampled. Both events only depend on the events in the immediate neighborhood in $G$, which makes the required information easily available for each node. The bits of the seed are fixed one by one. In order to fix a single bit $b$ (to $0$ or $1$) of the seed, each node computes conditional expectations for its two events, for both choices of $b$. To compute the conditional expectations, a node needs to know the values of the already fixed bits of the seed and the IDs of its active neighbors. In $G$, learning the relevant identifiers can be done in one round, because the two events are determined by the decisions of active neighbors in the immediate neighborhood. Then, the conditional expectations of all nodes are aggregated (summed up) at a leader node who can then decide on the better choice for the bit $b$. Then all nodes proceed with the next bit. The method of conditional expectation implies that all nodes' events  hold at the end of this process. 

\textbf{Sparsifying power graphs.} Fix some $1 \le s \le k$. The $s$th \textit{iteration} of the sparsification is simulated on the power graph $G^s$. The output of the previous iteration acts as the initial set of active nodes $Q_{s-1}$. The $s$th iteration results in $Q_s \subseteq Q_{s-1}$, where $Q_s$ is sparse in $G^s$, while weakening the domination distance of $Q_{s-1}$ by at most $2$ in $G^s$ (or $2s$ in $G$). The sum of increases in distance over $s$ iterations is $\sum_{j=1}^s 2j\leq s^2+s$, hence $Q_s$ is a $(s^2 + s)$-dominating set of $G$ (when initialized with $Q_0=V$).

A main challenge of our work is to ensure that all nodes can obtain the necessary information for derandomization (distance-$s$ neighbor's random bits and IDs) when sparsifying power graphs. 
To guarantee sparsity in $G^s$, all nodes must remain as observers (and also relay messages), taking part in the derandomization.
Here, in a nutshell, the sparsity with regard to the previous iteration helps to learn the required information. 
More formally, we build communication tools to efficiently run the algorithm on $G^s$, relying on the sparsity of $Q_{s-1}$.

\textbf{Communication tools.} In order to benefit from the sparsity of $Q := Q_s$, we develop communication tools (see \Cref{lem:comms}) that allow us to execute basic communication primitives on power graphs. The sparsity of $Q$ in $G^s$ can be used to efficiently send messages from $Q$ to their neighbors in $G^{s+1}$. The communication algorithms include sending broadcasts from nodes in $Q$ to their neighbors in $G^{s+1}$ in $O(s+\log n)$ rounds, and simulating one round of a \CONGEST algorithm on $G^{s+1}[Q]$ in $O(s + \log^2 n)$ rounds (then, one can basically assume that the algorithm is running on the communication network $G^{s+1}[Q]$ with $O(s+\log^2 n$) overhead). The efficiency is based on the bounded number $\maxdeg$ of distance-$s$ neighbors in $Q$ for \emph{all} nodes in $G$. This effectively bounds the number of messages forwarded through any edge in the graph. For example with broadcasts, for any edge $\{v,w\}$ of the communication network, the number of nodes $x \in Q$ whose broadcast must be forwarded from $v$ to $w$ is at most $d^s(v,Q) \le \maxdeg$, because the message is forwarded for at most $s+1$ hops. 
The communication tools are also used to obtain the sparsification lemma (\Cref{lem:sparsification}), where the sparsity of $Q_{s}$ in $G^{s}$ is used to run the $(s+1)$th iteration of sparsification on $G^{s+1}$ efficiently. 
In our ruling set application, we use the communication tools to simulate an MIS algorithm on $G^k$. 

\smallskip
\noindent\textbf{Deterministic $k$-ruling set of $G^k$.} Our deterministic sparsification algorithm is used to compute a sparse subset $Q:=Q_{k-1} \subseteq V$, while maintaining constant domination distance to the rest of the graph. After sparsifying the graph, we compute an MIS of $G^k[Q]$. Using our communication tools, we can simulate any MIS algorithm on $G^k[Q]$ in a black-box manner. In general, this approach yields a $(k+1, \beta+k)$-ruling set of $G$, where $\beta$ is the domination distance of $Q$ (see \Cref{theorem:sparsificationToRulingset} for the formal statement). With our sparsification algorithm, the result is a $(k+1,k^2)$-ruling set of $G$, or equivalently a $k$-ruling set of $G^k$.  Messaging between distance-$k$ neighbors in $Q$ can be implemented with a $(k+\log^2 n)$-factor slowdown. 
Combined with the state of the art MIS algorithm \cite{FGG22}, we compute a $(k+1,k^2)$-ruling set in $O(\poly \log n)$-rounds (\Cref{corollary:kksquaredRulingSet}). 
For a constant $k>1$ this improves exponentially upon prior work that required $O(k^2 \cdot n^{1/k})$ rounds \cite{SEW13,KMW18,EM19}.

\paragraph{Randomized MIS of \texorpdfstring{$G^k$}{G\^k}.} For our randomized results, we use the \textit{shattering} framework, where a randomized \textit{base algorithm} is used to solve the problem efficiently on most parts of the graph. The remaining connected components (in $G^k$) are small, with high probability. The small connected components are solved with a different algorithm. Also see the next paragraph on \Cref{thm:MISrepaired} for further details on the shattering framework.  
For the base algorithm, we cannot use a black-box simulation of Ghaffari's MIS algorithm from \cite{ghaffari16_MIS}, as simulation on $G^k$ would be prohibitively expensive. 
Hence, we use the \textit{BeepingMIS} algorithm of \cite{Ghaffari2017} (with a minor but crucial modification), which works in a simple beeping model of communication. 
However, we need to be careful when forwarding the beeps: With cycles in the communication network, nodes may confuse the beep of a neighbor with that of their own, when $k \ge 3$. To avoid this, we equip the beeps with an identifier of the beeping node. Nodes forward an arbitrary subset of at most two beeps, which is enough for any beeping node to distinguish if there is a beeping neighbor or not. 
In the \textit{post-shattering phase}, the remaining unsolved parts of the graph form small connected components in $G^k$, each with $N=O(\Delta^{4k} \cdot \log n)$ nodes, with high probability. To find a solution in the remaining components, we run $O(\log_N n)$ executions of the BeepingMIS algorithm in parallel, which guarantees that at least one of the executions succeeds, with high probability in~$n$. To limit congestion, we assign unique IDs to the nodes in the connected component from $[N]$. This bounds the total communication to $O(\log_N n \cdot \log N) = O(\log n)$ (for simulating one step for all instances). This approach achieves the same runtime as the state of the art algorithm for MIS of $G$ \cite{Gha19}, up to slowdown factors of $k$.

\paragraph{\Cref{thm:MISrepaired},  shattering in $G$, in \LOCAL and \CONGEST.} Since the seminal work of \cite{BEPS16} the shattering technique has become an essential tool in the area. The technique has two phases. In the \emph{pre-shattering phase}  the problem at hand (e.g., the MIS problem) is solved in large parts of the graph via a very efficient randomized process such that with high probability only small components---think of components of polylogarithmic size---remain afterwards. In the \emph{post-shattering phase}  one employs a different algorithm, usually a deterministic algorithm, that \emph{finishes} the small components very efficiently by exploiting their small size. The difficulty here is that the components are  actually of size $\poly\Delta \cdot \log n$, which is not small for large values of $\Delta$. 

\textbf{The high level solution.} Fortunately, components have further beneficial properties: Fix a small component $C$ and some suitable $\beta$ and $\alpha\geq 5$ (the constant $5$ depends on the properties of the pre-shattering phase and works for the MIS algorithms in \cite{BEPS16,ghaffari16_MIS,Gha19}, other problems may need other values). 
Then it is known that any $(\alpha,\beta)$-ruling set $R_C$ of $C$ has at most $O(\log n)$ nodes. 
If we had such a ruling set available, we could exploit its size algorithmically by creating a virtual graph as follows: 
Each node in $R_C$ forms a connected \emph{ball} of nodes around it, such that each node in $C$ joins a unique ball.
Then, one can build a virtual graph $H$ with one vertex for each such ball. 
Two balls are connected in $H$ if the respective balls contain nodes that are adjacent in the original graph $G$. Since $|V(H_C)|=|R_C|=O(\log n)$, and since one round of communication in $H_C$ can be simulated in $G$ in $O(\beta)$ rounds (in the \LOCAL model) one can compute a network decomposition (ND) of $H$ very efficiently, e.g., with the algorithm of \cite{RG20}. 
The details do not matter for the current exposition, but it is known that an ND of $H_C$ implies an ND of $C$ which can be used to compute an MIS on $C$.

\textbf{The challenge.} The challenge is that it is unclear how to compute a ruling set $R_C$. The algorithm would have to run on all small components in parallel and running it on the induced subgraph does not work as the resulting ruling set needs to be $5$-independent in $G$; $5$-independence in $G[C]$ is insufficient. 
It is also critical if a node $v\in C, v\neq R_C$ is  dominated by some node in the ruling set of some other small component $C'\neq C$.  To illustrate the difficulty, observe that a node cannot tell efficiently whether a node in distance $5$ in $G$ is in the same small component as itself or not. Hence, the nodes cannot easily determine whether they can both be contained in the ruling set or not. 

\textbf{The solution.} Inspired by a combination of the different versions of \cite{BEPS16}, we present the following solution. After the pre-shattering phase, we perform a second randomized pre-shattering phase that is run on all small components in parallel (it works w.h.p.\ in $n$ and there are at most $n$ components so we can perform a union-bound over the error probabilities of different components). It splits each component into so called \emph{tiny components}. Then, we prove the following lemma.

\begin{lemma*}[Informal version of \Cref{lem:secondphaseshattering}]
Let $\alpha=5$ and let $\beta$ be an integer and consider a small component $C$. Then, each $(\alpha,\beta)$-ruling set of its tiny components  has at most $O(\log n)$ nodes. Both distances of the ruling set are measured in the graph induced by $C$ and the bound on the size even holds if a node of one tiny component is dominated by a node in another tiny component.     
\end{lemma*}

The benefit of this lemma is that we can  run a ruling set algorithm (in order to dominate all remaining nodes in tiny components) on the graph induced by the small components, making the algorithm fully independent between different small components. Instead of worrying whether a node in distance $5$ in $G$ is within the same component as oneself, one can simply ignore all edges that are not contained in $G[C]$ and work in the graph induced by small components.

\textbf{Proof of the lemma.} While the whole setup is rather complicated and very similar to the one in \cite{BEPS16}, the proof of the lemma is very simple. If we fix one specific ruling set $R_C$ as in the lemma, then the core shattering argument of \cite{BEPS16} states that w.h.p.\ in $n$ the set $R_C$ is of size $O(\log n)$. The intuitive reason for this fact is that each node only remains undecided with a small probability and these events are independent  for nodes that are at least $5$ hops apart. Hence, w.h.p.\ not more than $O(\log n)$ nodes with pairwise distance $5$ can remain undecided.  
Now, we obtain the lemma by performing a union bound over all such ruling sets. The crucial point in this solution is that one can rely on the small size of $C$ (that holds w.h.p.\ after the first pre-shattering phase) to show that there are only $\poly n$ many of these sets making the union bound a feasible approach.\footnote{The informed reader will notice that this is different from (and simpler than) the standard proof for showing that small components emerge after the first pre-shattering phase. In that setting a similar union-bound is not possible. }

\textbf{Previous solutions.} The arXiv version  of \cite{BEPS16} has a very similar two-phase setup, but uses a different (and faulty\footnote{See \Cref{sec:revisedShattering} for details.}) approach to prove a statement that is similar to our lemma. The (correct) journal version of \cite{BEPS16} also uses a similar two-phase structure, but then uses a different and more involved argument to show that a network decomposition of the virtual graph $H$ can be computed efficiently. 
Besides some other technical details, the main difference is that they show that each ball  satisfies certain power graph connectivity requirements if nodes are assigned to rulers \emph{on-the-fly} while computing the ruling set (their proof is insufficient if nodes simply join the ball of the closest ruler). Hence, the proof requires internals of the used ruling set algorithm and it is difficult to formally black-box it, as it is  done in later works. 

\textbf{Our second solution.} Our second solution  shows that using the internals of current state-of-the-art ruling set algorithms one can omit the second pre-shattering phase. While this is simpler than the proof in the journal version of \cite{BEPS16}, we believe that our solution that is based on the two-phase structure is more approachable. However, the version with a single pre-shattering phase is easier to generalize to $G^k$, which is needed for \Cref{thm:MISPower} and \Cref{cor:kkBetaRulingRand}.

\section{Communication tools} \label{sec:comTools}
In this section, we present the communication tools that we need to communicate efficiently on sparse representations of power graphs. 

\begin{lemma}[Learning IDs in distance-$(s+1)$ $Q$-neighborhood]\label{lem:sendingIDs}
	Let $s \ge 1$ and $\maxdeg \ge 1$ be integers. Let $Q \subseteq V$ such that all $v \in V$ have at most $\maxdeg$ distance-$s$ $Q$-neighbors. Each $v \in V$ has a unique $a$-bit identifier ID$(v)$. Given that all $v \in V$ know the set of IDs in $N^{s}(v,Q)$, all $v \in V$ can learn the set of IDs in $N^{s+1}(v,Q)$ in $O(\maxdeg \cdot a / \bandwidth)$ rounds, using $\bandwidth$-bit messages.
	
	Furthermore, suppose that for each $v \in Q$, there is a depth-$s$ BFS tree $T_v$ rooted in $v$, given distributedly. 
	Each tree can be extended to depth $s+1$ in additional $O(\maxdeg \cdot a / \bandwidth)$ rounds.
\end{lemma}
\begin{proof}
	Each node $v \in V$ sends the set of IDs in $N^{s}(v,Q)$ to each of its neighbors $w \in N(v)$ by pipelining. By assumption, $d^{s}(v,X) \le \maxdeg$ for all $v \in V$. Hence, pipelining the IDs takes at most $\maxdeg \cdot a / \bandwidth$ rounds. Each $v \in V$ takes the union of the incoming sets of IDs to form $\cup_{w \in N(v)} N^{s}(w,Q) = N^{s+1}(v,Q)$. 
	
	To extend the BFS trees, each $v \in V$ is added to trees of nodes in $Q$ at distance exactly $s+1$. For each $x \in N^{s+1}(v,Q) \setminus N^s(v,Q)$, we add $v$ as a leaf node to the tree $T_x$. Let $w_x \in N(v)$ be a neighbor who sent $\ID(x)$ to $v$. One such neighbor $w \in N(v)$ is chosen arbitrarily, in case there are many. By assumption, $w$ is already included in $T_x$. Now, $v$ sends a confirmation to $w_x$ including $\ID(x)$. $v$ sets ancestor$(T_x,v) := w_x$, and $w$ adds $v$ to descendants$(T_x, w)$. The additional time complexity is the same as the time taken to send the IDs. 
\end{proof}

\begin{lemma}[Sending messages from $Q \subseteq V$] \label{lem:comms}
	Let $s \ge 1$ and $\maxdeg \ge 1$ be integers. Let $Q \subseteq V$ such that all $v \in V$ have at most $\maxdeg$ distance-$(s-1)$ $Q$-neighbors. For each $v \in Q$, let $T_v$ be a BFS tree rooted in $v$ with depth $s$. Assume that each $v \in Q$ knows the set of IDs in $N^{s}(v,Q)$, as well as $N^{s-1}(w,Q)$ for each $w \in N(v)$. There is a deterministic algorithm for the following tasks:
	\begin{itemize}
		\item Broadcast: Each $v \in Q$ sends an $m$-bit message $\text{msg}_v$ to all $w \in N^{s}(v)$ \\in $O(s + m\maxdeg /\bandwidth)$ rounds.  
		\item $Q$-message: Each $v \in Q$ sends an $m$-bit message $\text{msg}_{v,w}$ to each $w \in N^{s}(v, Q)$ \\in $O(s + (m+a)\maxdeg^2/\bandwidth)$ rounds. 
	\end{itemize}  
	Each node has a $a$-bit identifier that is unique within its $s$-neighborhood. The algorithm uses $\bandwidth$-bit messages, where $\bandwidth \ge~\maxdeg$.  
\end{lemma}
\begin{proof}
	Both algorithms use the BFS tree rooted in each $v \in Q$ for sending messages from $v$ to nodes in $N^{s}(v)$. First, we show that any edge is in at most $P:=2 \maxdeg$ trees. Let $e=\{v,w\} \in E$ be any edge and let $x \in Q$ such that $e$ is part of the tree of $x$. The tree $T_x$ is formed by a breadth-first search from $x$ to depth $s$. Hence, either $x \in N^{s-1}(v,Q)$ or $x \in N^{s-1}(w, Q)$. The distance-$(s-1)$ $Q$-degree of any node is at most $\maxdeg$, so $e$ is in at most $2\maxdeg$ trees.
	
	\textit{Broadcast}: We use the broadcast algorithm from \cite[Section 5]{GGR20}. The following paragraph is a summary of their algorithm. Each tree $T_v$ is only allocated $b':=\bandwidth/P$ bits of communication bandwidth on its edges. The root $v$ breaks its message into $m/b'$ pieces. The message is sent from the root piece by piece, with the $i$th piece being sent in the $i$th round. Other nodes of the tree forward pieces to their descendants, without congestion. We can execute this algorithm in parallel in all trees -- in total, the number of bits per round sent on a single edge is at most $P\cdot b' = \bandwidth$. The full message can be formed without including information about each tree: given a set of pieces, a sender $v$ forms a message to each $w \in N(v)$ including the pieces for trees $T$ where $w \in \text{descendants}(v,T)$. The sender $v$ sorts the pieces by the ID of $\text{root}(T)$ and concatenates them. The receiver $w$ splits the message into $b'$-bit pieces and maps them to the corresponding trees, using knowledge of $\text{ancestor}(w,T)$. The total time required for the last piece of a message to reach all leaves is $O(s + m/b') = O(s + m P/\bandwidth)$ rounds. In our application, $P=2\maxdeg$, so the round complexity of the broadcast algorithm is $O(s + m\maxdeg/\bandwidth)$. 
	
	\textit{Q-message}: Consider any $x \in Q$. It has an $m$-bit message $\text{msg}_{x,y}$ to each $y \in N^{s}(x,Q)$. The total number of messages from $x$ is at most $\Delta \cdot \maxdeg$, because $|N^{s}(x,Q)| = |\cup_{w \in N(v)} N^{s-1}(w,Q)| \le \Delta \cdot \maxdeg$, since $d^{s-1}(w,Q) \le \maxdeg$ by assumption. The algorithm has two steps. In the first step, we distribute the messages evenly between the immediate neighbors of $x$. Each message $\text{msg}_{x,y}$ is packaged as a tuple of the form $(\text{msg}_{x,y}, \ID(y))$. Now, $x$ sends each immediate neighbor $w \in N(x)$ a set of messages $S_{x,w} := \{ (\text{msg}_{x,y}, \ID(y)) : y \in N^{s-1}(w,Q) \}$, using the assumption that $x$ knows $N^{s-1}(w,Q)$ for each $w \in N(x)$. This can send some tuples to multiple neighbors. Still, for any $w$, $|S_{x,w}| \le \maxdeg$, since $d^{s-1}(w,Q) \le \maxdeg$ by assumption. The set of tuples sent from $x$ to $w$ fit in $M:=(m+a)\maxdeg$ bits. The first step takes $O(M/\bandwidth)$ rounds, and it can be done by all nodes in $Q$ in parallel. 
	
	In the second step, we split the tree $T_x$ of each $x \in Q$ into $O(\Delta)$ subtrees $T_{x,w}$, rooted at each $w \in N(x)$. $T_{x,w}$ is the subtree induced by $w$ and \textit{all} descendants of $w$ in $T_x$. For any $x \in Q$, the subtrees of $T_x$ are edge-disjoint and node-disjoint. Earlier, we showed that for any $e \in E$, there are at most $2\maxdeg$ root nodes $x \in Q$ such that $e \in T_x$. By edge-disjointness of the subtrees of $T_x$, for any $e \in E$, there are at most $2\maxdeg$ subtrees $T_{x,w}$ such that $e \in T_{x,w}$. Now, the task in each subtree $T_{x,w}$ is to deliver the messages $\text{msg}_{x,y}$ to each $y \in N^{s-1}(w,Q) \cap V(T_{x,w})$. This is done in parallel in all subtrees of all $x \in Q$ as follows. The full set of tuples $S_{x,w}$ is sent as a single broadcast in $T_{x,w}$, using the previously described broadcast algorithm of \cite{GGR20}. All nodes $v \in T_{x,w}$ have the required knowledge to run the broadcast algorithm: each non-root node knows $\text{ancestor}(v,T_{x,w})$ and $\text{descendants}(v,T_{x,w})$, since they are the same as in $T_x$. The ID of the original root $x$ can be used as the root ID, since the subtrees of $T_x$ are node-disjoint. Running the broadcast algorithm takes $O(s + M\maxdeg/\bandwidth)$ rounds, where $M$ is the size of the broadcasted message and the depth of $T_{x,w}$ is at most $s-1$. The nodes $y \in N^{s-1}(w,Q) \cap V(T_{x,w})$ find $\text{msg}_{x,y}$ from the received set of tuples using their ID. The total complexity is $O\left(\frac{M}{\bandwidth} + s + \frac{M\maxdeg}{\bandwidth}\right) = O\left(s + \frac{(m+a)\maxdeg^2}{\bandwidth}\right)$ 
\end{proof}
The results of \Cref{lem:comms} are tight, which can be seen in \Cref{fig:commTools}. 
\begin{figure}[t]
	\caption{Example graph for \Cref{lem:comms} with nodes in $Q$ in grey. With $s=3$, the total number of messages over $\{v,w\}$ is up to $\maxdeg$ (when nodes in $Q$ broadcast) and $\maxdeg^2/4$ ($Q$-message).}
	\label{fig:commTools}
	\centering
	\includegraphics[width=0.45\textwidth]{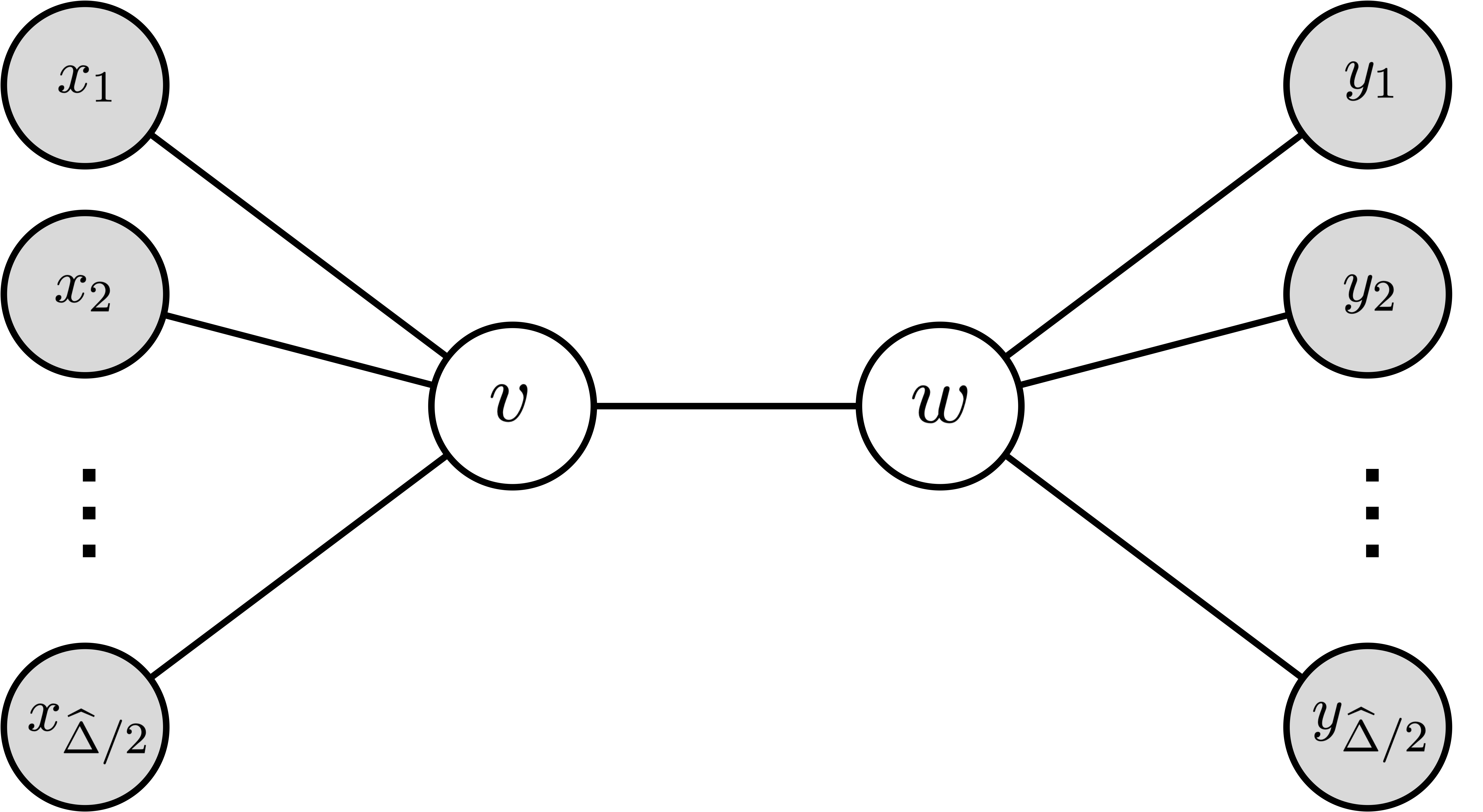}
\end{figure}

\begin{lemma}[Convergecast in spanning tree] \label{lem:convergecast}
	Let $T$ be a spanning BFS tree rooted in some $r \in V$, given distributedly. Suppose that each $v \in V$ has a $m$-bit non-negative integer $x_v$. There is a distributed algorithm that computes the sum $\sum_{v \in V} x_v $ in the root. The algorithm runs in \[O\left(\diam(G) +  (m+\log n)/\bandwidth\right)\] rounds in the \congest model, using $\bandwidth$-bit messages.  
\end{lemma}
\begin{proof}
	The sum $\sum_{v \in V} x_v$ is at most $n\cdot 2^b$. Using the algorithm of \cite[Lemma 5.1]{GGR20} with $O(\log n + b)$-bit messages, $\sum_{v \in V} x_v$ is computed in the root in $O(\diam(G) + (m+\log n)/\bandwidth)$ rounds.
\end{proof}

\noindent A spanning BFS tree for \Cref{lem:convergecast} can be formed by leader election in $O(\diam(G))$ time, by starting a BFS token from each node and forwarding the token of the tree $T$ whose root has the smallest identifier. When using \Cref{lem:convergecast} in network decomposition clusters, we use the Steiner tree of the cluster as the spanning tree. 

\subsection*{Simulating virtual graphs with communication network $G$}

\begin{definition}[Virtual graph $H=(V_H,E_H)$]
	Let $G=(V,E)$ be the communication network. Let $V_H \subseteq V$ and $E_H \subseteq \{\{v,w\} \mid v,w \in V_H, v \neq w \}$.  $H=(V_H, E_H)$ is called a \textit{virtual graph} of $G$.
\end{definition}

\noindent Any subgraph of $G$ is a virtual graph of $G$, but the definition does not require the edges of $H$ to be present in $G$. We may for example define $H$ as $G^s$, for some $s \ge 1$. 

Let $\mathcal{A}$ be any \congest algorithm (deterministic or randomized). Suppose that we want to run $\mathcal{A}$ on $H$, but the available communication network is $G$ instead of $H$. We define generally what it means to \textit{simulate $\mathcal{A}$ on $H$ with communication network $G$}:

\begin{definition}[Simulating $\mathcal{A}$ on virtual graph $H$ with communication network $G$]
	Let $\mathcal{A}$ be any \congest algorithm and let $H=(V_H, E_H)$ be a virtual graph of $G$. Assume that nodes in $V_H$ know their input for $\mathcal{A}$, if there is any, as well as know their degree in $H$, if it is required in the algorithm. $\mathcal{A}$ is simulated round by round. To simulate a round of $\mathcal{A}$:
	\begin{enumerate}
		\item Each $v \in V_H$ form a set of messages $\{\text{msg}_{v,w} : w \in N_H(v) \}$, based on their state in $\mathcal{A}$, in zero rounds. 
		\item Communication between neighbors in $H$ is done in the communication network $G$, using some message passing algorithm. The algorithm may relay to each $v \in V_H$ the full set of incoming messages $\{\text{msg}_{w,v} : w \in N_H(v) \}$, or some aggregation of the messages with sufficient information for $v$ to update its state in $\mathcal{A}$, exactly as if it received the full set of messages. The type of aggregation is specified alongside the simulated algorithm. 
		\item Each $v \in V_H$ updates its state based on its previous state in $\mathcal{A}$ and the received information, without further communication. 
	\end{enumerate}
\end{definition}

\noindent From now on, the communication network is assumed to be $G$, without explicitly specifying it.

The efficiency of simulating $\mathcal{A}$ on $H$ depends heavily on the type of algorithm, the virtual graph $H$ and how $H$ relates to the communication network $G$. In general, edges of $G$ may be forced to relay messages between an arbitrary number of pairs of nodes in $V_H$. Furthermore, in many algorithms, it is not possible to compact the set of relayed messages. We will do algorithm-specific optimizations later. 
For now, we focus on the setting where the virtual graph is $G^s[Q]$, for some $s \ge 2$ and $Q \subseteq V$. 
The next lemma shows that if $Q$ satisfies certain sparseness conditions, any \congest algorithm can be simulated efficiently: 

\begin{lemma}[Simulating $\mathcal{A}$ on $G^s{[Q]}$ with communication network $G$]
	\label{lem:simulation}
	Let $\mathcal{A}$ be any \congest algorithm. Let $s \ge 1$ and $\maxdeg \ge 1$ be integers. Let $Q \subseteq V$ such that all $v \in V$ have at most $\maxdeg$ distance-($s-1$) $Q$-neighbors. For each $v \in Q$, let $T_v$ be a BFS tree rooted in $v$ with depth $s$, known distributedly. Assume that each $v \in Q$ knows the set of IDs in $N^{s}(v,Q)$, as well as $N^{s-1}(w,Q)$ for each $w \in N(v)$. Let $T^{\mathcal{A}}(H)$ be the time complexity of running algorithm $\mathcal{A}$ on an input graph $H$. There is a deterministic \congest algorithm that simulates $\mathcal{A}$ on $G^s[Q]$ with communication network $G$ in $O((s + \maxdeg^2) \cdot T^{\mathcal{A}}(G^s[Q]))$ rounds. 
\end{lemma}
\begin{proof}
	We use the \textit{$Q$-message} algorithm of \Cref{lem:comms}, which allows each $v \in Q$ to send a message of $O(\log n)$ bits to each $w \in N^s(v,Q)$ in $O(s + \maxdeg^2)$ rounds. These are exactly the neighbors of $v$ in $G^s[Q]$.  Note that the required knowledge to run $Q$-message is available by assumption. Simulating each round takes $O(s + \maxdeg^2)$ rounds, so the total time complexity is $O((s + \maxdeg^2) \cdot T^{\mathcal{A}}(G^s[Q]))$.
\end{proof}

\section{Sparsification of Power Graphs}
\label{sec:detSparsification}

The goal of this section is to prove \Cref{lem:sparsification}, our main sparsification result for power graphs. First, in the next two sections, we prove the following deterministic sparsification result.  
\Cref{lem:Gsparsification} is standalone and does not refer to power graphs. In \Cref{sec:sparsGk}, we use it iteratively  to obtain a sequence $V\supseteq Q_1\supseteq Q_2\supseteq \ldots \supseteq Q_k$ where $Q_i$ is sparse in the power graph~$G^i$. 

\newpage
\begin{lemma}[Deterministic Sparsification] \label{lem:Gsparsification}
    Let $A \subseteq V$ be a set of initially active vertices. There is a deterministic distributed algorithm that finds a set of vertices $Q \subseteq A$ such that for all $v \in V$, 
    \begin{itemize}
        \item (bounded $Q$-degree): $d(v,Q) \le 72\log n = O(\log n)$
        \item (domination): $\dist_{G}(v, Q) \le 2 + \dist_{G}(v, A)$
    \end{itemize}
    Let $\maxactivedeg \ge \max_{v \in V} d(v,A)$ be an input parameter given to all nodes, which is at least the maximum number of active neighbors. 
    The algorithm runs in
    $O(\diam(G) \cdot \log^2 n \cdot \log \maxactivedeg)$ rounds in \congest.
\end{lemma}

\subsection{Randomized Sparsification via Sampling}
\label{sec:randomized sampling}

We start with a \textit{randomized} sparsification algorithm to find a certain sparse set of vertices, satisfying the properties of \Cref{lem:Gsparsification}, with high probability. See Algorithm \ref{algorithm:samplingrandomized}. The algorithm consists of $r := \lfloor \log \maxactivedeg - \log \log n\rfloor-5$ stages. Let $H_1 := A$, where $A \subseteq V$ is any set of initially active nodes. For $1 \le i \le r$, $H_i$ is a set of \textit{active nodes} in the respective stage.  
In each stage $i$, we sample a set $M_i \subseteq H_i$. Each node $v \in H_i$ is included in $M_i$ with probability $\frac{24\cdot 2^i \cdot \log n}{\maxactivedeg}$, where the decisions of the nodes are $8 \log n$-wise independent. We deactivate all sampled nodes, as well as active nodes that have a sampled node within 2 hops in $G$. 
This is done by sending a flag from each sampled node, propagated for two hops, where multiple incoming flags can be forwarded as one. Once nodes are deactivated, they stay inactive forever. Let $H_{i+1} = H_i \setminus (M_i \cup N^2(M_i))$ 
be the remaining active nodes. After $r$ stages, the algorithm returns $Q := \cup_{i=1}^{r+1} M_i$, consisting of the sampled sets $M_1, \dots, M_r$ and the remaining active nodes $M_{r+1}:=H_{r+1}$. 

\begin{algorithm}[h]
    \DontPrintSemicolon
    \SetAlgoLined
    \SetKwIF{If}{ElseIf}{Else}{if}{:}{else if}{else}{}
    \SetKwFor{For}{for}{:}{}
    \SetKwFor{ForEach}{foreach}{:}{}
    \KwIn{$\maxactivedeg \ge \max_{v \in V} d(v,A)$,\\ \hspace{2mm} Each $v \in V$ knows if it is in a set of initially active nodes $A \subseteq V$.}
    Set $r := \lfloor \log \maxactivedeg - \log \log n\rfloor-5$ and $H_1 := A$\;
    \For{\upshape{stage} $i= 1, \dots, r$}{
        $M_i := \emptyset$\;
        \ForEach{$v \in H_i$ in parallel}{
            Join $M_i$ with probability $\frac{24\cdot 2^i \cdot \log n}{\maxactivedeg}$\;
        }
        $H_{i+1} := H_i \setminus (M_i \cup N^2(M_i))$\;
    }
    $M_{r+1} := H_{r+1}$\;
    \Return{$Q:=\cup_{i=1}^{r+1} M_i$}
    \caption{Randomized sparsification}
    \label{algorithm:samplingrandomized}
\end{algorithm}

\vspace{-3mm}
\begin{definition}[Active degree]
    For $v \in V$, its \textit{active degree} in stage $i$ is defined as $d(v,H_i)$.
    We say that $v \in V$ has a \textit{high active degree in stage~$i$} if $d(v,H_i) \ge \maxactivedeg / 2^i$. 
\end{definition}

\noindent The active degree of a node changes throughout the algorithm. Note that inactive nodes are never reactivated, so we have $d(v, H_1) \ge d(v,H_2) \ge \dots \ge d(v,H_{r+1})$ for all $v \in V$. Also, whether $v$ itself is active or inactive does not affect its active degree directly. Next, we prove the properties that hold after each stage with high probability.

\begin{theorem}[Theorem 2.5 in \cite{SSS95}]
    \label{theorem:chernoff}
    Let $X$ be the sum of $p$-wise independent $[0, \lambda]$-valued random variables with expectation $\mu = \ev(X)$ and let $\delta \le 1$. 
    Then $\pr(|X - \mu| \ge \delta \mu) \le e^{-\lfloor \min\{ p/2, \delta^2 \mu / (3\lambda) \} \rfloor }$. 
\end{theorem}

\begin{lemma}[Stage $i$ of randomized sparsification] \label{lem:randstage}
    Fix some stage $1 \le i \le r$ and let $H_i \subseteq V$ be a set of active nodes, such that all nodes $v \in V$ have at most $\maxactivedeg/2^{i-1}$ neighbors in $H_i$. The $i$th stage of Algorithm \ref{algorithm:samplingrandomized} returns $M_i$ and $H_{i+1}$ such that for all $v \in V$:
    \begin{itemize}
        \item [(i)] $d(v,M_i) \le 72\log n$, with probability at least $1-1/n^3$. 
        \item [(ii)] \textbf{if} $d(v,H_i) \ge \maxactivedeg / 2^i$, \textbf{then} $v \in M_i \cup N(M_i)$, with probability at least $1-1/n^3$. 
        \item [(iii)] $d(v,H_{i+1}) < \maxactivedeg /2^{i}$, with probability at least $1-1 / n^3$.   
    \end{itemize}
    One stage requires 2 rounds. The claims hold if the random choices are $8\log n$-wise independent.  
    \end{lemma}
\begin{proof}
    For an active node $w \in H_i$, let $X_w$ be an indicator variable for the event that $w$ is sampled. 

\noindent     \textbf{Proof of (i).} Let $v \in V$ be any vertex. By assumption, $v$ has at most $\maxactivedeg /2^{i-1}$ active neighbors.  Let $W$ be a set of fake vertices, added to the set of active neighbors of $v$ such that $v$ has exactly $\maxactivedeg /2^{i-1}$ active neighbors. Define the indicator variable $X_{w}$ for fake vertices $w \in W$ similarly as for real active nodes. Clearly $v$ has at most $72\log n$ real sampled neighbors whenever at most $72 \log n$ vertices in $N(v,H_i) \cup W$ are sampled. The probability of this event can be lower bounded using \Cref{theorem:chernoff}. Let $X = \sum_{w \in N(v,H_i) \cup W} X_w$ be a sum of the indicator variables, with expected value $\mu := \ev[X] = \frac{\Delta_A}{2^{i-1}} \frac{24\cdot 2^i \cdot \log n}{\maxactivedeg} = 48\log n$. Let $\delta = \frac{1}{2}$ in \Cref{theorem:chernoff}:
    \begin{align*}
        \pr(X \ge 72\log n) 
        &\le \pr(|X - 48\log n| \ge \frac{1}{2}48\log n) \\
        %&\le \pr(|X - 48\log n| \ge (48/2)\log n) \\
        &\le e^{-\lfloor \min\{ 8\log n/2, 48\log n / 12 \} \rfloor } 
        = e^{-\lfloor 4\log n \rfloor } \le n^{-3} \ . %\le \frac{1}{n^3} \ .
    \end{align*}
    Hence, $v$ has at most $72\log n$ neighbors in $M_i$ with probability at least $1-1/n^3$.

\noindent \textbf{Proof of (ii).} Let $v \in V$ be a node with at least $\maxactivedeg / 2^i$ active neighbors. We will compute a lower bound for the probability that $v$ is adjacent to a sampled node. Fix any subset $S \subseteq N(v,H_i)$ of active neighbors such that $|S| = \maxactivedeg /2^i$. Let $X = \sum_{w \in S} X_w$ be a sum of indicator variables, with expected value $\mu := \ev[X] =  \frac{\maxactivedeg}{2^i} \frac{24 \cdot 2^i \cdot \log n}{\maxactivedeg}=24\log n$. At least one neighbor is sampled when $X > 0$. The probability of $X=0$ can be upper bounded using \Cref{theorem:chernoff}. Let $\delta = \frac{3}{4}$:
    \begin{align*}
        \pr(X=0)
        %\le \pr(|X - \mu| \ge \frac{3}{4} \mu) 
        &\le \pr(|X - \mu| \ge (3/4) \cdot \mu) \\
        &\le e^{-\lfloor \min\{ 8\log n/2, 72\log n / 16 \} \rfloor }
        %= e^{-\lfloor 4\log n \rfloor} \le \frac{1}{n^3} \ .
        = e^{-\lfloor 4\log n \rfloor} \le n^{-3} \ .
    \end{align*}
    Hence, at least one active neighbor of $v$ is sampled with probability at least $1- 1/n^3$.

    \textbf{Proof of (iii).} First, consider $v \in V$ with $d(v,H_i) < \maxactivedeg / 2^i$. Since nodes are never reactivated, the active degree of $v$ will be less than $\maxactivedeg / 2^i$ in stage $i+1$ as well. Now, let $v \in V$ be any node with high active degree in stage $i$, that is, $d(v,H_i) \ge \maxactivedeg / 2^i$. The algorithm deactivates the distance-2 neighborhood of sampled nodes. Hence, all of $N(v,H_i)$ is deactivated, whenever some $w \in N(v,H_i)$ is sampled. By (ii), this happens with probability at least $1-\frac{1}{n^3}$. In this case, the active degree $d(v,H_j)$ of $v$ is zero for all remaining stages $i < j \le r$.
\end{proof}

We are now ready to prove a randomized version of \Cref{lem:Gsparsification}, running in $O(\log \Delta)$ rounds: The $r=O(\log \maxactivedeg)=O(\log \Delta)$ stages of Algorithm \ref{algorithm:samplingrandomized} produce $Q := \cup_{i=1}^{r+1} M_i$, consisting of the sampled sets $M_1, \dots, M_r$ and the remaining active nodes $M_{r+1}:=H_{r+1}$. \Cref{lem:randstage}~(i) states that any node has $O(\log n)$ neighbors in any $M_i$, with high probability, given that maximum active degree decreased to $\maxactivedeg/2^{i-1}$ in the previous stage. The maximum active degree decreases for any node, with high probability by \Cref{lem:randstage}~(iii). 
The sets $M_i$ and $M_j$ are at distance at least $2$ from each other for any $i \neq j$, so any node does not have neighbors in more than one $M_i$. Taking a union bound over all stages, we get that any node has at most $O(\log n)$ neighbors in the whole $Q=\cup_{i=1}^{r+1} M_i$. This also includes the remaining active nodes $M_{r+1}:=H_{r+1}$. 
Next, we construct a deterministic sparsification algorithm by derandomizing the random choices of the active nodes in each stage.

\subsection{Deterministic Sparsification via Derandomization}
\label{sec:derandomizing sampling}
We construct a deterministic sparsification algorithm to prove \Cref{lem:Gsparsification}. The deterministic algorithm, Algorithm \ref{algorithm:DetSparsification}, is referred to as \detsparsification. 
The structure of \detsparsification is the same as the randomized sparsification algorithm (Algorithm \ref{algorithm:samplingrandomized}). As before, the input is a set of active nodes $A \subseteq V$ and a maximum active degree parameter $\maxactivedeg \ge \max_{v \in V} d(v,A)$, where each $v \in V$ knows whether $v \in A$, and the value of $\maxactivedeg$. There are $r = \lfloor \log \maxactivedeg - \log \log n \rfloor - 5$ stages. For $1 \le i \le r$, $H_i$ is the set of active nodes in the $i$th stage, where $H_1 := A$ and $H_i \subseteq H_{i-1}$. Fix some stage $1 \le i \le r$. In the $i$th stage, \detsparsification selects a set $M_i \subseteq H_i$ by derandomizing the sampling procedure of the $i$th stage of Algorithm \ref{algorithm:samplingrandomized}, in $O(\diam(G)\cdot \log^2 n)$ rounds.
This is described in the next lemma. The rest of the algorithm works exactly like the randomized version. The remaining active nodes are $H_{i+1} = H_i \setminus (M_i \cup N^2(M_i))$. 
After $r$ stages, the algorithm returns $Q := \cup_{i=1}^{r+1} M_i$, consisting of the sampled sets $M_1, \dots, M_r$ and the remaining active nodes $M_{r+1}:=H_{r+1}$.

\begin{algorithm}[h]
    \DontPrintSemicolon
    \SetAlgoLined
    \SetKwIF{If}{ElseIf}{Else}{if}{:}{else if}{else}{}
    \SetKwFor{For}{for}{:}{}
    \SetKwFor{ForEach}{foreach}{:}{}
    \KwIn{$\maxactivedeg \ge \max_{v \in V} d(v,A)$,\\ \hspace{2mm} Each $v \in V$ knows if it's in a set of initially active nodes $A \subseteq V$.}
    $r := \lfloor \log \maxactivedeg - \log \log n\rfloor-5$\;
    $H_1 := A$\;
    \For{\upshape{stage} $i= 1, \dots, r$}{
        Find $M_i \subseteq H_i$ using \Cref{lem:conditionalEVs} s.t. neither $\Psi_v$ nor $\Phi_v$ (\ref{eq:Phi}, \ref{eq:Psi}) occur for any $v \in V$\;
        $H_{i+1} := H_i \setminus (M_i \cup N^2_G(M_i))$\;
        \ForEach{$v \in V$ in parallel}{
            Inform each $w\in N(v)$ whether $v \in H_{i+1}$\;
        }
    }
    $M_{r+1} := H_{r+1}$\;
    \Return{$Q:=\cup_{i=1}^{r+1} M_i$}
    \caption{DetSparsification}
    \label{algorithm:DetSparsification}
\end{algorithm}

\begin{restatable}[Derandomizing $i$th stage]{lemma}{lemDerandithStage} \label{lem:derandstage}
    The round complexity of a stage is $O(\diam(G) \cdot \log^2 n)$.
    Fix a stage $1 \le i \le r$ and let $H_i \subseteq V$ be a set of active nodes, such that all nodes $v \in V$ have at most $\maxactivedeg/2^{i-1}$ neighbors in $H_i$. The $i$th stage of \detsparsification returns $M_i$ and $H_{i+1}$ such that for all $v \in V$:
    \begin{itemize}
        \item [(i)]  $d(v,M_i) \le 72\log n$, ,
        \item [(ii)] \textbf{if} $d(v,H_i) \ge \maxactivedeg / 2^i$~, \textbf{then} $v \in M_i \cup N(M_i)$~,
        \item [(iii)] $d(v,H_{i+1}) < \maxactivedeg /2^{i}$~.
    \end{itemize}    
\end{restatable}
\begin{proof} 
    We derandomize one stage $1 \le i \le r$ of the randomized sampling algorithm (see \Cref{sec:randomized sampling} and Algorithm \ref{algorithm:samplingrandomized}). Each active node $v \in H_i$ is sampled to $M_i$ with probability $\frac{24 \cdot 2^i \log n}{\maxactivedeg}$. After our derandomization, there are two types of events that we want to guarantee for all $v \in V$. 
    First, all $v \in V$ have a maximum of $72 \log n$ neighbors in $M_i$. Secondly, nodes with high active degree in stage $i$ are have at least one sampled neighbor, or are sampled themselves. The purpose of this is to guarantee that the maximum active degree in the graph decreases, because the neighborhood of any high active degree node is deactivated. 
    
    Define indicator variables for the complements of these events. Let $Z = \{v \in V: d(v,H_i) \ge \maxactivedeg / 2^i \}$ be the set of nodes with high active degree in stage~$i$. For each $v \in Z$, let  $\Phi_v$ be an indicator variable for the event that $v$ is not sampled and does not have a sampled neighbor:
    \begin{equation}\label{eq:Phi}
        \Phi_v := \begin{cases} 1 & v \not\in M_i \cup N(M_i) \\ 0 &\text{else} \end{cases}
    \end{equation}
    For convenience, we also define $\Phi_v$ for $v \in V \setminus Z$ as $\Phi_v=0$. By  \Cref{lem:randstage}~(ii), 
    for all $v \in Z$, the probability that $v$ does not have any of its neighbors sampled or is not sampled itself is at most $1/n^3$. Hence, for all $v \in V$, 
    $$\pr(\Phi_v = 1) \le \frac{1}{n^3}$$
    Next, for all $v \in V$ (including $Z$), let $\Psi_v$ be an indicator variable for the event that $v$ has more than $72 \log n$ sampled neighbors, 
    \begin{equation}\label{eq:Psi}
        \Psi_v := \begin{cases} 1 & d(v, M_i) > 72\log n \\ 0 &\text{else} \end{cases}
    \end{equation}
    By assumption, the maximum active degree at the start of stage $i$ is at most $\maxactivedeg / 2^{i-1}$. With this assumption,  \Cref{lem:randstage}~(i) 
    states that, for all $v \in V$,
    $$\pr(\Psi_v = 1) \le \frac{1}{n^3}$$
    
    For each $v \in H_i$, let $X_v$ be an indicator variable for the event that is sampled to $M_i$. We use \Cref{lem:conditionalEVs} to fix the decisions of the active nodes: 
    
    \begin{claim} \label{lem:conditionalEVs}
        Fix some stage $1 \le i \le r$. Let $H_i \subseteq V$ and let $\maxactivedeg = O(n)$ be an integer parameter. Let $T$ be a spanning BFS tree with root $r \in V$, given distributedly. Let $\{X_v\}_{v \in H_i}$ be a set of $8\log n$-wise independent binary variables with $\pr(X_v = 1) = \frac{24\cdot 2^i \cdot \log n}{\maxactivedeg}$. For each $v \in V$, let $\Phi_v$ and $\Psi_v$ be events, each with probability at most $\frac{1}{n^3}$. Let $\text{\upshape vbl}(v) \subseteq H_i$ be a (unique minimal) set of nodes such that the values of $X_w, w \in \text{\upshape vbl}
        (v)$  determine $\Phi_v$ and $\Psi_v$. Assume that each $v \in V$ knows the IDs in $\text{\upshape vbl}(v)$. There is a deterministic \congest algorithm that selects the values of $X_v$ for each $v \in H_i$ such that none of the events occur. All communication is done using convergecast operations in $T$. The round complexity is $O(\diam(G) \cdot \log^2 n)$. 
    \end{claim}
    \begin{proof}
        We simulate the random choices of the active nodes by choosing a hash function $h: [n] \ra [\maxactivedeg]$ uniformly at random from a family $\mathcal{H}$ of $8\log n$-wise independent hash functions. Setting $X_v=1$ with probability $\frac{24\cdot 2^i \cdot \log n}{\maxactivedeg}$ is equivalent to setting $X_v=1$ if $h(v) \le 24\cdot 2^i \cdot \log n$. By  \Cref{lem:hashfamily}, choosing a random function from $\mathcal{H}$ takes $\gamma := 8\log^2 n$ random bits. 
        
        We call a hash function $h \in \mathcal{H}$ \textit{good} for a variable ($\Phi_v$ or $\Psi_v$), if it makes its value to be zero, that is, the corresponding event does not occur. The probability of each event is at most $1 /n^3$ and the total number of events is $2n$. Hence,
        $$
        \ev\bigg(\sum_{v \in V} \Phi_v + \Psi_v \bigg) \le 2n \cdot \frac{1}{n^3} = \frac{2}{n^2}
        $$
        For large enough $n$, the expected number of unwanted events is less than 1. Hence, by law of total expectation, there exists a hash function such that none of the unwanted events occur. 
        
        Let $B=(B_1, \dots, B_\gamma)$ be the random bits that are used to select a hash function. We fix the values of the random bits one by one, using the method of conditional expectations. Now, we describe how to select the value of the $j$th random bit $B_j$, given that the values $B_1=b_1, \dots, B_{j-1}=b_{j-1}$ have already been fixed. Consider any node $v \in V$. $v$ computes two conditional expectations:
        \begin{align*}
            \alpha_{v, 0} &=\ev[\Phi_v + \Psi_v \mid B_1=b_1, \dots, B_{j-1}=b_{j-1}, B_j = 0]\\
            \alpha_{v, 1} &=\ev[\Phi_v + \Psi_v \mid B_1=b_1, \dots, B_{j-1}=b_{j-1}, B_j = 1]
        \end{align*}
        The expectation $\alpha_{v, b}, b \in \{0,1\}$ is the average outcome of $\Phi_v + \Psi_v$ over the hash functions with the prefix $b_1, \dots, b_{j-1}, b$. Given any hash function $h \in \mathcal{H}$, $v$ can check whether $h$ is good for its variables. To do this, $v$ determines the values of $X_w, w \in \text{vbl}(v)$: the value of $X_w$ is 1 iff $h(\ID(w)) \le 24\cdot 2^i \cdot \log n$. This takes zero rounds, since these IDs in $\text{vbl}(v)$ are known to $v$ by assumption. Hence, computing the conditional expectations takes zero rounds. Next, we compute the sums $\sum_{v \in V} \alpha_{v,b}, b \in \{0,1\}$ in the root $r$ of $T$. This is done by running two instances of \Cref{lem:convergecast} in parallel, one for each bit. The round complexity of computing the sums is $O(\diam(G))$ rounds\footnote{When fixing the $i$th bit, $\alpha_{v,b}$ is a fraction of the form $x_{v,b}/2^{\gamma - i + 1}$, where $x_{v,b} \in [2\cdot 2^\gamma]$ is the total number of times $\Psi_v$ and $\Phi_v$ occurs over the at most $2^\gamma = n^2$ hash functions with prefix $b_1, \dots, b_{i-1}, b$. The denominator is the same for all nodes, so we can instead send $x_{v,b}$, taking $O(\log n)$ bits of space.}.
        The root chooses $b_j := \argmin_{b \in \{0,1\}} \sum_{v \in V} \alpha_{v,b}$ and sends the chosen value to all the nodes. 
        
        The choice of hash function is deterministic, once all $\gamma$ bits have been fixed. The outcomes of $\Phi_v$ and $\Psi_v$ are now determined by the fixed seed for all nodes $v \in V$. The number of unwanted events that occur can now be written as $\sum_{v \in V} \ev[\Phi_v + \Psi_v \mid B_1=b_1, \dots, B_{\gamma}=b_{\gamma}]$. By the law of total expectation, we have
        \begin{align*}
            \sum_{v \in V} \ev[\Phi_v + \Psi_v \mid B_1=b_1, \dots, B_{\gamma}=b_{\gamma}] 
            &\le \sum_{v \in V} \ev[\Phi_v + \Psi_v] 
            = \ev \Big[\sum_{v \in V} \Phi_v + \Psi_v\Big] < 1
        \end{align*}
        Hence none of the unwanted events occur. 
        
        The total time complexity is $O(\diam(G) \cdot \log^2 n)$. This consists of fixing the $\gamma = O(\log^2 n)$ bits one by one. For each bit, nodes can compute the conditional expectations $\alpha_{v,0}, \alpha_{v,1}$ in zero rounds.
        Gathering the sums of the conditional expectations to some root node takes $O(\diam(G))$ rounds. The root node sends the fixed value of the bit to the rest of the graph in $O(\diam(G))$ rounds. Once all $\gamma$ bits have been fixed, active nodes extracts their decision from the hash function in zero rounds.
        \renewcommand{\qed}{\ensuremath{\hfill\blacksquare}}
    \end{proof}
    \renewcommand{\qed}{\hfill \ensuremath{\Box}}
    For all $v \in V$, the events $\Phi_v$ and $\Psi_v$ depend only on the decisions of active neighbors. Each $v \in V$ can learn the IDs in $N(v,H)$ in one round. The probabilities of the events are at most $\frac{1}{n^3}$. \Cref{lem:conditionalEVs} provides a deterministic algorithm to fix the values of $X_v, v \in H$ such that neither $\Phi_v$ nor $\Psi_v$ occur for any $v \in V$. Let $M_i = \{v \in H_i: X_v = 1\}$ be the set of sampled nodes. 
    
    \textbf{Proof of (i).} All nodes have $d(v,M_i) \le 72\log n$, since $\Psi_v=0$ for all $v \in V$. \\
    \textbf{Proof of (ii).} Any $v \in V$ with at least $\maxactivedeg/2^i$ active neighbors is in $M_i \cup N(M_i)$, since $\Phi_v=0$ for all $v \in Z$.\\
    \textbf{Proof of (iii).} The sampled nodes and their distance-2 neighborhood are removed from the set of active nodes, just like in the randomized algorithm. Let $H_{i+1} \subseteq H_i$ be the set of remaining active nodes.  The neighborhood of any $v \in V$ with $d(v, H_i) \ge \maxactivedeg / 2^{i}$ is fully deactivated because at least one of the active nodes in $N(v,H_i)$ was sampled. Hence, $d(v,H_{i+1}) < \maxactivedeg / 2^{i}$ for all $v \in V$.
    
    The total time complexity is $O(\diam(G) \cdot \log^2 n)$: the derandomization takes $O(\diam(G) \cdot \log^2 n)$ by \Cref{lem:conditionalEVs}, and deactivating the distance-2 neighborhood of sampled nodes takes 2 rounds. 
\end{proof}

\begin{proof}[Proof of  \Cref{lem:Gsparsification}] 
Run the \detsparsification algorithm. The number of stages is $r=\lfloor \log \maxactivedeg - \log \log n\rfloor - 5$. Note that we can assume that $\maxactivedeg \ge 2^5 \log n$,\footnote{If $\maxactivedeg < 2^5\log n$, then $\max_{v \in V} d(v,A) < 2^5\log n$ by definition of $\maxactivedeg$. Return the initial set of active nodes $A$, which now satisfies both conditions of  \Cref{lem:Gsparsification}.} which makes $r$ non-negative. \detsparsification returns $Q = \cup_{i=1}^{r+1} M_i \subseteq A$, where $M_1, \dots, M_r$ are the sets selected in stages $1 \le i \le r$ and $M_{r+1}=H_{r+1}$ is the set of remaining active nodes. 

\noindent \textbf{Domination property:} We start by showing the domination property of $Q$, that is, $\forall v \in V: \dist_G(v, Q) \le 2 + \dist_G(v, A)$. Let $v \in V$ be any node and let $w \in A$ be any initially active node that is closest to $v$, i.e., $\dist_G(v,w) = \dist_G(v, A)$. There are three possible outcomes for $w$ in the algorithm:
(1) $w \in Q$. Now $\dist_G(v,Q) = \dist_G(v,A)$. (2) $w$ is deactivated in some stage $i$. By definition, there exists some $w' \in N^2(w)$ such that $w' \in M_i$. Hence $\dist_G(v, Q) \le \dist_G(v,A) + 2$. (3) $w$ is never deactivated. Hence $w \in H_{r+1}$. The remaining active nodes $M_{r+1}:=H_{r+1}$ are included in $Q$. Hence $\dist_G(v, Q)=\dist_G(v, A)$. In all cases, $\dist_G(v,Q) \le 2 + \dist_G(v,A)$. 

\noindent \textbf{Sparsity:} We prove the claim about bounded $Q$-degree, $\forall v \in V: d(v, Q) \le 72\log n$. For any stage $1 \le i \le r$, \Cref{lem:derandstage} states that any $v \in V$ has at most $72 \log n$ neighbors in $M_i$, assuming that the maximum active degree at the start of stage $i$ is at most $\maxactivedeg /2^{i-1}$. The assumption  holds for $i=1$, since $\maxactivedeg \ge \max_{v \in V} d(v,A)$ by definition. For stages $i=2, \dots, r$, this is guaranteed by  \Cref{lem:derandstage}~(iii), which states that the maximum active degree at the end of stage $i-1$ is at most $\maxactivedeg / 2^{i-1}$. Finally, $M_{r+1} = H_{r+1}$ consists of the nodes who are still active  after $r$ stages of sampling. By \Cref{lem:derandstage} (iii), the maximum active degree after stage $r=\lfloor \log \maxactivedeg - \log \log n\rfloor-5$ is 
$$
    \frac{\maxactivedeg}{2^{\lfloor \log \maxactivedeg - \log \log n\rfloor-5}} 
    \le \frac{\maxactivedeg}{2^{ \log \maxactivedeg - \log \log n-6}} = 64\log n
$$ 
Hence, any $v \in V$ has at most $72 \log n$ neighbors in $M_i$, for any $1 \le i \le r+1$. Finally, nodes in $M_i$ and $M_j$, $i \neq j \in 1, \dots, r$ do not have any common neighbors, because the distance-2 neighborhood of sampled nodes is deactivated. Nodes in $M_{r+1}$ and $M_i, 1 \le i \le r$ do not have any common neighbors for the same reason. We conclude that each node has at most  $72\log n = O(\log n)$ neighbors in $Q$. 

The runtime of \detsparsification consists of $\lfloor \log \maxactivedeg - \log \log n \rfloor -5$ stages, each taking \linebreak$O(\diam(G) \cdot \log^2 n)$ rounds by \Cref{lem:derandstage}. The total runtime is $O(\diam(G) \cdot \log \Delta \cdot \log^2 n)$. 
\end{proof}

\subsection{Sparsification in Power Graphs}
\label{sec:sparsGk}

In this section, we present a deterministic sparsification algorithm for $G^k$, proving \Cref{lem:sparsification}. 

\paragraph{Construction outline} Let $Q_0 \subseteq V$ be any set of initially active nodes. Our algorithm consists of $k$ iterations of the deterministic sparsification algorithm for $G$ (\detsparsification, Algorithm~\ref{algorithm:DetSparsification}).
The first iteration runs \detsparsification in $G$ with the active nodes initialized as $Q_0$. 
In iteration $2 \le s \le k$, \detsparsification is simulated on a power graph $G^s$ with communication network $G$. 
The set of active nodes for the $s$th iteration is initialized as $Q_{s-1} \subseteq V$, where $Q_{s-1}$ is the result of the previous iteration. The result of the $s$th iteration is some $Q_s \subseteq Q_{s-1}$. After $k$ iterations, the algorithm outputs $Q_k$. 

A full algorithm description, as well as details on how to simulate \detsparsification on $G^s$ will be given later. The efficiency of the simulation in each iteration $s$ relies on properties of the set of active nodes $Q_{s-1}$, returned by \detsparsification in the previous iteration. Concretely, we maintain the following invariants:

\paragraph{Algorithm invariants} The sparsification algorithm for $G^k$ finds a sequence of sets of nodes $Q_0 \supseteq Q_1 \supseteq \dots \supseteq Q_{k}$. The following invariants hold deterministically after all iterations $1 \le s \le k$:
\begin{enumerate}
	\item [\textbf{I1.1}] (bounded distance-$s$ $Q_s$-degree) $\forall v \in V: d^{s}(v,Q_s) \le 72\log n$
	\item [\textbf{I1.2}] (bounded distance-$(s+1)$ $Q_s$-degree) $\forall v \in V: d^{s+1}(v,Q_s) \le 72 \Delta \log n$
	\item [\textbf{I2}] (domination) $\forall v \in V: \dist_G(v,Q_s) \le \sum_{j=1}^s 2j + \dist_G(v, Q_0) = s^2 + s + \dist_G(v, Q_0)$
	\item [\textbf{I3}] (knowledge of distance-($s+1$) $Q_s$-neighborhood) All $v \in V$ know the set of IDs in $N^{s+1}(v,Q_s)$. Moreover, for each $x \in Q_s$, there is a BFS tree of depth $s+1$ rooted in $x$. Each \linebreak $v \in V$ knows, for each $T_x$ it belongs to, the ID of the root $x$, ancestor$(T_x, v) \in N(v)$ and descendants$(T_x,v) \subseteq~N(v)$.
\end{enumerate}

\paragraph{Algorithm description} 
Let $k \ge 1$ be an integer and $Q_0 \subseteq V$ be any set. The algorithm consists of $k$ iterations of \detsparsification  (Algorithm~\ref{algorithm:DetSparsification}), where the $s$th iteration is simulated on the power graph $G^s$. 

In the first iteration, we run \detsparsification on $G$. The set of active nodes is initialized as $Q_0$, and the maximum active degree parameter $\maxactivedeg$ is set to $\maxactivedeg^{(1)} := \Delta$. Let $Q_1$ be the set of nodes returned by \detsparsification in the first iteration. To prepare for simulation on $G^2$ in the next iteration, we must maintain invariant (I3). To do this, each $v \in V$ sends the set of IDs in $N(v,Q_1)$ to all neighbors $w \in N(v)$ in $O(\log n)$ rounds, as described in \Cref{lem:sendingIDs}.
As a result, each $v \in V$ knows the set of IDs in $N^2(v, Q_1)$, and for each $x \in Q_1$, there is a BFS tree of depth 2 with root $x$.

Fix some iteration $2 \le s \le k$. Assume that the invariants hold for iteration $s-1$. We simulate \detsparsification on the power graph $G^s$ with communication network $G$. Active nodes are initialized as $Q_{s-1}$, where $Q_{s-1}$ is the result of the previous iteration. The maximum active degree parameter $\maxactivedeg$ is set as $\maxactivedeg^{(s)} := 72\Delta \log n$ (for all $2 \le s \le k$). 
The details of the simulation are given in \Cref{lem:simDetSparsification}. Let $Q_{s} \subseteq Q_{s-1}$ be the set of nodes returned by \detsparsification in this iteration. At the end of the iteration, we must again maintain (I3). Each $v \in V$ sends the set of IDs in $N^s(v,Q_s)$ to each of its neighbors $w \in N(v)$ in $O(\log n)$ rounds, as described in \Cref{lem:sendingIDs}. Now all nodes $v \in V$ know the set of IDs in $\cup_{w \in N(v)} N^s(w,Q_s) = N^{s+1}(v,Q_s)$, and the BFS trees of nodes in $Q_s$ are extended to depth $s+1$. The trees of nodes in $Q_{s-1} \setminus Q_s$ are not used anymore. This concludes one iteration of the algorithm.

After $k$ iterations, the algorithm outputs $Q_k$. 

\vspace{5mm}

\noindent\textit{Proof of invariants.} We prove the invariants using induction over the number of iterations.

Base case $s=1$. (I1.1) and (I2) follow from the analysis of \detsparsification in  \Cref{lem:Gsparsification}, which states that for any $v \in V$ it holds that $d_G(v,Q_1) \le 72\log n$ and $\dist_G(v, Q_1) \le 2 + \dist_G(v, Q_0)$. (I1.2) is a direct consequence of (I1.1), since for any $v \in V$, $d^2(v, Q_1) = |N^2(v, Q_1)| = |\cup_{w \in N(v)} N(w, Q_1)| \le \sum_{w \in N(v)} |N(w,Q_1)| \le \Delta \cdot 72\log n$. (I3) holds by \Cref{lem:sendingIDs}. 

Assume that the invariants hold for $s-1$. In the $s$th iteration, we set $Q_{s-1}$ as the initial set of active nodes. By (I1.2) for iteration $s-1$, the distance-$s$ $Q_{s-1}$-degree is at most $72\Delta \log n$ for all $v \in V$. The maximum active degree parameter $\maxactivedeg$ in \detsparsification was set to $\maxactivedeg^{(s)} = 72\Delta \log n$. This satisfies the requirement $\maxactivedeg \ge \max_{v \in V} d(v, A)$ in  \Cref{lem:Gsparsification} (which has the form $\maxactivedeg \ge \max_{v \in V} d^s(v,A)$ as we are simulating \detsparsification on $G^s$). Now, we can apply \Cref{lem:Gsparsification}. For $G^s$, \Cref{lem:Gsparsification} (bounded $Q$-degree) states that,  for all $v \in V$, $d^s(v,Q_s) \le 72\log n$, proving (I1.1). (I1.2) is again a consequence of (I1.1), as for any $v \in V$, $d^{s+1}(v, Q_{s}) = |N^{s+1}(v, Q_s)| = |\cup_{w \in N(v)} N^s(w, Q_s)| \le \sum_{w \in N(v)} d^s(w,Q_s) \le \Delta \cdot 72\log n$. For (I2), we use \Cref{lem:Gsparsification} (domination), which states that $\forall v \in V: \dist_{G^s}(v,Q_s) \le 2 + \dist_{G^s}(v,Q_{s-1})$. The increase of 2 in distance in $G^s$ corresponds to an increase of at most $2s$ in $G$. Combined with (I2) for iteration $s-1$, we get $\dist_G(v,Q_s) \le 2s + \dist_G(v,Q_{s-1}) \le 2s + (s-1)^2 + s-1 + \dist_G(v,Q_0) = s^2 + s + \dist_G(v,Q_0)$. Lastly, (I3) holds by construction (see \Cref{lem:sendingIDs}). \QED

\begin{proof}[Proof of \Cref{lem:sparsification}] 
	Initialize the active nodes as $Q_0 \subseteq V$. Run the algorithm of \Cref{sec:sparsGk} with $k$ iterations. Let $Q:=Q_k$ be the result of the last iteration. 
	
	The correctness of the algorithm follows from the invariants. For the $k$th iteration, (I1.1) states that $\forall v \in V: d^k(v, Q) \le 72\log n$ and by (I2), $\forall v \in V: \dist_G(v, Q) \le k^2 + k + \dist_G(v, Q_0)$. 
	
	Lastly, we analyze the runtime of the algorithm. There are $k$ iterations. In each iteration $s$, we run \detsparsification with $\lfloor \log \maxactivedeg - \log \log n\rfloor - 5$ stages. The parameter $\maxactivedeg$ is set to $\maxactivedeg^{(1)} = \Delta$ for $s=1$ and $\maxactivedeg^{(s)} = 72\Delta \log n$ for $2 \le s \le k$. Hence, the number of stages is $O(\log \Delta)$ in all iterations. Simulating \detsparsification on $G^s$ takes $O(\diam(G) \cdot \log^2 n \cdot \log \Delta + s \log \Delta)$ rounds by \Cref{lem:simDetSparsification}. At the end of the iteration, each $v \in V$ sends the set of IDs in $N^s(v, Q_s)$ to each of their neighbors $w \in N(v)$. Pipelining the IDs takes $O(\log n)$ rounds, since $d^s(v, Q_s)= O(\log n)$ by (I1). The time complexity of the $s$th iteration is $O(\diam(G) \cdot \log^2 n \cdot \log \Delta + s \log \Delta + \log n) = O(\diam(G) \cdot \log^2 n \cdot \log \Delta + k\log \Delta)$. Hence, the total time complexity is $O(k \cdot \diam(G) \cdot \log^2 n \cdot \log \Delta + k^2 \cdot \log \Delta)$. 
\end{proof} 

Lastly, we show how to simulate \detsparsification on $G^s$, for any $1 \le s \le k$, where the active nodes are initialized as the result of the previous iteration $s-1$. The efficiency of the simulation relies on the invariant (I1.1) for the previous iteration.

\begin{lemma}[Simulating \detsparsification on $G^s$]
	\label{lem:simDetSparsification}
	Consider iteration $2 \le s \le k$. Let $\maxactivedeg = 72\Delta \log n$. Initialize the set of active nodes as $Q_{s-1}$, where $Q_{s-1} \subseteq Q_0$ is the set of nodes returned in iteration $s-1$. Given that the invariants hold for iteration $s-1$, there is a deterministic \congest algorithm that simulates \detsparsification on $G^s$ with communication network $G$ in $O(\diam(G) \cdot \log^2 n \cdot \log \Delta + s \log \Delta)$ rounds.
\end{lemma}
\begin{proof}
	We simulate \detsparsification, consisting of $r=\lfloor \log \maxactivedeg - \log \log n \rfloor -5 = O(\log \Delta)$ stages. The active nodes are initialized as $H_1 := Q_{s-1}$ and $\maxactivedeg = 72\Delta \log n$. A stage of \detsparsification was described in \Cref{lem:derandstage}. We show how to efficiently simulate it on $G^s$ with communication network $G$. Fix any stage $1 \le i \le r$. Assume that each $v \in V$ knows the IDs of its active neighbors, $N^s(v, H_{i})$. This is true for the first stage by (I3) for iteration $s-1$: each $v \in V$ knows the set of IDs in $N^s(v, H_1) = N^s(v, Q_{s-1})$. We later show how to learn the IDs in $N^s(v,H_{i+1})$ at the end of the stage.
	
	The $i$th stage starts with derandomizing the choices of the active nodes. For clarity, we first rewrite the definitions of events in \Cref{lem:derandstage} for $G^s$. For $v \in H_i$, let $X_v$ be an indicator variable for the event that $v$ is sampled. Let $Z = \{v \in V: d^s(v,H_i) \ge \maxactivedeg / 2^{i}\}$ be the set of nodes with high active degree in $G^s$ in the $i$th stage. For each $v \in Z$, let $\Phi_v := \One(v \not\in M_i \cup N^s(M_i))$ and for each $v' \in V \setminus Z$, let $\Phi_{v'} := 0$. For each $v \in V$, let $\Psi_v := \One(d^s(v, M_i) > 72\log n)$. 
	Note that these are exactly the events of \Cref{lem:derandstage}, rewritten for $G^s$. Given that we are simulating the $i$th stage of the sparsification algorithm in $G^s$, the randomized analysis in \Cref{lem:randstage} states that for all $v \in V$, $\pr(\Phi_v=1)\le 1/n^3$ and $\pr(\Psi_v=1)\le 1/n^3$.
	
	As in the \detsparsification, we apply \Cref{lem:conditionalEVs} to fix the decisions of the active nodes. 
	For all $v \in V$, the events $\Phi_v$ and $\Psi_v$ depend on the decisions of their active neighbors in $G^s$. By assumption, each $v \in V$ knows the IDs in $N^s(v,H_i)$. 
	Hence, by \Cref{lem:conditionalEVs}, the values of $X_v, v \in H$ can be fixed in $O(\diam(G) \cdot \log^2 n)$ rounds, such that neither $\Phi_v$ nor $\Psi_v$ occur for any $v \in V$. Note that the only communication in \Cref{lem:conditionalEVs} is done by aggregating values on a spanning tree subgraph of $G$. In particular, there is no need to simulate communication between neighbors in $G^s$. Let $M_i = \{v \in H_i: X_v = 1\}$ be the set of sampled nodes.
	
	Each $v \in M_i$ deactivates itself. 
	We also need to deactivate the distance-2 neighborhood in $G^s$ (distance-$2s$ neighborhood in $G$) of sampled nodes. 
	Each $v \in M_i$ sends a flag \textit{sampled} that propagates for $2s$ hops in $G$. Nodes forward an incoming flag, keeping track of the distance left. Multiple flags can be combined into one, since they do not contain any information specific to the sender, so there is no congestion. The round complexity is $2s=O(s)$. Any active node $v' \in H_i$ who receives a flag \textit{sampled} deactivates itself. Let $H_{i+1} \subseteq H_i$ be the remaining active nodes.
	
	To prepare for the next stage, each $v \in V$ needs to learn the IDs in $N^s(v,H_{i+1}) \subseteq N^s(v, H_i)$. To do this, each $v \in H_i \setminus H_{i+1}$ informs its neighborhood in $G^s$ about the fact that $v$ is no longer active. Using \Cref{lem:comms}, each deactivated node $v$ \textit{broadcasts} a message \textit{deactivated}, along with $\ID(v)$, to all of its neighbors in $G^s$. By (I1.1) for iteration $s-1$, any $v \in V$ has at most $O(\log n)$ distance-$(s-1)$ $Q_{s-1}$-neighbors. The BFS trees of depth $s$ around each $x \in Q_{s-1}$ are given by (I3) for iteration $s-1$. Hence, \Cref{lem:comms} sends the broadcasts from  all deactivated nodes (possibly all of $Q_{s-1} \supseteq H_i$) in parallel in $O(s + \log n)$ rounds. Nodes $w \in V$ use the incoming broadcasts to form knowledge of $N^s(v,H_{i+1})$ by removing deactivated nodes from the set $N^s(v,H_i)$.\footnote{Note that nodes $v \in V$ who have a sampled neighbor in $G^s$ could perform the update in zero rounds: $v$ checks if any of its active neighbors $w \in N^s(v,H_i)$ are sampled, using the hash function $h$ determined by the fixed random bits. If there exists a sampled node $w \in N^s(v,H_i)$, all of the remaining active nodes in $N^s(v,H_i)$ will have already been deactivated, because they are in the distance-$2$ neighborhood of $w$ in $G^s$. However, the sampled node may not be in $N^s(v,H_i)$, but only in the distance-$2s$ neighborhood of some active node $w \in N^s(v,H_i)$. In this case, $w$ is deactivated but $v$ does not know it yet.} This concludes the simulation of one stage of \detsparsification. 
	
	The time complexity of one stage is $O(\diam(G) \cdot \log^2 n + 2s + s + \log n)$. There are $r = O(\log \Delta)$ stages, so the time complexity of simulating \detsparsification in $G^s$ is $O(\diam(G) \cdot \log^2 n \cdot \log \Delta + s \log \Delta)$.
\end{proof}

\subsection{Sparsification with no Diameter Dependency}
\label{app:noDiameter}
We combine the result of \Cref{lem:sparsification} with a network decomposition algorithm to improve the sparsification runtime to $O(\poly \log n)$. See \Cref{def:networkDecomp} for the definition of a network decomposition. With a network decomposition, the classic approach is to compute a solution greedily, by iterating through the colors of the network decomposition. However, we need to be careful not to spoil the solution for nodes in clusters of later colors. When computing a solution for a cluster $C$, we must also consider nodes in the distance-$k$ neighborhood $N^k(C)$. This \textit{cluster border} acts as inactive observers in the algorithm to make sure that they have a bounded number of distance-$k$ neighbors in $Q$. Furthermore, nodes in the cluster  border are important, because connections in $G^k[C]$ can also be formed through neighbors of $C$ (up to distance $\lfloor k/2 \rfloor$). With this in mind, we need a separation of $2k+1$ between clusters to run \Cref{lem:sparsification} independently for each cluster of same color. 

\begin{restatable}{lemma}{corSparsification} \label{lem:ndSparsification}
	Let $k \ge 1$ (potentially a function of $n$). Let $T^{ND}$ be the complexity of computing a weak $(c,d)$-network decomposition deterministically with cluster distance $2k+1$ and congestion $\tau$.
	
	There is a deterministic distributed algorithm that, given a subset $Q_0 \subseteq V$, finds a set of vertices $Q \subseteq Q_0$ such that for all $v \in V$, 
	\begin{itemize}
		\item (bounded distance-$k$ $Q$-degree): $d^k(v,Q) \le 72 \log n = O(\log n)$
		\item (domination): $\dist_{G}(v, Q) \le k^2 + k + \dist_{G}(v, Q_0) = O(k^2) + \dist_{G}(v, Q_0)$
	\end{itemize}
	The algorithm runs in 
	$O\big(T^{ND} + c (d \tau k \cdot \log^2 n \cdot \log \Delta + k^2 \cdot \log \Delta) \big)$ rounds in the \congest model. 
	The runtime is $\widetilde{O}(k^2\cdot \log^4 n \cdot \log \Delta)$, using the network decomposition of \Cref{thm:networkDecomp}. 
\end{restatable}

\begin{proof}
	We start by computing a weak $(c,d)$-network decomposition deterministically with cluster distance $2k+1$ and congestion $\tau$. 
	Define a set of \textit{globally active nodes} $H_G := Q_0$. We iterate through the $c$ color classes of the network decomposition. The sparsification algorithm of \Cref{lem:sparsification} is run on clusters of the same color in parallel. The set of globally active nodes is updated after each color. 
	
	Fix a color $1 \le j \le c$ and let $C \subseteq V$ be any cluster of color $j$. The distance-$k$ neighborhood of $C$ also takes part in the algorithm. We refer to $N^k(C)$ as the \textit{cluster border}. 
	Separation of $2k+1$ guarantees that cluster borders are disjoint for clusters of the same color. Let $T_C$ be the Steiner tree of $C$, given by the network decomposition. We can extend $T_C$ to include $N^k(C)$ in $O(k)$ rounds by running a BFS starting from each $v \in C$. Nodes outside of the cluster forward one of the incoming searches for $k$ hops in total. Let $T_C$ be the updated Steiner tree (not necessarily a tree anymore). By disjointness of the cluster borders, congestion is increased on any edge by at most one.

	The final output $Q$ is the union of $Q_C$ sets over all colors and clusters $C$. We prove correctness by showing the two properties of \Cref{lem:sparsification} for $Q$. (Bounded distance-$k$ $Q$-degree): Consider any $v \in V$. Let $C$ be the first cluster containing a distance-$k$ neighbor of $v$ that is selected to $Q$. $C$ is unique, since clusters of the same color are at least $2k+1$ hops apart. The number of distance-$k$ $Q$-neighbors of $v$ is at most $72\log n$ by (I1) for iteration $k$. Furthermore, any sampled node deactivates the globally active nodes in its distance-$2k$ neighborhood. Hence, no more distance-$k$ neighbors of $v$ join $Q$ in later colors. (Domination): Let $x \in Q_0$ be any initially active node closest to $v$. There are three possibilities: (i) $x \in Q$, (ii) $x$ starts as active in its own cluster, but $x \not\in Q$ (iii) $x$ is deactivated before its own cluster runs the algorithm. For (i), the domination clearly holds. In the second case, there is another node in $Q$ in the cluster of $x$, because $x$ is not in $Q$. The distance from $x$ to a node in $Q$ in the cluster is at most $k^2 + k$ by \Cref{lem:sparsification}. In case (iii), the reason that $x$ was deactivated is because there is a selected node in another cluster, within distance $2k$ from $x$. Hence, the domination property holds in all three cases.  
	
	Lastly, we analyze the runtime. Computing the network decomposition takes $T^{\text{ND}}$ rounds. We iterate through $c$ colors. The diameter of each cluster is $d$ and communication in the Steiner tree of each cluster is slowed down by a $O(\tau)$ factor due to congestion. Computing the cluster border and deactivating globally active nodes in other clusters takes $O(k)$ rounds each. Hence, running the sparsification algorithm for a single color class takes $O(d \tau k \cdot \log^2 n \cdot \log \Delta + k^2 \log \Delta))$ rounds. The total runtime is $O(T^{\text{ND}} + c(d\tau k \cdot \log^2 n \cdot \log \Delta + k^2 \log \Delta))$. Using the network decomposition algorithm of \Cref{thm:networkDecomp}, we compute a $(\widetilde{O}(\log n), O(k \cdot \log n))$-network decomposition with separation $2k+1$ and $\tau=1$ in $\widetilde{O}(k \cdot \log^3 n)$ rounds. The total runtime is $\widetilde{O}(k^2 \cdot \log^4 n \cdot \log \Delta)$. 
\end{proof}

\section{Deterministic \texorpdfstring{$(k+1,k^2)$}{(k+1,k\^2}-ruling set (Theorem~\ref{corollary:kksquaredRulingSet})}
\label{sec:DetRulingSets}
In this section, we prove our main result (\Cref{corollary:kksquaredRulingSet}), that is, we present our $\poly \log n$-time deterministic $(k+1, k^2)$-ruling set algorithm. We start by showing that the simplest known deterministic ruling set algorithms extend to power graphs. The downside of these algorithms is that they either provide a poor domination property or are exponentially slower than our main result.

\paragraph{Existing ruling set algorithms for power graphs}
There is a well-known deterministic $(k+1, k\cdot\log n)$-ruling set algorithm \cite{awerbuch1989network}, adapted to the \congest model in \cite{henzinger2016deterministic}. A general form of the algorithm was given by \cite{KMW18}:

\begin{theorem}[\cite{awerbuch1989network,SEW13,henzinger2016deterministic,KMW18}~] 
	\label{thm:detRulingOld}
	Let $k$ be a positive integer. Given a distance-$k$ coloring of $G$ with $\gamma$ colors, there is a deterministic distributed algorithm that, for any $B \ge 2$, computes a $(k+1, k \cdot \lceil \log_B \gamma\rceil)$-ruling set of $G$ in $O(k \cdot B \cdot \log_B \gamma)$ rounds of the \congest model. 
\end{theorem}
\noindent\textit{Proof sketch.} % see eg. KMW18 Lemma 4.2
Consider the base-$B$ representation of the input coloring, which can be given in $m:=\lceil \log_B \gamma \rceil$ digits. Let $\chi_i(v) \in \{0, \dots, B-1\}$ be the $i$th digit of the color of $v \in V$ in base $B$. Let $U:=V$ be a set of nodes, which is eventually returned as the ruling set. Initially, $U$ is a trivial dominating set, with no independence guarantees. 
The algorithm iterates through the $m$ digits of the coloring. Each iteration $1 \le i \le m$ maintains the following invariant: for any $v, w \in U$ such that $\dist_G(v,w) \le k$, the colors of $v$ and $w$ agree in the first $i$ digits. This implies that after $m$ iterations, any pair of nodes in $U$ within distance $k$ agree in all digits, i.e., have the same color. Such a pair does not exist, since the input is a coloring of $G^k$, hence the result is $(k+1)$-independent. 

Fix some iteration $1 \le i \le m$. It consists of $B$ steps, going through the possible values of the $i$th digit. In the $s$th step, for any $0 \le s \le B-1$, the nodes $v \in U$ with $\chi_i(v) = s$ send a beep to their distance-$k$ neighborhood. Any $w \in U$ with $\chi_i(w) > s$, who receives a beep, removes itself from $U$ permanently (note that beeping nodes do not need to be able to listen to beeps of other nodes while beeping). This maintains the invariant for the $i$th iteration: after $B$ steps, there are no distance-$k$ neighbors $v,w \in U$ with $\chi_i(v) \neq \chi_i(w)$. The $s$th iteration  weakens the domination of $U$ by at most $k$: nodes who beep and cause some of their distance-$k$ neighbors to be removed, will remain in $U$ for the rest of the iteration. Hence, the domination of the final ruling set is $k \cdot m = k \cdot \lceil \log_B \gamma \rceil$. The runtime of one iteration is $k \cdot B$ rounds. The total runtime is $O(k \cdot B \cdot m)=O(k \cdot B \cdot \log_B \gamma)$ \QED

\vspace{3mm}

The algorithm of \Cref{thm:detRulingOld} scales well for larger distances, but the domination is non-constant and depends on the size of the input coloring. Furthermore, deterministic coloring of power graphs is difficult. For larger distances, one usually has to rely on a coloring given by the unique IDs, which makes the domination logarithmic in $n$. Constant domination can be achieved by choosing a non-constant base $B$, but the runtime grows significantly:

\begin{corollary}
	Let $k$ be a positive integer. There is a deterministic distributed algorithm that, for any $c \ge 1$, computes a $(k+1, ck)$-ruling set of $G$ in $O(k \cdot c \cdot n^{1/c})$ rounds of the \congest model.   
\end{corollary}
\begin{proof}
	Use the $(k+1, k \lceil \log_B \chi\rceil)$-ruling set algorithm of \Cref{thm:detRulingOld} with the unique IDs of the nodes as colors and $B=\lceil n^{1/c}\rceil$. The runtime is $O(k \cdot c \cdot n^{1/c})$ rounds. The resulting ruling set has domination 
	$k \lceil \frac{\log n}{\log \lceil n^{1/c}\rceil} \rceil \le k c$. 
\end{proof}

\noindent To our knowledge, this is the only deterministic ruling set algorithm that, for $k \ge 2$ and $k=o(\log n)$,  computes a ruling set with domination constant in $n$. For $k = \Omega(\log n)$, it is possible to compute a $(k+1,O(k^2))$-ruling set with \Cref{thm:detRulingOld} in $O(k \cdot \log n)$ time. 

Network decompositions also yield ruling sets. Given a $(c,d)$-network decomposition (weak or strong diameter), with separation $2k$ between clusters of the same color, it is possible to compute a $(k+1,\max(d,k))$-ruling set in $O(c(d + k))$ rounds. Initially, all nodes are active and the ruling set is empty. Iterate through the $c$ color classes. In clusters of a given color class, pick any active node (if exists) to join the ruling set, e.g. by choosing the one with the minimum ID in O(d) time. Nodes joining the ruling set deactivate their distance-$k$ neighborhood in the whole graph. The distance-$k$ neighborhoods are disjoint for distinct clusters of the same color, hence the result is $(k+1)$-independent. The domination is at least $d$, because there is a ruling set node in any cluster, unless all of the nodes of the cluster are already dominated. The current state of the art network decomposition with larger separation (see \Cref{thm:networkDecomp}) does not improve from the result of \Cref{thm:detRulingOld}. Furthermore, to find ruling sets of power graphs with constant domination, we would need the cluster diameter to be constant. This is not possible with a polylogarithmic number of colors, as a constant diameter implies $\Omega(n^{1/\lambda})$ color classes for some constant $\lambda$ \cite{LS93}.

\paragraph{Deterministic $(k+1,k^2)$-ruling set} 
Our deterministic ruling set algorithm computes a $(k+1,k^2)$-ruling set for any $k\ge 1$ in $\poly \log n$-time. The algorithm combines sparsification with a maximal independent set procedure on an induced subgraph of $G^k$. The goal of sparsification is to find a set of vertices $Q \subseteq V$, such that all $v \in V$ are close to $Q$ in $G$, while having a bounded number of distance-$(k-1)$ neighbors in $Q$. The second property makes it possible to efficiently simulate algorithms on $G^k[Q]$ with communication network $G$. In particular, we can run any maximal independent set algorithm in a black-box way to compute a MIS of $G^k[Q]$. The result is returned as the ruling set. The following lemma describes the general approach formally:

\begin{restatable}[Ruling set via sparsification]{lemma}{thmDetRuling}
	\label{theorem:sparsificationToRulingset}
	Let $k \ge 1$ (potentially a function of $n$). Denote the complexity of solving the following problems deterministically in the \congest model by
	\begin{itemize}
		\item $T^{\mathsf{MIS}}(n, \Delta)$: complexity of solving MIS in graphs with $n$ nodes and maximum degree~$\Delta$ 
		\item $T^\mathsf{sparsification}(k, \beta, \maxdeg)$: complexity of finding a subset $Q \subseteq V$ such that all $v \in V$ have $d^k(v,Q) \le \maxdeg$ and $\dist_G(v,Q) \le \beta$ (in graphs with $n$ nodes and maximum degree~$\Delta$) 
	\end{itemize}
	There is a deterministic distributed algorithm that, given a graph $G=(V,E)$ with $n$ nodes and \linebreak maximum degree $\Delta$, computes a $(k+1,\beta + k)$-ruling set of $G$ in 
	$O(T^\mathsf{sparsification}(k-1,\beta,\maxdeg) + k  \maxdeg + \linebreak (k+ \maxdeg^2) \cdot T^{\text{MIS}}(n,\Delta \cdot \maxdeg))$
	rounds of the \congest model. 
	
\end{restatable}

\begin{proof}%[Proof of \Cref{theorem:sparsificationToRulingset}:]
	Let $\beta$ and $\maxdeg$ be any positive integers. Start by finding a subset of nodes $Q \subseteq V$ such that all $v \in V$ have at most $\maxdeg$ distance-($k-1$) $Q$-neighbors and $\dist_G(v,Q) \le \beta$. The time complexity of finding $Q$ is $T^\mathsf{sparsification}(k-1,\beta,\maxdeg)$. 
	
	We simulate computing a maximal independent set on $G^{k}[Q]$ with communication network $G$ using \Cref{lem:simulation}. To prepare for the simulation, we need to form BFS trees of depth $k$ rooted at nodes in $Q$, used for communication between nodes in $Q$. The trees are also provided by invariant (I3) of \Cref{sec:sparsGk}, so the following preparation step can be skipped when using the sparsification algorithm of \Cref{lem:sparsification}. For each $v \in Q$, let $T_v$ be a BFS tree, initially containing just $v$. Run $k-1$ iterations of \Cref{lem:sendingIDs}. In the $s$th iteration, for any $1 \le s \le k-1$, each $v \in V$ learns the set of IDs in $N^{s+1}(v,Q)$ and the tree of each $v \in Q$ is extended to cover $N^{s+1}(v)$. Any $v \in V$ knows, for each tree $T$ it belongs to, the ID of the root, descendants$(T,v) \subseteq N(v)$ and ancestor$(T,v) \in N(v)$.  The $s$th iteration takes $O(\maxdeg)$ rounds, using the $\Theta(\log n)$-bit communication bandwidth and $O(\log n)$-bit IDs. After $k-1$ iterations, required BFS trees are formed and known distributedly. Also, each $v \in Q$ knows its degree and the IDs of its neighbors in $G^k[Q]$. The preparation step takes $O(k \cdot \maxdeg)$ rounds. 
	
	The maximum degree of $G^{k}[Q]$ is at most $\Delta \cdot \maxdeg$, since all $v \in V$ have at most $\Delta$ immediate neighbors in $G$, each of which has at most $\maxdeg$ distance-$(k-1)$ $Q$-neighbors. The number of nodes in $G^{k}[Q]$ is at most $n$. The complexity of computing a MIS on an input graph $G^k[Q]$ is $T^{\mathsf{MIS}}(n, \Delta \cdot \maxdeg)$. We simulate computing MIS on $G^{k}[Q]$ with communication network $G$. By \Cref{lem:simulation}, we can simulate any \congest algorithm on $G^k[Q]$ with communication network $G$ with an $O(k +\maxdeg^2 )$ slowdown factor. This gives a runtime of $O((k +\maxdeg^2)\cdot T^{\mathsf{MIS}}(n, \Delta \cdot \maxdeg))$ for the MIS step. Let $\IS \subseteq Q$ be the resulting MIS of $G^k[Q]$. 
	
	For correctness, note that a maximal independent set of $G^{k}$ is a $(k+1,k)$-ruling set of $G$. A maximal independent set of $G^{k}[Q]$ is independent in $G^{k}$. It is not maximal in $G^k$, we only have $\dist_G(v,\IS) \le k$ for all $v \in Q$. Hence, for all $v \in V$, $\dist_G(v,\IS) \le \dist_G(v, Q) + k \le \beta + k$. The total runtime is $T^\mathsf{sparsification}(k-1) + k\maxdeg + (k+\maxdeg^2 ) \cdot T^\mathsf{MIS}(n, \Delta\cdot \maxdeg)$.
\end{proof}

We use the sparsification algorithm of \Cref{lem:ndSparsification} and the maximal independent set algorithm of \cite{FGG22} to compute a $(k+1,k^2)$-ruling set in $\poly \log n$-time: 

\corDetRulingFinal*

\begin{proof}%[Proof of \Cref{corollary:kksquaredRulingSet}:]
	Apply \Cref{theorem:sparsificationToRulingset}. We have $T^\mathsf{sparsification}(k) = \widetilde{O}(k^2 \cdot \log^4 n \cdot \log \Delta)$ by \Cref{lem:ndSparsification}. For the computation of maximal independent set, we use \cite[Theorem 3.1]{FGG22} which computes a MIS in $O(\log^2 \Delta \cdot \log\log \Delta \cdot \log n)$ rounds in any $n$-node graph with maximum degree $\Delta$. This gives $T^{\text{MIS}}(n, \Delta \log n) = \widetilde{O}(\log n \cdot \log^2 \Delta)$. The runtime of \Cref{theorem:sparsificationToRulingset} becomes
	$\widetilde{O}(k^2 \cdot \log^4 n \cdot \log \Delta + (k + \log^2 n)\cdot \log n \cdot \log^2 \Delta) = \widetilde{O}(k^2 \cdot \log^4 n \cdot \log \Delta)$.
\end{proof}

\section{Maximal Independent Set on \texorpdfstring{$G$}{G}: Shattering Revisited (Theorem~\ref{thm:MISrepaired})}
\label{sec:revisedShattering}

In this section, we  present   randomized algorithm for computing an MIS (of $G$) in the \LOCAL and the \CONGEST model.
Just as previous works, we rely on the shattering framework and our high level approach is identical to these works, mainly to \cite{BEPS16,BEPS16,ghaffari16_MIS,Gha19}. In fact, we use the randomized pre-shattering phase of \cite{ghaffari16_MIS,Gha19} in a black-box manner in order to effectively reduce the problem to small unsolved components. \pagebreak Then in \cite{BEPS16,BEPS16,ghaffari16_MIS,Gha19} and in our work these small components are solved in the \emph{post-shattering phase}. We present two different post-shattering approaches that differ algorithmically and/or analytically from previous works. Additionally, they are simpler than previous approaches.

\subsection{Basics of Shattering} 
As we rely on the same general shattering framework \cite{BEPS16} and the pre-shattering phase of \cite{ghaffari16_MIS} this section is very similar to statements in these works. 

\paragraph{Pre-shattering.} The first part of the algorithm is called the \textit{pre-shattering} phase. It consists of $O(\log \Delta)$ steps of a random process, simply referred to as IndependentSet. 
The pre-shattering phase computes an independent set $\IS \subseteq V$. Some nodes $B := V \setminus (\IS \cup N(\IS))$ remain undecided, as they didn't join $\IS$ and do not have a neighbor who joined it. It can be shown that, for any 5-independent set $U \subseteq V$, the probability that all nodes in $U$ remain undecided is low:
\begin{lemma}[\cite{ghaffari16_MIS}]\label{lem:5indepUndecided}
    Let $H=(V,E)$ with at most $n$ nodes and maximum degree at most $\Delta$. Let $c > 0$ be an arbitrary constant. Let $U \subseteq V$ be any 5-independent set in $H$. The probability that $U \subseteq B$ after running IndependentSet on $H$ for $\Theta(c \log \Delta)$ steps is at most $\Delta^{-c|U|}$.
\end{lemma}
\noindent The same statement as in \Cref{lem:5indepUndecided} holds for the algorithm of \cite{BEPS16} after $O(\log^2 \Delta)$ steps. A consequence of \Cref{lem:5indepUndecided} is that, for any fixed 5-independent $U \subseteq V$ with size at least $t := \log_\Delta n$, the probability that all nodes in $U$ remain undecided is at most $1/n^c$. This only holds for a fixed $U$, and the number of choices for such $U$ is prohibitively high to get a with high probability guarantee by taking a union bound. For this reason, an additional connectedness requirement is added. Recall that a set $U \subseteq V$ is $s$-connected in $G$ if for all $U' \subset U: \dist_G(U', U \setminus U) \le s$. We also use the fact that a $(5,4)$-ruling set of a $s$-connected set $U$ is $5$-independent and $(s+8)$-connected.

\vspace{2mm}
\begin{restatable}[Ruling set of $s$-connected component]{lemma}{lemrulingConnected}
	\label{lem:54rulingConnected}
	Let $H=(V,E)$. Let $\alpha, s \ge 1$ and let $\beta \ge \alpha -1$. Let $U \subseteq V$ be an $s$-connected set in $H$. Any $R \subseteq U$ that is an $(\alpha,\beta)$-ruling set of $U$ with respect to distances in $H$, is $\alpha$-independent and $(s+2\beta)$-connected in $H$. 
\end{restatable}
\begin{proof}
    See \Cref{fig:abRulingConnected} for an illustration of the proof. First, $R$ is $\alpha$-independent in $H$ by definition. To prove that $R$ is $(s+2\beta)$-connected, we need to show that for all $R' \subset R: \dist_H(R', R \setminus R') \le s+2\beta$. Note that by the domination of $R$, for any $v \in U$, there exists some $w \in R$ s.t. $\dist_H(v,w) \le \beta$. Now, take any $R' \subset R$. Let $X = N^\beta_H(R',U) = \cup_{u \in R'} N^\beta_H(u,U)$ be the distance-$\beta$ $U$-neighborhood in $H$ of nodes in $R'$. By the $s$-connectedness of $U$, we have $\dist_H(X, U\setminus X) \le s$. Let $x \in X$ and $y \in U \setminus X$ be some minimizers of $\dist_H(X, U\setminus X)$. We have $\dist_H(R', x) \le \beta$ by definition of $X$. For $y$, there exists some $w \in R$ s.t. $\dist_H(y, w) \le \beta$. Furthermore, since $y \not\in X$, we know that $y$ cannot be dominated by a node in $R'$, hence $w \in R \setminus R'$. In total, this gives
	\begin{align}
		\dist_H(R', R \setminus R') \le \dist_H(R', x) + \dist_H(x,y) + \dist_H(y, R \setminus R') \le \beta + s + \beta & 
		\qedhere
	\end{align} 
\end{proof}

\begin{figure}[h]
	\caption{A $\beta$-dominating set $R$ (red and green nodes) of a $s$-connected set $U$ (in grey) is $(s+2\beta)$-connected. $X \subseteq U$ is the distance-$\beta$ $U$-neighborhood of $R' \subset R$. The distance from $X$ to $U \setminus X$ is at most $s$. There, $y \in U\setminus X$ is dominated by another node $w \in R \setminus R'$.}
	\label{fig:abRulingConnected}
	\centering
	\includegraphics[width=0.7\textwidth]{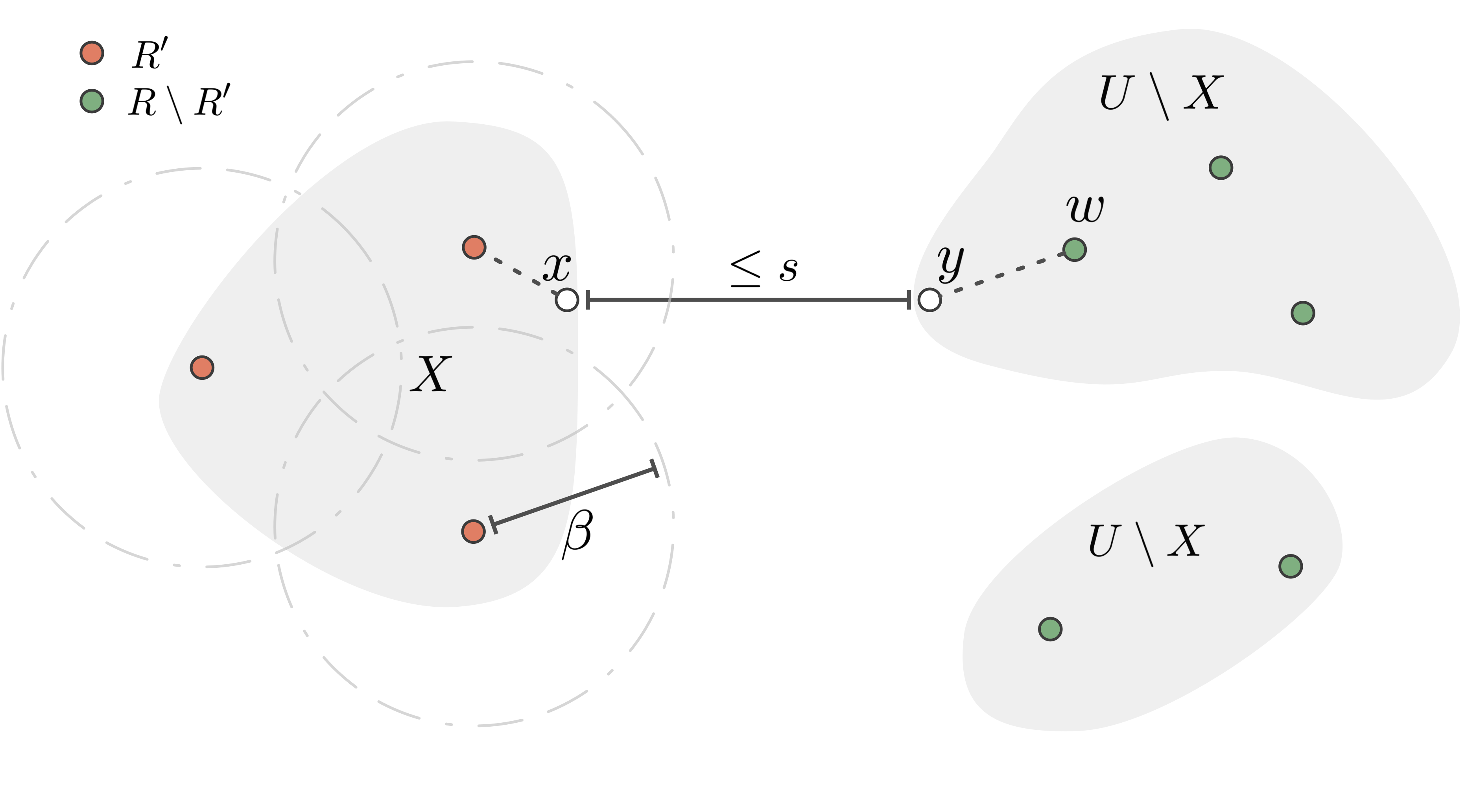}
\end{figure}

This makes it possible to prove the following \textit{shattering guarantee} (\cite[Theorem 3.3]{BEPS16}, \cite[Lemma 4.2]{ghaffari16_MIS}). We use a slightly more general version of the same statement:

\newpage
\begin{lemma}[Shattering]\label{lem:shatteringV2}
    Let $s \ge 1$ and $H=(V,E)$ be any graph with at most $n$ nodes and maximum degree at most~$\Delta$. Let  $t=\log_\Delta n$. Run $\Theta(s \log \Delta)$ steps of IndependentSet on $H$. Let $\IS \subseteq V$ be the computed independent set and $B = V \setminus (\IS \cup N(\IS))$ be the undecided nodes. With high probability in $n$, 
    \begin{itemize}
        \item [(P1)] There is no 5-independent $(8+s)$-connected $U \subseteq B$ s.t. $|U| \ge t$.
        \item [(P2)] All s-connected sets $C \subseteq B$ have at most $O(t \cdot \Delta^4)$ nodes.
    \end{itemize}
\end{lemma}
%\newpage
\begin{proof}
    (P1) Let $s' = s+8$. Any $U \subseteq B$ with these properties forms a $t$-node tree in $H^{[5,s']}$, where $H^{[5,s']}$ is the power graph where $v,w \in V$ are adjacent if $5 \le \dist_H(v,w) \le s'$. There may be many possible trees for $U$. In general, the number of rooted unlabeled $t$-node trees is less than $4^t$, because the Euler tour of such trees can be encoded in $2t$ bits. The number of ways to embed such a tree in $H^{[5,s']}$ is less than $n\cdot \Delta^{s'(t-1)}$. By \Cref{lem:5indepUndecided}, for any constant $c > 0$, the probability that $U \subseteq B$ after $\Theta((s+c) \log \Delta)$ steps is at most $\Delta^{-(s+c) \cdot |U|}$. By union bound, the probability that any such $U$ is contained in the undecided nodes is at most
    $$
        4^t \cdot n \cdot \Delta^{s'(t-1)} \cdot \Delta^{-(s+c)\cdot t}
        = n^{\log_\Delta 4 + 1} \cdot \Delta^{s'(t-1) -(s+c) \cdot t} 
        \le n^{\log_\Delta 4 + 1} \cdot \Delta^{t(8 - c)} 
        % s'(t-1) - (s+c)t = s't -s' -st + ct = st +8t -s -8 - st + ct = 8t - s - 8 + ct = t(8+c) - (s + 8)
        %= n^{\log_\Delta 4 + 1} \cdot n^{8-c} 
        \le n^{10-c}
    $$
    (P2) Suppose there is a $s$-connected set $C\subseteq B$ with $t \cdot \Delta^4$ nodes. We show how to construct $U \subseteq B$ that violates (P1). Form a $(5,4)$-ruling set $U\subseteq C$ of $C$, with respect to distances in $G$, greedily. Initially, pick any $v \in C$ and let $U:= \{v\}$. Iteratively select $v \in C \setminus U$ for which $\dist_H(v,U) \ge 5$ and set $U = U \cup \{v\}$. This removes all nodes in $C$ within distance 4 of $v$ from consideration, which is at most $\Delta^4$ nodes. Hence $U$ has size at least $t \cdot \Delta^4 / \Delta^4 = t$. Since $C$ is $s$-connected and $U$ is a $(5,4)$-ruling set of $C$ w.r.t. distances in $H$, \Cref{lem:54rulingConnected} states that $U$ is 5-independent and $(s+8)$-connected in $H$. This is a contradiction of (P1).
\end{proof}

\subsection{Post-shattering.} \Cref{lem:shatteringV2} states that after one pre-shattering phase lasting for $O(\log \Delta)$ rounds, the  graph is shattered into \textit{small connected components} of undecided nodes: 
Formally, \Cref{lem:shatteringV2} (P2), with $s=1$, states that after one phase of pre-shattering the number of nodes in any connected component $C$ of $G[B]$ is at most $O(\log_{\Delta}n \cdot \Delta^4)$, with high probability. 
As $\Delta$ can be much larger than polylogarithmic, we rather rely on property $(P1)$ in order to algorithmically exploit the smallness of the component. 
As a consequence of (P1), in any connected component $C$, there does not exist any $U \subseteq V(C)$ such that $U$ is 5-independent in $G$ and $|U| \ge t$, with high probability. To see this, note that such a $U$ would be a subset of some $(5,4)$-ruling set $U^*$ of $V(C)$. As a $(5,4)$-ruling set of a connected component, $U^*$ is 5-independent and 9-connected by  \Cref{lem:54rulingConnected}. Hence, by (P1) $U^*$ does not remain fully undecided by \Cref{lem:shatteringV2}, with high probability.

\paragraph{\boldmath How to use $(P1)$ algorithmically?}  To get an intuition on how the bound of $P1$  could be used in the rest of the MIS algorithm, suppose we are given a ruling set $R_C \subseteq V(C)$, such that $R_C$ is 5-independent in $G$ and $h$-dominating in $C$. We can contract each undecided node to the closest ruling set node, forming balls of radius $h$. This defines a virtual \textit{ball graph}, with $R_C$ as nodes, and an edge between $v,w \in R_C$ if the corresponding balls share an edge. A network decomposition (see \Cref{sec:preliminaries} for a definition) of the ball graph is computed. This induces a network decomposition of the connected component:

\begin{claim} \label{claim:ndBallGraph}
    A network decomposition of a ball graph can be transformed to a network decomposition of $C$. The cluster diameter increases by a factor proportional to the ball diameter. 
\end{claim}
\begin{proof}
    Each cluster (a set of balls) is expanded to a set of nodes in $C$ in the natural way by taking a union of the balls. Any adjacent nodes in $C$ that are in the same ball are also in the same cluster. If they are in different balls, the balls are adjacent in the ball graph. A network decomposition of the ball graph clusters these balls in the same cluster, or in two clusters of different colors.  
\end{proof}

On a $t$-node graph, a network decomposition can be computed in $\poly\log t=\poly\log\log n$ time. Additionally, there is an $h$-factor slowdown due to the fact that each node is an $h$-radius ball in the communication network. The final MIS is computed by iterating through the colors of the network decomposition, and solving each cluster in time proportional to the cluster diameter (in \congest this requires more work).

\paragraph{Why do standard approaches not \emph{realize} the bound given by (P1)?} 
One approach is to compute a $(5,h)$-ruling set $R$ of $B$ w.r.t. distances in $G$, for all components at once. 
Now, nodes in $C$ may be dominated by a ruling set node in another component $C'\neq C$, i.e., balls formed around the ruling set nodes are not confined to one connected component, but may contain nodes from other components at distance $h$. Without better connectedness guarantees in $G$, \Cref{lem:shatteringV2} (P1) cannot bound the number of balls in a connected component of the ball graph. Another approach is to run the ruling set algorithm in each induced subgraph $C$, but this is wrong because \Cref{lem:shatteringV2} requires 5-independence in $G$, not $C$.

The rest of the section is devoted to presenting two approaches for actually making use of the small size of the components.

\subsubsection{Approach 1: Two pre-shattering phases}
\label{sec:twoPhaseShattering}
This version of the ruling set argument is most similar to the original proof in \cite[journal version]{BEPS16}. It is conceptually simple and, unlike the original proof, does not rely on the internals of the ruling set algorithm used. 

We start by running $O(\log \Delta)$ steps of IndependentSet on $G$. Let $\IS \subseteq V$ be the computed independent set, and let $B := V \setminus (\IS \cup N(\IS))$ be the undecided nodes. In light of previous considerations, we run the $O(\log \Delta)$ steps of IndependentSet again, this time on each connected component $C$ of $G[B]$ in parallel. 
Nodes can determine their incident edges in $G[B]$, which makes it possible to run the algorithm on each $G[V(C)]$ in parallel. 
Let $\IS_C \subseteq V(C)$ be the independent set computed in $C$, and let $B_C := V(C) \setminus (\IS_C \cup N_C(\IS_C))$ be the remaining undecided nodes. Running the MIS algorithm on $C$ guarantees that executions of nodes at least 5 hops apart \emph{in $C$} are independent. This allows us to apply \Cref{lem:5indepUndecided} on the subgraph $C$. We compute a $(5,O(\log \log n))$-ruling set $R_C$ of $B_C$, with respect to distances in $C$ with the algorithm of \cite[Lemma 2.2]{Gha19}. The algorithm succeeds with high probability in $n$, in fact for all the components. 

So far, this approach is the same as in \cite{BEPS16}, up to a different ruling set algorithm.\footnote{As   
\cite{BEPS16} aimed for a larger overall runtime, they could use a simpler ruling set algorithm, i.e., one that computes a $(5,O(\log \Delta)$-ruling set of $B_C$ with respect to distances in $C$. However, the value of the domination is not relevant for  this proof. } The next step is to bound the size of $|R_C|$ for each connected component, where our solution differs from \cite{BEPS16} and is significantly shorter. We use that after the first pre-shattering phase each component $C$ of undecided nodes has at most $t \cdot \Delta^4$ nodes by \Cref{lem:shatteringV2} (P2) for $G$, with high probability. We can use this to bound the number of possible subsets $U \subseteq V(C)$ to bound the size of $R_C$ by a  union bound:

\begin{lemma}\label{lem:secondphaseshattering}
    Let $G=(V,E)$ be any graph with $n$ nodes and maximum degree $\Delta$. Let $C$ be any subgraph of $G$ with at most $t \cdot \Delta^4$ nodes and maximum degree at most $\Delta$. Run $\Theta(\log \Delta)$ steps of IndependentSet on $C$. Let $\IS_C \subseteq V(C)$ be the computed independent set and $B_C = V(C) \setminus (\IS_C \cup N(\IS_C))$ be the undecided nodes. There is no $U \subseteq B_C$ such that $U$ is 5-independent in $C$ and $|U| \ge t$, with high probability in $n$.
\end{lemma}
\begin{proof} 
Recall that $t=\log_{\Delta} n$.
    Let $U \subseteq C$ be any 5-independent set in $C$. 
    By \Cref{lem:5indepUndecided} for the graph $C$, the probability of $U$ remaining after $\Theta(c\log \Delta)$ rounds is at most $\Delta^{-c|U|}$, where $c$ is a large-enough constant. 
    We use the fact that the first pre-shattering phase significantly reduces the total number of nodes. This reduces the number of choices for subsets of size $t$ to $\binom{t \cdot \Delta^4}{t} \leq\left(\frac{e \cdot t \cdot \Delta^4}{t}\right)^t\leq \Delta^{5t}$~. 
    By union bound over the different sets, the probability that any $5$-independent $U$, with $|U| = t$ is contained in the undecided nodes is at most $\Delta^{5t}\cdot \Delta^{-c|U|}=\Delta^{5t-ct}=n^{5-c}$. 
\end{proof}
By \Cref{lem:secondphaseshattering}, the size of $R_C$ is less than $t$, with high probability in $n$. 
Form a ball graph $\mathcal{B}$, where each $v \in V(C)$ joins the ball of the closest ruling set node, with respect to distances in $C$. 
The ball graph has an edge between each pair of balls that are adjacent in $G$. 
$\mathcal{B}$ has at most $t$ nodes, and it is contained in the small component. 
Next, a network decomposition of the ball graph is computed. 
This provides a network decomposition of $\mathcal{C}$ (see \Cref{claim:ndBallGraph}). 
The rest of the MIS algorithm works as in \cite{ghaffari16_MIS} or \cite{Gha19}, depending on the model used.

\subsubsection{Approach 2: One pre-shattering phase}
\label{sec:onePhaseShattering}

As an alternative to the previous section, we present another approach that works without rerunning the IndependentSet algorithm on each connected component of undecided nodes. After the pre-shattering phase on $G$, we compute a ruling set of the undecided nodes with respect to distances in $G$. We show that specific balls formed around ruling set nodes are well-connected, which can be used to bound the number of balls in a connected component of the ball graph.

Run $\Theta(s\log \Delta)$ steps of IndependentSet on $G$, with $s=8$. Let $\IS \subseteq V$ be the computed independent set and let $B = V \setminus (\IS \cup N(\IS))$ be the set of undecided nodes. So far, this is essentially what was done in the previous approach. Next, we use \Cref{lem:GhaRuling} to compute a $(5, O(\log \log n))$-ruling set $R$ of $B$ with respect to distances in $G$. This uses the ruling set algorithm of \cite[Lemma 2.2]{Gha19}. Additionally, it partitions $B$ into a set of disjoint balls $\{\ball(v) \subseteq B: v \in R\} $ around the ruling set, such that each ball is an $8$-connected set in $G$: 

\begin{claim} \label{lem:GhaRuling}
    A $(k+1, O(k^2 \log\log n))$-ruling set $R$ of $B$ can be computed in $O(k^2 \log \log n)$ rounds, with high probability, together with a partition of $B$ into a set of disjoint clusters $\{\ball(v) \subseteq B \mid v \in R\}$. For each $v \in R$, $\ball(v)$ is $2k$-connected and has weak diameter $O(k^2 \log \log n)$. For each $v \in R$, there is a Steiner tree $T_v$ with $\ball(v)$ as the terminal nodes. Each edge in $E$ is in at most $O(k\log \log n)$ trees.
\end{claim}
\noindent\textit{Proof sketch.}
    The ruling set is computed with the algorithm of \cite[Lemma 2.2]{Gha19}. 
    Please see the original proof for correctness of the ruling set computation. 
    We show how the algorithm can be used to partition $B$ into $2k$-connected balls around the ruling set nodes. 
    The ruling set algorithm starts with $R:=B$ as the candidate ruling set. 
    For each $v \in R$, we maintain a set of nodes $\ball(v) \subseteq B$, initialized as $\ball(v) = \{v\}$ for each $v \in R$. 
    The Steiner tree of $v \in R$ is initialized as only containing $v$. 
    The ruling set algorithm consists of two parts, where nodes are gradually removed from $R$. 
    The \textit{coarse-grained} part is a randomized sampling process to sparsify $R$, consisting of $O(\log\log n)$ phases. 
    In the second, \textit{fine-grained} part, the remaining $R$ is made $(k+1)$-independent in $k$ \textit{epochs}. 
    In the $i$th epoch, some nodes are removed from $R$ to make it $(i+1)$-independent. 
    This is done by first computing a distance-$i$ $O(\log^5 n)$-coloring of $R$. 
    The coloring is used to compute a $(i+1,O(i\log\log n))$-ruling set of $R$ with the algorithm of \cite{awerbuch1989network} (see \Cref{thm:detRulingOld}). 
    Each ruling set computation consists of $O(\log \log n)$ phases.
 
    We show how to update the ball and Steiner tree for each remaining $u \in R$ after each phase, for both the coarse-grained and fine-grained parts. 
    In both parts, whenever a node $v$ is removed from $R$, we can point to a unique remaining node $u \in R$ that \emph{knocked} $v$ \emph{out}, i.e., that  caused the removal of $v$. 
    When this happens, the ball of $v$  is added to the ball of $u$. 
    In the coarse and fine-grained parts, the distance from $v$ to the node that knocked it out is at most $2k$ and $k$, respectively. 
    Furthermore, for any $v'$ on the shortest path between $v$ and $u$ in $G$, it can be guaranteed that in case $v'$ is also knocked out, it is knocked out by the same node $u$. 
    To update the Steiner tree of $u$, the tree $T_v$, as well as a shortest path between $v$ and $u$ is added to $T_u$.
    
    A simple induction over the phases of the two parts shows that the ball of each remaining node is $2k$-connected. 
    Whenever a node is removed, the shortest path to a remaining node is added to the Steiner tree of the remaining node. 
    The length of this path is at most $2k$. 
    Over a total of $O(\log \log n + k \log \log n)$ phases, this leads to a Steiner tree diameter of $O(k^2 \log \log n)$. 
    In any phase, for any edge $e \in E$, there is at most one $u \in R$ that knocks out some other nodes, such that $e$ is on a shortest path between $u$ and the removed nodes. Hence, $e$ is added to at most $O(k \log\log n)$ trees over all phases. 
    \QED
\vspace{3mm}

In general, \Cref{lem:GhaRuling} shows that specific balls formed around a $(k+1,O(k^2 \log \log n))$-ruling set computed with \cite[Lemma 2.2]{Gha19} are $2k$-connected in $G$. Similarly, it can be shown that specific balls formed around a  $(k+1, k\log n)$-ruling set computed with the algorithm of \cite{awerbuch1989network} are $k$-connected \cite[journal version, Theorem 2.4]{BEPS16}. 

\Cref{lem:GhaRuling} gives a partition of $B$ into balls $\ball(v) \subseteq B, v \in R$. Each ball is $8$-connected in $G$. For each $v \in R$, a Steiner tree $T_v$ with $\ball(v)$ as the terminal nodes is given. We form a ball graph $\mathcal{B}$, with $R$ as the nodes. In the ball graph, there is an edge between $v, w \in R$ if the corresponding balls are adjacent in $G$, i.e. $\exists v' \in \ball(v), \exists w' \in \ball(w): (v',w') \in E$. The rest of the MIS algorithm is executed on each connected component of the ball graph in parallel. Let $\mathcal{C}$ be a connected component of $\mathcal{B}$. Let $R_\mathcal{C}$ be the corresponding ruling set nodes and let $B_\mathcal{C} := \cup_{v \in R_\mathcal{C}} \ball(v)$ be the set of all undecided nodes contained in the balls of $R_\mathcal{C}$.

\begin{claim}\label{claim:Bc8Conn}
    $B_\mathcal{C}$ is 8-connected in $G$.
\end{claim}
\begin{proof}
    $B_\mathcal{C} := \cup_{v \in R_\mathcal{C}} \ball(v)$, where each $\ball(v)$ is 8-connected. As $\mathcal{C}$ is a connected component of the ball graph, for any $v \in R_\mathcal{C}$, there exists $w \in R_\mathcal{C}, v \neq w$ such that $\dist_G(\ball(v), \ball(w)) =1$ (unless $\mathcal{C}$ only contains one node). It follows that, for any $S \subset B_\mathcal{C}, \dist_G(S, B_\mathcal{C} \setminus S) \le \max(1,8) =~8$. 
\end{proof}

We prove that the number of balls in $\mathcal{C}$ is at most $t = \log_\Delta n$. 
We bound $|V(\mathcal{C})| = |R_\mathcal{C}|$ with \Cref{lem:shatteringV2} (P1), by showing that $R_\mathcal{C}$ is a subset of some 5-independent, $16$-connected set $R'_\mathcal{C} \supseteq R_\mathcal{C}$. 
The set $R_\mathcal{C}$ is 5-independent, but the nodes may still be far apart from each other. 
For analysis, construct another set $R'_\mathcal{C} \supseteq R_\mathcal{C}$, such that $R'_\mathcal{C}$ is a $(5,4)$-ruling set of $B_\mathcal{C}$, with respect to distances in $G$. 
By \Cref{lem:54rulingConnected}, $R'_\mathcal{C}$ is 16-connected in $G$, where we use the fact that $B_\mathcal{C}$ is an 8-connected set in $G$. 
Recall that we run $\Theta(s \log \Delta)$ steps of IndependentSet on $G$, with $s=8$. 
Hence, by \Cref{lem:shatteringV2} (P1), with high probability, there are no 5-independent, 16-connected sets $U \subseteq B$ with $|U| \ge t$.  
If no such set exists, the size of $R_\mathcal{C} \subseteq R'_\mathcal{C}$ is at most $t$, proving the claim. 
We can also bound the total number of undecided nodes in the connected component $\mathcal{C}$, which is useful in \congest. 
The set $B_\mathcal{C}$ is 8-connected. \Cref{lem:shatteringV2} (P2) states that the size of any 8-connected subset of $B$ is at most $O(t \cdot \Delta^4)$.  

Continuing in each connected component $\mathcal{C}$ of the ball graph, we compute a network decomposition of $\mathcal{C}$. 
A network decomposition of $\mathcal{C}$ gives a network decomposition of the underlying set of nodes $B_\mathcal{C}$ (see \Cref{claim:ndBallGraph}). 
The distance from one connected components of the ball graph to another is at least 2 in $G$, by definition of the ball graph. 
This makes it possible to solve the rest of the problem independently for each connected component of the ball graph. 
The final MIS is computed by iterating through the colors of the network decomposition, and solving each cluster in time proportional to the cluster diameter (in \congest this requires more work).

\subsection{The approach in  \cite[arXiv]{BEPS16} (and presumably also  in \cite{ghaffari16_MIS,Gha19})}
In this section, we discuss how our approaches and proofs related to algorithms that use the shattering framework for symmetry breaking problems.
We have discussed earlier that the approach in~\cite[journal]{BEPS16} gives correct algorithms with very involved, but correct proofs for properties required by the shattering framework.
However, many follow up works build on top of the flawed earlier approach in~\cite[arXiv]{BEPS16}.
While our first approach provides a different proof for the troublesome parts (see \Cref{sec:twoPhaseShattering}), we believe that it is beneficial for the community to underline these technical subtleties, simply as the shattering technique has grown into an essential technique for developing randomized algorithms. 

Similar to our first approach they use two-phases of pre-shattering\footnote{Of course our work is inspired by theirs and not the other way around.}  where the first phase produces small components and the second phase, executed on all small components independently and in parallel splits each of them into (many) tiny components.
In order to explain the issue in their work let us fix some notation. Consider one small component $C$ and let $B_C$ be the nodes of $C$ that are still undecided after the second pre-shattering phase, i.e., the nodes in the tiny components. In order to profit from the smallness of the components, they compute\footnote{In fact all three works use different ruling set algorithms. One contribution of \cite{Gha19} is actually to come up with a $(5,O(\log\log n)$-ruling set algorithm that works in the \CONGEST model and can be used in this situation.} a $(5,h)$-ruling set $R_C$ for some non-constant $h$. 
In order to bound the size of $R_C$, they want to use \Cref{lem:shatteringV2} (P1), with $s=1$, for the graph $C$ with the randomness of the second pre-shattering phase.  While (P1) indeed guarantees that there are no 5-independent, 9-connected sets $U \subseteq B_C$ with $|U| \ge t$, with high probability, the constructed $(5,h)$-ruling set $R_C$ is not necessarily $9$-connected. Therefore, the authors use a greedy procedure that adds undecided nodes to the ruling set until it becomes $9$-connected (and while always ensuring that it remains $5$-independent). In fact, the greedy procedure should return some $(5,4)$-ruling set $R'_C \supseteq R_C$ of $B_C$, w.r.t. distances in $C$. Note that this set is only defined for the purposes of analysis. However, such a construction cannot always make the set  $9$-connected. 
If the tiny components are far from each other (in $C$) it is impossible to add further undecided nodes (nodes in $B_C$) to the set to obtain a sufficiently connected ruling set. Thus \Cref{lem:shatteringV2} cannot be applied. In fact, the issue in \cite[arxiv v3]{BEPS16} work appears on page 19 in Step~3 and 4 where this greedy construction is performed. Here also nodes of $C$ that are already decided (not contained in a tiny component) are added to the ruling set to make it $9$-connected.  
But then the probabilistic analysis of the second pre-shattering phase does not bound the existence of $R'_C$ as it only bounds the probability that \emph{all} nodes of $5$-independent sets of nodes remain undecided in the second pre-shattering phase.  Also, one cannot  rely on a probabilistic analysis of the first pre-shattering phase, as that would require the distance between nodes in $R'$ to be measured in $G$ (instead of in the small component). 

\smallskip

Enlarging this ruling set greedily is actually an integral part in all known proofs for the first pre-shattering phase, and in particular for proving $(P2)$. The crucial difference in these proofs is that \emph{the prover} can create the ruling set greedily such that it becomes $9$-connected. This becomes straightforward, as in this setting one can focus on a single connected component of undecided nodes. 

\paragraph{The solution of \cite[journal]{BEPS16}.} Recall, when using the smallness of the components algorithmically, it is more difficult to focus on a single connected component. When computing a ruling set (in a black-box manner) with some algorithm and forming balls around each ruler by assigning nodes to the closest ruling set node, it may happen that a node is assigned to a ruler that actually lives in a different small component than its own. Hence, the journal version of \cite{BEPS16} has a different approach. It also has the two-phase pre-shattering approach, but uses the internals of the ruling set algorithm to assign nodes to ruling set nodes. 
This ensures that each ball remains $4$-connected. 
Still it is not ruled out that  a node can be assigned to a ruling set node contained in a different component. Thus, their solution restricts the analysis to a connected component of the constructed ball graph, instead of focusing on a connected component of undecided nodes.  If the balls are created as described, this setting meets the requirement to apply (P1). Our second approach shows that the second pre-shattering phase is not necessary in this approach. Hence, our second approach is an alternative way (besides our first approach) to simplify the framework. 

\section{Randomized Symmetry Breaking on Power Graphs (Thm.~\ref{thm:MISPower}, Cor.~\ref{cor:kkBetaRulingRand})}
\label{sec:randomizedPowerGraphs}
In this section we provide our randomized algorithms for computing maximal independent sets and ruling set of $G^k$, i.e., we prove Theorem~\ref{thm:MISPower} (in \Cref{ssec:MISPower} ) and Corollary~\ref{cor:kkBetaRulingRand} (in \Cref{ssec:randRuling}). In \Cref{sec:lubyGk}, we begin with a short sketch on why Luby's algorithm can be extended to $G^k$ (this does not hold for all versions of the algorithm). 

\subsection{Luby's Algorithm on Power Graphs}
\label{sec:lubyGk}

Luby's algorithm \cite{Luby:1986ub,alon86} extends to compute a MIS of $G^k$ in $O(k \cdot \log n)$ rounds, with high probability. All start as \textit{undecided} and become \textit{decided} once they or one of their neighbors (in $G^k$) joins the independent set. In $G$, one step of the algorithm consists of two rounds. Each undecided node $v$ picks a random number $x_v$ from $[n^c]$~\cite{MJNZ11}, where $c$ is a sufficiently large constant. 
The value is compared with the values of undecided neighbors. If $v$ has the minimum value, $v$ joins the independent set and informs its neighbors. The algorithm computes a MIS of $G$ in $O(\log n)$ rounds, with high probability (see \cite{MJNZ11} for a simple analysis). The algorithm can be simulated on $G^k$ with a $k$-factor slowdown, required for communicating the minimum of the random numbers, as well as alerting distance-$k$ neighbors when joining the MIS. Importantly, the algorithm does not require nodes to know their degree. 

As nodes do not know their degree in $G^k$, it is unclear how to modify the  version of the algorithm to $G^k$ where nodes mark themselves with probabilities depending on their degree.

\subsection{MIS of \texorpdfstring{$G^k$}{G\^k} (Theorem~\ref{thm:MISPower})}
\label{ssec:MISPower}

Similar to the algorithm for \Cref{thm:MISrepaired} our MIS algorithm for $G^k$ is based on the shattering framework. Both, our pre-shattering and post-shattering phase are modification of the  \textit{BeepingMIS} algorithm that was originally introduced in \cite[Section 2.2]{Ghaffari2017}. It can be seen as a communication saving variant of the pre-shattering phase of \cite{ghaffari16_MIS}.

\paragraph{Beeping an MIS (of $G$) \cite[Section 2.2]{Ghaffari2017}:}  All nodes start as \textit{undecided}, and become \textit{decided} once they or one of their neighbors joins the independent set. The algorithm runs in steps (originally called rounds), each consisting of two communication rounds. In the first step, undecided nodes $v \in V$ mark themselves with some probability $x_v$. Marked nodes notify their neighbors with a single bit \textit{beep}. Nodes update their marking probability $x_v$, based on if there was \emph{at least one} marked neighbor. In the second step, marked nodes without marked neighbors join the independent set $\IS$ and notify their neighbors. The nodes joining $\IS$ and their neighbors are deactivated and stop participating in the algorithm. 

BeepingMIS can be used to compute a MIS of $G$ in $O(\log n)$ rounds, with high probability. After $O(\log \Delta)$ rounds, we get  the following \textit{shattering guarantee} from \cite[Theorem 3.3]{BEPS16} and  \cite[Lemma~4.2]{ghaffari16_MIS}.

\begin{lemma}[Shattering]\label{lem:shatteringBeeping}
    Let $s \ge 1$ and $G=(V,E)$ be any graph with at most $n$ nodes and maximum degree at most~$\Delta$. Let  $t=\log_\Delta n$. Run $\Theta(s \log \Delta)$ steps of BeepingMIS on $G$. Let $\IS \subseteq V$ be the computed independent set and $B = V \setminus (\IS \cup N(\IS))$ be the undecided nodes. With high probability in $n$, 
    \begin{itemize}
        \item [(P1)] There is no 5-independent $(8+s)$-connected $U \subseteq B$ s.t. $|U| \ge t$.
        \item [(P2)] All s-connected sets $C \subseteq B$ have at most $O(t \cdot \Delta^4)$ nodes.
    \end{itemize}
\end{lemma}
\begin{proof}
For each node $v$, the probability that $v$ becomes decided after $\Theta(\log d(v) + \log 1/\epsilon)$ is at least $1-\epsilon$ \cite[Theorem 2.1]{Ghaffari2017}. Moreover, one can carefully trace the probabilities---decisions of nodes further than $2$ hops apart are independent when restricted to a single step---and obtain that for any $5$-independent set $U \subseteq V$, the probability that all nodes in $U$ remain undecided after $\Theta(c \log \Delta)$ steps is at most $\Delta^{-c|U|}$.   This provides a version of \Cref{lem:5indepUndecided}. With that result, the remaining proof is identical to the proof of  \Cref{lem:shatteringV2}.  
\end{proof}

\paragraph{High level overview for computing an MIS of $G^k$.}
We run $O(\log \Delta^k)$ steps of BeepingMIS on $G^k$ (where $\Delta^k$ is an upper bound on $\Delta(G^k)$). After that, $G^k$ is shattered into \textit{small components} of undecided nodes. 

The high level idea of the post-shattering phase is similar to the algorithm for $G$. We first construct a ball graph by computing a suitable ruling set of the power graph of the components A network decomposition of the ball graph is computed, which then yields a network decomposition with few color classes and small cluster radius of the small components. 
The main difficulty lies in processing a single cluster (independently from all other clusters). To complete the solution on one cluster we run $O(\log n)$ instances of the Beeping MIS algorithm (for $G^k$) in parallel. We show that w.h.p.\ (in $n$) one of the instances succeeds and all nodes of the cluster can agree on such a winning instance in time that is essentially proportional to the cluster diameter. Nodes of the cluster take the output from that instance and we proceed.

We continue with the details of this process and begin with the necessary tools to simulate BeepingMIS on $G^k$. Afterwards we present several tools for working with ball graphs in the $G^k$ setting. The proof of \Cref{thm:MISPower} that uses all these primitives follows at the end of the section.

\paragraph{Simulating BeepingMIS  on $G^k$.} To do this, the beeps must be accompanied with IDs, for beeping nodes not to confuse their own beep with the beep of another node. Note that this difficulty does not appear when when running the algorithm for $G$. As shown next (\Cref{lem:simBeeping}), BeepingMIS can be simulated on $G^k$ with a slowdown factor of $k$. As we need to run instances of BeepingMIS in parallel in the post-shattering phase, we cannot afford to use $O(\log n)$ bits for the accompanying IDs. We can make use of a distance-$k$ coloring to reduce the bandwidth used by one instance:

\begin{lemma}[Simulating beeping on $G^k$]
    \label{lem:simBeeping}
    Let $k \ge 1$ and $S \subseteq V$ be any set of nodes. Each $v \in V$ can learn if there exists some $w \in N^k(v,S)$ (where $w \neq v$) in $\widetilde{O}(k \cdot \lceil a / \bandwidth\rceil)$ rounds, using $\bandwidth$-bit messages, given that each node has an $a$-bit identifier unique up to distance $k$.
\end{lemma}
\begin{proof}
    Each $x \in S$ \textit{beeps} by sending a tuple $(\ID(x), $k$)$, where the second element is a distance left counter, initialized as $k$ and decreased every time the tuple is forwarded. For $k$ steps, each $v \in V$ forwards to each neighbor $w \in N(v)$ an arbitrary subset of at most two incoming tuples with distinct identifiers, with the maximum of the distances left. After $k$ steps, any $v \in V$ with $N^k(v,S) \neq \emptyset$ receives at least one beep from some $w \in N^k(v,S)$, regardless of whether $v \in S$. Each step can be implemented in $\lceil(a+\log k)/\bandwidth\rceil$ rounds. The total time complexity is $\widetilde{O}(k \cdot  \lceil a / \bandwidth\rceil)$.
\end{proof}

\paragraph{Distance-k ball graph.}
Let $B \subseteq V$ be a set of (undecided) nodes and let $R \subseteq B$ be some ruling set of $B$.
Assume we are given a partition $\{\ball(v) \subseteq B: v \in R\}$ of the undecided nodes.  
Recall that a \textit{ball graph for $\{\ball(v) \subseteq B: v \in R\}$} is a virtual graph with nodes $R$. 
There is an edge between $v,w \in R$ if the corresponding balls $\ball(v)$ and $\ball(w)$ are adjacent in $G$. 

In the post-shattering phase, a network decomposition is computed for the remaining undecided nodes. 
In $G$, a network decomposition of the ball graph for $\{\ball(v) \subseteq B: v \in R\}$ induces a network decomposition of $G$ (see \Cref{claim:ndBallGraph}). 
This does not immediately work for graph powers. 
Instead, we compute a network decomposition of a \textit{distance-$k$ ball graph for $\{\ball(v) \subseteq B: v \in R\}$}.  
A graph $\mathcal{B}$ with nodes $R$ is a distance-$k$ ball graph if, for any $v,w \in R$,  $\dist_G(\ball(v),\ball(w)) \le k$ implies that $\dist_\mathcal{B}(v,w) \le k$.
Next, we show how a distance-$k$ ball graph can be formed in the communication network $G$:

\begin{lemma}[Forming distance-$k$ ball graph] \label{lem:ballgraph}
    Let $R \subseteq B \subseteq V$ and suppose we are given a partition $\{\ball(v) \subseteq B: v \in R\}$ of $B$. Assume that for each $v \in R$, there is a Steiner tree $T_v$ with weak diameter $O(D)$, with $\ball(v)$ as the terminal nodes, such that any edge in $E$ is in at most $\tau$ trees. There is an $O(k)$-round deterministic \congest algorithm that forms $\{\ball^+(v) \subseteq V : v \in R\}$, such that
    \begin{itemize}
        \item $\ball^+(v) \subseteq V$ and $\ball^+(v) \supseteq \ball(v)$ for all $v \in R$,
        \item $\ball^+(v) \cap \ball^+(w) = \emptyset$ for all $v \neq w \in R$,
        \item The ball graph $\mathcal{B}$ for $\{\ball^+(v) : v \in R\}$ is a distance-$k$ ball graph for $\{\ball(v) : v \in R\}$. 
        \item For each $v \in R$, there is a Steiner tree $T_v^+$ with weak diameter $O(D+k)$ and $\ball^+(v)$ as the terminal nodes. Any edge in $E$ is in at most $\tau+1$ trees. 
    \end{itemize}
\end{lemma}
\begin{proof}
    For each $v \in R$, let $\border(v) \subseteq N^k(v)$ be a set of nodes around $\ball(v)$. 
    Borders only consist of nodes in $V \setminus B$, and $\border(v) \cap \border(w) = \emptyset$ for all $v \neq w \in R$. 
    See \Cref{fig:GkBallgraph}. 
    Eventually, we set $\ball^+(v) := \ball(v) \cup \border(v)$. 

    Initially, $\border(v) := \emptyset$ for all $v \in R$. 
    For each $v \in R$ in parallel, run a breadth first search from all nodes of $\ball(v)$ for at most $k$ hops in $G$. 
    $\ID(v)$ is used as an identifier for the searches originating from $\ball(v)$. 
    Let $z \in V \setminus B$. When $z$ receives one or more searches for the first time, it accepts the one with the smallest identifier (suppose this is $v$) and joins $\border(v)$. 
    Afterwards, $z$ forwards the accepted search to its other neighbors, for at most $k$ hops in total. 
    Nodes in $B$ do not join borders or forward searches. 

    Let $\ball^+(v) := \ball(v) \cup \border(v)$ for all $v \in R$. 
    A Steiner tree $T_v^+$ is formed by adding the paths of the BFS token to $T_v$. 
    Each edge is added to at most one tree. 
    Let $\mathcal{B}$ be the ball graph for $\{\ball^+(v) : v \in R\}$. 
    For correctness, consider any $v,w \in R$ with $\dist_G(\ball(v), \ball(w)) \le k$ (see \Cref{fig:GkBallgraph}). 
    We prove that $\dist_{\mathcal{B}^+}(v,w) \le k$, i.e., $v$ and $w$ are adjacent in $\mathcal{B}^k$.
    Let $(v', x_1, \dots, x_s, w')$ be a shortest path between $\ball(v)$ and $\ball(w)$ in $G$, where $s \le k-1$, $v' \in \ball(v)$ and $w' \in \ball(w)$. 
    The BFS from $\ball(v)$ and $\ball(w)$ travels for $k$ hops in each direction, until it reaches another ball or border. 
    Hence, each $x_i$, $1 \le i \le s$ belongs to $\ball^+(y_i)$ for some $y_i \in R$ (possibly $y_i=v$ or $y_i=w$).
    By the number of nodes on the path, the total number of distinct balls $\ball^+(y_i)$ is at most $k-1$. 
    These nodes form a path from $v$ to $w$ in the ball graph $\mathcal{B}^+$, with length at most $k$. 
\end{proof}

\begin{figure}[h]
    \centering
    \includegraphics[width=0.75\textwidth]{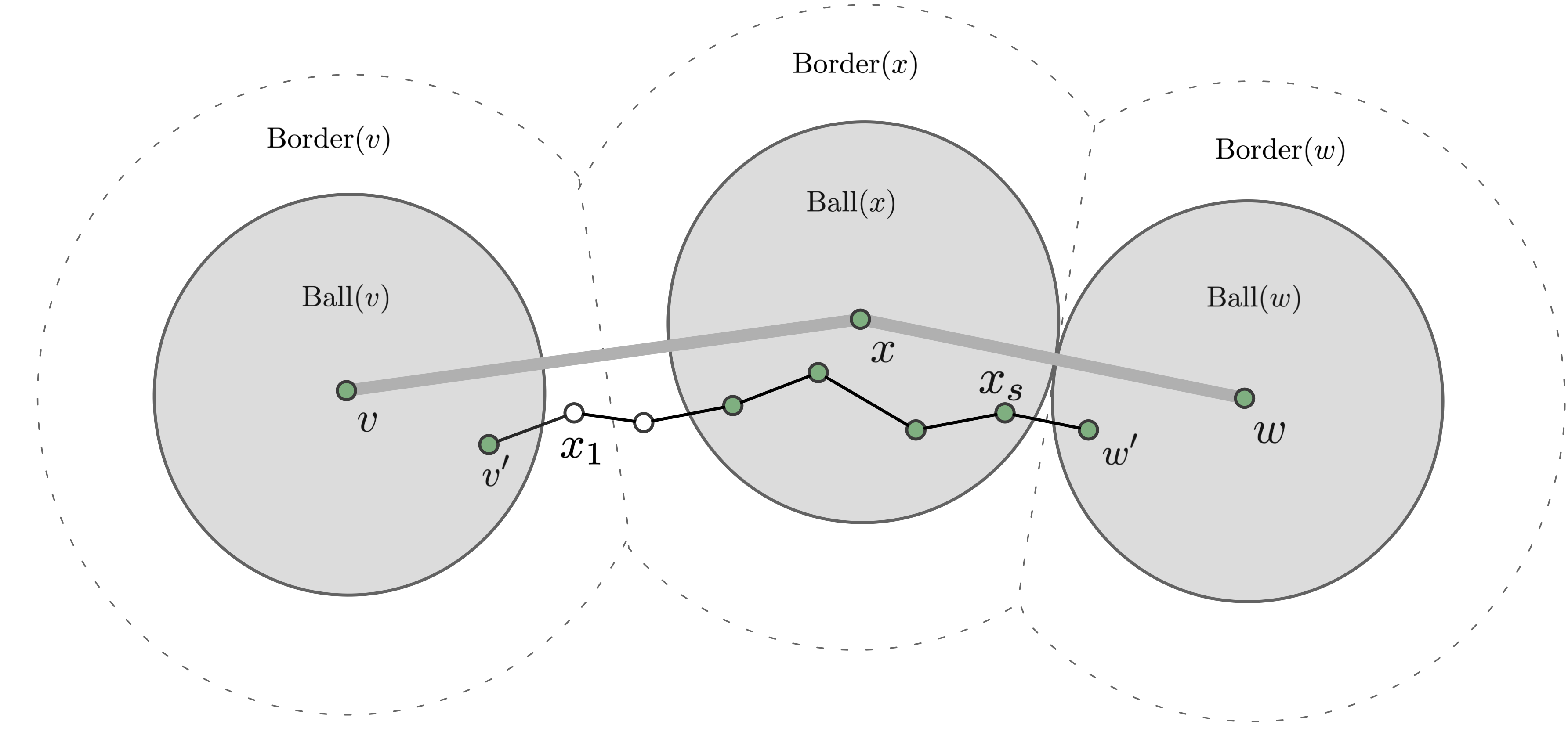}
    \vspace{-3mm}
    \caption{Distance-$k$ ball graph $\mathcal{B}$ with three vertices $v,x,w$ and parts of the underlying graph $G$. The edges of $\mathcal{B}$ are shown with thicker lines. Nodes $v$ and $x$ are adjacent in $\mathcal{B}$ because, while $\ball(v)$ and $\ball(x)$ do not share an edge in $G$, their borders do. }
    \label{fig:GkBallgraph}
\end{figure}

Let $\mathcal{B}$ be a distance-$k$ ball graph for $\{\ball(v) \subseteq R : v \in R\}$, formed with \Cref{lem:ballgraph}. We compute a network decomposition of $\mathcal{B}$ with separation $k+1$, which gives a network decomposition of $G^k$ for the undecided nodes: 
\begin{claim}\label{claim:ndGkBallGraph}
    A network decomposition of $\mathcal{B}^k$ can be transformed to a network decomposition of $G^k$ for nodes in $B$. The cluster diameter increases by a factor proportional to $k$ and the diameter of the balls.
\end{claim}
\begin{proof}
    A network decomposition of $G^k$ is formed in the natural way: for each $v \in R$, nodes in $\ball(v)$ join the cluster of $v$. 
    
    The cluster diameter in $G$ increases by an $O(k+D)$-factor, where $D$ is the weak diameter of balls in $\{\ball(v) \subseteq B: v \in R\}$, because nodes in $\mathcal{B}$ (as formed in \Cref{lem:ballgraph}) are actually balls of diameter $O(k+D)$. 
    We prove that the clusters are properly separated in $G^k$. 
    Suppose for a contradiction that there are two distinct clusters of the same color that are adjacent in $G^k$. 
    Let $v, w$ be some nodes from the two clusters such that $\dist_G(v,w) \le k$. 
    The nodes are in two distinct balls $\ball(v')$ and $\ball(w')$ (otherwise $v$ and $w$ join the same cluster).
    By definition of $\mathcal{B}$, we have $\dist_\mathcal{B}(v', w') \le k$. 
    Since we computed a network decomposition of $\mathcal{B}^k$, $v'$ and $w'$ must belong to the same cluster in $\mathcal{B}$. 
    This is a contradiction of $v$ and $w$ belonging to different clusters in $G$.  
\end{proof}

We are ready to prove the main result of this section. 
To compute a maximal independent set of $G^k$, we use the shattering framework, with BeepingMIS as the base algorithm.
The reader is advised to take a look at \Cref{sec:revisedShattering}, containing a more detailed description of shattering in $G$.  
We use the one-phase pre-shattering approach (\Cref{sec:onePhaseShattering}). 
In the post-shattering phase, multiple instances of BeepingMIS are run in parallel. 
To do this efficiently in $G^k$, we assign the remaining undecided shorter IDs. 
Combined with the simulation tools of \Cref{lem:simBeeping}, this achieves the same runtime as for MIS of $G$ (up to $k$ factors in the runtime). 

\thmMISPower*

\noindent The runtime can also be given as $\widetilde{O}(k \log \Delta(G^k) \cdot \log \log n + k^4 \log^5 \log n)$.

\begin{proof}[Proof of \Cref{thm:MISPower}:]
    
    \textbf{Pre-shattering.} 
    We simulate the \textit{BeepingMIS} algorithm of \cite{Ghaffari2017} on $G^k$ with communication network $G$. 
    Run the algorithm for $O(s\log \Delta^k)$ steps, where $s=8$ and $\Delta^k$ is an upper bound on $\Delta(G^k)$. 
    Each step can be simulated on $G^k$ with \Cref{lem:simBeeping},  with a factor-$k$ slowdown, using the original $O(\log n)$-bit IDs and the full $\Theta(\log n)$-bit communication bandwidth. The total runtime of the pre-shattering phase is $O(k^2 \log \Delta)$. The BeepingMIS algorithm computes an independent set $\IS \subseteq V$ of $G^k$, which is part of the final result. 
    
    \textbf{Post-shattering: ruling set.} 
    Let $B = V \setminus (\IS \cup N^k(\IS))$ be the nodes that remain undecided after the $\Theta(\log \Delta^k)$ steps of BeepingMIS on $G^k$. 
    We compute a ruling set $R$ of $B$, such that nodes in $R$ have independent executions in the pre-shattering phase. 
    For this, nodes in $R$ must be at least 5 hops apart in $G^k$, or equivalently at distance at least $4k+1$ in $G$ (for convenience, we use a larger independence of $5k+1$). 
    Compute a $(5k+1, O(k^2 \log \log n)$-ruling set $R$ of $B$, with respect to distances in $G$,  using the ruling set algorithm of \cite[Lemma 2.2]{Gha19}.
    This takes $O(k^2 \log \log n)$ rounds and succeeds with high probability. 
    Running in parallel with the ruling set algorithm, we use \Cref{lem:GhaRuling} to partition $B$ into a set $\{\ball(v) \subseteq B : v \in R\}$ of disjoint balls around the ruling set nodes. 
    For each $v \in R$, $\ball(v)$ is an $10k$-connected set in $G$. 
    This implies that each $\ball(v)$ is 10-connected in $G^k$. 
    For each $v \in R$, there is a Steiner tree $T_v$ in $G$, with $\ball(v)$ as the terminal nodes and diameter $O(k^2 \log \log n)$, and any edge is in at most $O(k \log \log n)$ Steiner trees. 
    
    \textbf{Ball graph.} 
    We form a distance-$k$ \textit{ball graph} $\mathcal{B}$ for $\{\ball(v) \subseteq B : v \in R\}$, using \Cref{lem:ballgraph} in $O(k)$ rounds. 
    The set of nodes in $\mathcal{B}$ is $R$. 
    Each node $v \in R$ corresponds to a set $\ball^+(v) \subseteq V$, where $\ball^+(v)$ is a superset of $\ball(v)$, and $\ball^+(v), \ball^+(w)$ are disjoint for any $v \neq w \in R$.
    For each $v \in R$, let $T_v$ be the Steiner tree for $\ball^+(v)$, combining the original Steiner tree for $\ball(v)$ and the parts added in \Cref{lem:ballgraph}.
    Any edge in $G$ is in at most $O(k \log \log n)$ trees (\Cref{lem:ballgraph} does not increase congestion). 
    $\mathcal{B}$ is defined as the ball graph for $\{\ball^+(v) : v \in R\}$, i.e., $v,w \in R$ are adjacent in $\mathcal{B}$ if $\ball^+(v)$ and $\ball^+(w)$ are adjacent in $G$. 
    By \Cref{lem:ballgraph}, for any $v \neq w \in R$ it holds that if $\dist_G(\ball(v),\ball(w)) \le k$, then $\dist_\mathcal{B}(v,w) \le k$. 
    Hence, any two distinct connected components of the ball graph are at distance at least $k+1$ from each other in $G$. 
    The rest of the MIS algorithm is executed independently on each connected component of the ball graph. 
    
    Fix a connected component $\mathcal{C}$ of $\mathcal{B}$. 
    We bound the number of balls in $\mathcal{C}$, as well as the total number of undecided nodes contained in the balls of $\mathcal{C}$.
    Let $R_\mathcal{C} = V(\mathcal{C}) \subseteq R$ be the nodes in $\mathcal{C}$, and let $B_\mathcal{C} := \cup_{v \in R_\mathcal{C}} \ball(v)$ be the set of undecided nodes contained in the balls of nodes in $R_\mathcal{C}$. 
    The set $B_\mathcal{C}$ is $10k$-connected in $G$, as each ball is $10k$-connected and adjacent balls are at most $2k$ hops apart in $G$ (c.f. \Cref{claim:Bc8Conn}). 
    Equivalently, $B_\mathcal{C}$ is 10-connected in $G^k$. 
    Start by bounding the number of nodes $|R_\mathcal{C}|$.
    \Cref{lem:shatteringBeeping} for $G^k$, with $s=10$, states that, with high probability, in $G^k$ there does not exist a 5-independent, 18-connected set $U \subseteq R$ with $|U| \ge t=\log_\Delta n$ (where $\Delta$ is a lower bound on $\Delta(G^k)$). 
    This is not immediately usable, as the set $R_\mathcal{C}$ is 5-independent in $G^k$, but it is not necessarily 18-connected. 
    For the purposes of analysis, we construct another set $R'_\mathcal{C} \supseteq R_\mathcal{C}$, where $R'_\mathcal{C} \subseteq B_\mathcal{C}$.
    $R'_\mathcal{C}$ is any greedily chosen $(5,4)$-ruling set of $B_\mathcal{C}$, with respect to distances in $G^k$. 
    By \Cref{lem:54rulingConnected}, $R'_\mathcal{C}$ is 5-independent and 18-connected in $G^k$, where we use the fact that $B_\mathcal{C}$ is 10-connected in $G^k$. 
    Given that, with high probability, the size of any 5-independent and 18-connected set in $G^k$ is less than $t$, we get $|R_\mathcal{C}| < t$. 
    For the second bound, recall that $B_\mathcal{C}$ is a 10-connected set in $G^k$.
    Given that \Cref{lem:shatteringBeeping} (P1) holds, (P2) (again for $G^k$ and $s=10$) states that $|B_\mathcal{C}|$ is at most $t \cdot \Delta^{4k}$, where $\Delta^{k}$ is an upper bound on $\Delta(G^k)$.  

    \textbf{Network decomposition.} 
    We compute a network decomposition of $\mathcal{C}$. 
    We use \Cref{thm:networkDecomp} to deterministically compute a network decomposition of $\mathcal{C}$ with $\widetilde{O}(\log \log n)$ colors, weak diameter $O(k \cdot \log \log n)$ and separation $2k+1$ between clusters of the same color. 
    The fact that the network decomposition algorithm can be simulated on a ball graph is verified in \Cref{claim:ndAlgoBallGraph}.
    The simulation runs with an $O(k^3 \log^2 \log n)$-factor slowdown, due to the fact that each node is simulated by a Steiner tree with weak diameter $O(k^2 \log \log n)$ and congestion $O(k \log \log n)$. 
    Hence, the runtime is $\widetilde{O}(k^4 \log^5 \log n)$. 
    By \Cref{claim:ndGkBallGraph}, we can extract a network decomposition of $G$ for $B_\mathcal{C}$: each node in $B_\mathcal{C}$ joins the cluster of the ball it belongs in. 
    From now on, the clusters of the network decomposition are considered as sets of nodes in $B_\mathcal{C}$.
    In $G$, the actual diameter of the clusters is $O(k^3 \log^2 \log n)$. 
    For each cluster, there is a Steiner tree, originally in the ball graph $\mathcal{C}$. 
    By replacing each node $v \in R_\mathcal{C}$in the ball graph with the Steiner tree of the corresponding ball $T_v$, we get a Steiner tree for the cluster in $G$. 
    There is no congestion in the Steiner trees produced by \Cref{thm:networkDecomp}, while each edge in $E$ is in at most $O(k \log \log n)$ Steiner trees of balls. 
    Hence, the congestion in the Steiner trees formed for the clusters is at most $O(k \log \log n)$. 

    We process each color of the network decomposition separately. 
    Fix some color $j$ and let $S$ be a cluster of color $j$.  
    Recall that $N:= t \cdot \Delta^{4k}$ is an upper bound on $|B_\mathcal{C}|$. 
    We assign each node $v \in S$ an identifier from $[N]$, such that nodes in a cluster have unique identifiers; these identifiers use $O(\log N)=O(\log \log n+k\log \Delta)$ bits. 
    The IDs are assigned using the Steiner tree of the cluster, in $O(k^3 \log^2 \log n)$ time, by distributing a range of IDs to each subtree, based on the number of nodes in $S$ contained in the subtree. 
    
    \textbf{Final MIS.} Now we are ready to make the independent set maximal for nodes in $S$.
    We run BeepingMIS for $O(\log N) = O(k\log \Delta + \log \log n)$ steps, with $O(\log_N n)$ executions of the algorithm running in parallel. Each execution is allocated $\Theta(\log N)$ bits of communication bandwidth, corresponding to the size of the new IDs.  Hence, the total bandwidth used by the parallel executions is $O(\log N \cdot \log_N n) = O(\log n)$, which fits in a single message. 
    As one step of BeepingMIS in $G^k$ requires $O(k)$ communication rounds in $G$, running the parallel executions for $O(\log N)$ steps takes $O(k \log N) = O(k^2 \log \Delta + k \log \log n)$ rounds in total. 
    Executions of clusters of the same color can be run independently, because messages travel for at most $k$ hops from an undecided node, and the clusters are at least $2k+1$ hops apart. 
    Each execution outputs an MIS of $G^k$ in the cluster with probability at least $1- 1/N^c$, for some constant $c$.  
    The probability that at least one of the executions succeeds is at least $1-1/N^{c \cdot \log_N n} = 1-1/n^c$. 
    To find a successful run, we aggregate in the Steiner tree of the cluster whether each execution was successful or not (an unsuccessful execution is not maximal, i.e. there exists an undecided node), using a single bit indicator for each execution, in $O(k^4 \log^3 \log n)$ rounds. 
    The cluster leader picks one successful run and informs the rest of the cluster. 
    The chosen independent set is added to the output. 
    All nodes in $B_\mathcal{C}$ within $k$ hops of the independent set also become decided, regardless of whether they are in a cluster of color $j$. 
    The runtime for one color class is $O(k^2 \log \Delta + k^4 \log^3 \log n)$. 
    Over all $\widetilde{O}(\log \log n)$ colors this becomes $\widetilde{O}(k^2 \log \Delta \cdot \log \log n + k^4 \log^4 \log n)$. 
    The total runtime, including pre-shattering, computing the ruling set and network decomposition is $\widetilde{O}(k^2 \log \Delta \cdot \log \log n + k^4 \log^5 \log n)$.
\end{proof}

\begin{corollary}\label{cor:MISPowerInduced}
    The MIS algorithm of \Cref{thm:MISPower} computes a MIS of $G^k[Q]$, with high probability, when used with a subset $Q \subseteq V$. The runtime is $\widetilde{O}(k \log \Delta_Q \cdot \log \log n + k^4 \log^5 \log n)$, where $\Delta_Q = \Delta(G^k[Q])$. 
\end{corollary}
\begin{proof}
    Only nodes in $Q$ are allowed to join the independent set. Other nodes also take part in the algorithm, but they are set to the \textit{decided} state at the start of the algorithm. Like other decided nodes, they still take part in forwarding messages between the remaining undecided nodes, as they may be part of edges of $G^k[Q]$. Both the pre-shattering as well as the post-shattering phase can be executed in this setting.
\end{proof}

\subsection{Ruling sets in \texorpdfstring{$G^k$}{G\^k} (Corollary~\ref{cor:kkBetaRulingRand})}
\label{ssec:randRuling}
In this section, we prove the following corollary of \Cref{thm:MISPower}.
\newpage
\thmkkBetaRulingRand*

We can use our MIS of $G^k$ algorithm in combination with the sparsification procedure of \cite{KP12} and \cite{BKP14} to compute ruling sets of power graphs. We start by briefly explaining the algorithm of \cite{KP12} in $G$. Let $H=(V_H, E_H)$ be any graph with maximum degree $\Delta_H$ and let $f\ge 2$ be a parameter. The algorithm of \cite{KP12} samples a set of nodes $Q \subseteq V_H$ in $O(\log_f \Delta_H)$ rounds, such that (1) the maximum degree in $H[Q]$ is $O(f \cdot \log n)$ with high probability, and (2) $Q$ is a dominating set of $V_H$. All communication is done with beeps, sent by sampled nodes to notify their neighbors about being sampled. This can be simulated in $G^k$ in $k$ rounds, using \Cref{lem:simBeeping} (in fact without including IDs, because beeping nodes do not need to listen to other beeps). The algorithm also does not require nodes to know their degree (see algorithm \texttt{Sparsify-GG} of \cite{BKP14}). Hence, we can simulate it in $G^k$ with communication network $G$, with an $O(k)$-factor slowdown. 

%% note that BKP denotes the parameter used in the sth iteration by f_{s-1}
Given $\beta \ge 2$, it is possible to compute a $(2,\beta)$-ruling set by iterating the the algorithm of \cite{KP12} $\beta-1$ times, combined with a MIS computation on the final sampled subgraph \cite{BKP14, Gha19}. The graph is sparsified iteratively, producing a sequence of subsets $V \supseteq Q_1 \supseteq \dots \supseteq Q_{\beta-1}$, where $Q_s$ is the result of the $s$th iteration, computed with the algorithm of \cite{KP12} on $G[Q_{s-1}]$. Let $f_1 > \dots > f_{\beta-1}$ be the parameters used in each iteration. For each $1 \le s \le \beta-1$, it holds that (1) the maximum degree $\Delta_s$ of $G[Q_s]$ is $O(f_i \cdot \log n)$ with high probability, and (2) $Q_s$ is an $s$-dominating set of $G$. The $s$th iteration takes $O(\log_{f_s} \Delta_{s-1})=O(\frac{\log f_{s-1} + \log\log n}{\log f_s})$ rounds. By setting $f_s=2^{(\log \Delta)^{1-s/(\beta-1)}}$, the runtimes of the $\beta-1$ iterations are balanced, taking $O((\beta-1)\cdot \log^{1/(\beta-1)} \Delta)$ rounds in total. The maximum degree of $G[Q_{\beta-1}]$ is $O(f_{\beta-1} \cdot \log n) = O(\log n)$. Now, a maximal independent set algorithm can be used to compute a MIS of $G[Q_{\beta-1}]$, resulting in a $(2,\beta)$-ruling set of $G$.

\begin{proof}[Proof of \Cref{cor:kkBetaRulingRand}]
    We start by sparsifying $G^k$ for $\beta-1$ iterations, using the algorithm of \cite{KP12}. For $1 \le s \le \beta-1$, let $Q_s \subseteq V$ be the result of the $s$th iteration. Initially, $Q_0 := V$. In the $s$th iteration, we set $f_s=2^{(\log \Delta^k)^{1-s/(\beta-1)}}$. We simulate the algorithm of \cite{KP12} on $G^k[Q_{s-1}]$ and let $Q_s$ be the sampled nodes in this iteration. By the analysis in \cite{KP12} applied for $G^k$, it holds that (1) the maximum degree of $G^k[Q_s]$ is $O(f_s \cdot \log n)$ with high probability, and (2) $Q_s$ is a dominating set of $Q_{s-1}$ in $G^k$. This implies that $Q_s$ is a $k$-dominating set of $Q_{s-1}$ in $G$. The algorithm of \cite{KP12} runs in 
    $$O\left(\log_{f_s} \Delta_{s-1}\right) 
    %= O\left(\frac{\log (f_{s-1}\cdot \log n)}{\log f_s}\right) 
    = O\left(\frac{\log f_{s-1}}{\log f_s} + \log\log n \right)
    = O\left((\log \Delta^k)^{1/(\beta-1)} + \log \log n\right)$$
    The simulation in $G^k$ runs with a $k$-factor slowdown.
    
    The result of the last iteration is a set $Q_{\beta-1}$, such that the maximum degree in $G^k[Q_{\beta-1}]$ is 
    $O(f_{\beta-1} \cdot \log n) = O(2^{(\log \Delta^k)^{0}} \cdot \log n) = O(\log n)$. 
    We compute a MIS $\IS$ of $G^k[Q_{\beta-1}]$ with the algorithm of \Cref{thm:MISPower} (see \Cref{cor:MISPowerInduced}). 
    Given that $\Delta \big(G^k[Q_{\beta-1}] \big) = O(\log n)$, computing the MIS takes $\widetilde{O}(k^4 \cdot \log^5 \log n)$ rounds. The result is $(k+1)$-independent in $G$. The domination is $(\beta - 1) \cdot k + k = \beta \cdot k$. Hence, the result is an $(k+1, \beta \cdot k)$-ruling set of $G$. The total runtime is 
    \begin{align*}
        &\widetilde{O}\left( (\beta-1)\cdot k \cdot \left((\log \Delta^k)^{1/(\beta-1)} + \log \log n\right) + k^4 \cdot  \log^5 \log n\right)\\
        &= \widetilde{O}\left( \beta \cdot k^{1 + 1/(\beta-1)} \cdot (\log \Delta)^{1/(\beta-1)} + \beta \cdot k \cdot \log\log n + k^4 \cdot \log^5 \log n\right) & &\qedhere
    \end{align*}
\end{proof}

\clearpage 
\bibliographystyle{alpha}
\bibliography{ref}

\appendix

\section{Network decomposition for power graphs}
\label{app:networkDecomp}
Our algorithms use network decompositions of power graphs as a subroutine. More formally, we require that clusters of the same color class are at least $k+1$ hops apart for some parameter $k$. The respective decompositions can be computed with the algorithm of \cite{MU21}. However, as recent results on network decomposition \cite{GGH22} also work in \CONGEST, are significantly faster than the result of \cite{MU21}, and provide better guarantees, we sketch on how to adapt their algorithm to provide the necessary decompositions.\footnote{The methods of \cite{GGH22} essentially directly provide the required results. Additionally, we have confirmed this fact with one of the authors of \cite{GGH22}.} Recall \Cref{def:networkDecomp} for the definition of a network decomposition.

\begin{theorem}[Network decomposition of $G^k$]\label{thm:networkDecomp}
    Let $k \ge 1$ (potentially a function of $n$). There is a deterministic \congest algorithm that in $\widetilde{O}(k\cdot \log^3 n)$ rounds computes a network decomposition of $G^k$ with $O(\log n \cdot \log \log n)$ colors and weak-diameter $O(k \cdot \log n)$ in $G$. For each color $c$ and each cluster $C \subseteq V$, there is a Steiner tree $T_C$ with radius $O(k \cdot \log n)$ in $G$, with $C$ as the terminal nodes. Each edge of $G$ is in at most one tree of clusters of color $c$. 
\end{theorem}

\begin{lemma}[Low-degree clustering. Adaptation of {\cite[Theorem 4.1]{GGH22}}] \label{lem:lowDegClustering}
    Let $s \ge 2$ (potentially a function of $n$). Let $A \subseteq V$ be a set of living nodes. There is a deterministic \congest algorithm that, in $\widetilde{O}(s \cdot \log^2 n)$ rounds, computes a clustering $\mathcal{C}$, where $C \subseteq A$ for each cluster $C \in \mathcal{C}$, with   
    \begin{enumerate}
        \item weak diameter $O(s \cdot \log n)$ in $G$. For each cluster $C \in \mathcal{C}$, there is a Steiner tree of radius $O(s \cdot \log n)$ in $G$, with $C$ as the terminal nodes. Each edge of $G$ is in at most one tree.
        
        \item $s$-hop degree\footnote{See \cite{GGH22} for the definition of $s$-hop degree of a clustering} of $\mathcal{C}$ is at most $\lceil 100 \log \log n \rceil$, defined with respect to the Steiner trees,
        \item the number of clustered nodes is at least $|A|/2$.
    \end{enumerate}
\end{lemma}

\noindent \textit{Proof sketch.}
    We start by giving an overview of the algorithm of \cite[Theorem 5.1]{GGH22}. 
    Their algorithm consists of two main components, clustering from given delays and computing a suitable delay function. 
    A delay function assigns each $v \in V$ a positive integer delay $\text{del}(v)$ from 1 to $O(s \log n)$. 
    Let $\kappa= \lceil 100\log\log n\rceil$ be a constant (replacing $k$ in \cite{GGH22}, not to be confused with the parameter for the power of $G$ in \Cref{thm:networkDecomp}).
    Together with the separation parameter $s$, a delay function induces a clustering, computed by running a breadth first search from each node after waiting for a number of steps given by the delay. 
    Each $v \in V$ forwards a subset of at most $\kappa$ tokens per step. 
    $v$ is associated with the cluster of the node whose BFS token arrives first (choosing the smallest ID to break ties).  
    After the arrival of the first token, $v$ continues to forward arriving tokens to others  for $2s$ steps, or until a total of $\kappa$ tokens have been forwarded. 
    The total number of tokens received during the interval of $2s$ steps from the arrival of the first token is counted. 
    The set of nodes whose tokens arrived to $v$ during this interval is denoted $\text{frontier}^{2s}(v)$. 
    Each $v \in V$ with $\text{frontier}^{2s}(v) \le \kappa$ joins the cluster of the node whose token arrived first.
    In total, the BFS procedure takes $\widetilde{O}(s \cdot \log n)$ rounds.
    The cluster diameter is bounded by the maximum delay $O(s \log n)$, since after that any node initiates a token itself. 
    The second main component of the algorithm is computing a suitable delay function. The delay function is computed such that a constant fraction of nodes satisfy $\text{frontier}^{2s}(v) \le \kappa$, i.e., a constant fraction of nodes are clustered. The delay function is chosen from an exponential distribution, by derandomizing a pairwise independent coin-flipping procedure step by step. We make the following modifications to the two main components.

    Clustering from given delays: Only living nodes join clusters. \textit{Dead} nodes $V \setminus A$ do not initiate or join clusters. Dead nodes forward the incoming BFS tokens normally. This is necessary, for distances to be measured in the input graph and not just $G[A]$. The produced clusters have weak diameter instead of strong diameter, because there may be dead nodes inside the corresponding strong diameter cluster. The paths taken by the BFS token of the cluster acts as the Steiner tree, which has radius $O(s \log n)$. Note that the distance to a nearest cluster leader is only bounded for living nodes (the distance is at most $O(s \log n)$, because after at most $O(s \log n)$ steps, any living node starts a cluster itself). To guarantee that the algorithm terminates for dead nodes, we stop execution after $5sR + 2s = O(s \log n)$ steps, which is the maximum number of steps required to compute the clusters and frontiers for all living nodes. 
    
    Computing the delay function: Originally, the choice of delay function is derandomized to guarantee that at least $|V|/2$ nodes are clustered. For our application, the potential functions are computed as sums over living nodes, instead of all nodes in the graph. This guarantees that the number of clustered nodes is at least $|A|/2$. 
\QED

\begin{lemma}[From Low-Degree to Isolation. Adaptation of {\cite[Theorem 5.1]{GGH22}}] \label{lem:isolatedClustering}
    Let $s \ge 2$ (potentially a function of $n$). Given a set $A \subseteq V$ of living nodes, and a clustering $\mathcal{C}$ from \Cref{lem:lowDegClustering}, there is a deterministic \congest algorithm that, in $\widetilde{O}(s \cdot \log^2 n)$ rounds, computes a clustering $\mathcal{C}^{\text{out}}$, where $C \subseteq A$ for each cluster $C \in \mathcal{C}^{\text{out}}$, with
    \begin{enumerate}
        \item weak diameter $O(s \cdot \log n)$. For each cluster $C \in \mathcal{C}^\text{out}$, there is a Steiner tree $T_C$ with radius $O(s \cdot \log n)$ in $G$, with $C$ as the terminal nodes. Each edge of $G$ is in at most one tree.
        \item separation $s$
        \item the number of clustered nodes of $A$ is at least $\frac{|\cup \mathcal{C}|}{1000 \log \log n}$
    \end{enumerate} 
\end{lemma}

\begin{proof}
    See proof of \cite[Theorem 5.1]{GGH22}. The only difference to the original statement is that clusters have weak diameter, which does not matter because there is no congestion. Note that the clusters in $\mathcal{C}^\text{out}$ consist of only nodes in $A$, because each $C \in \mathcal{C}^\text{out}$ is a subset of some $C' \in \mathcal{C}$, where $C \subseteq A$ by definition in \Cref{lem:lowDegClustering}.
\end{proof}

\begin{proof}[Proof of \Cref{thm:networkDecomp}]
    Fix $G=(V,E)$. Let $A := V$ be a set of living nodes. The result is obtained in $O(\log n \cdot \log \log n)$ iterations. Each iteration forms a color class of the network decomposition with a clustering procedure combining \Cref{lem:lowDegClustering} and \Cref{lem:isolatedClustering}. 

    Fix some iteration $i$. Start by computing a low-degree clustering $\mathcal{C}$ of the living nodes in $G$ using \Cref{lem:lowDegClustering} with $s=k+1$. The runtime is $\widetilde{O}(k \cdot \log^2 n)$ rounds, and the number of clustered nodes is at least $|A|/2$. Using the computed clustering as input, we apply \Cref{lem:isolatedClustering}, which computes a $(k+1)$-separated clustering $\mathcal{C^\text{out}}$ that clusters at least $|A| / 2000 \log \log n$ nodes, in $\widetilde{O}(k \cdot \log^2 n)$ rounds. We color the clusters in $\mathcal{C}^{out}$ with the $i$th color. Let $A = A \setminus (\cup \mathcal{C}^\text{out})$ be the remaining living nodes. 
    
    Each iteration clusters at least a $\log \log n$-fraction of the remaining nodes, so $O(\log n \cdot \log \log n)$ iterations is sufficient. The total runtime is $\widetilde{O}(k \cdot \log^3 n)$.
\end{proof}

\vspace{3mm}
For the MIS algorithm in \Cref{ssec:MISPower}, we additionally prove that the network decomposition algorithm of \cite{GGH22} (and consequently \Cref{thm:networkDecomp}) can be simulated on a ball graph.
The simulation is specific to the algorithm of \cite{GGH22}, and cannot be done efficiently for \congest algorithms in general. All messages from a node in the ball graph to its neighbors are sent along a Steiner tree of the ball. Communication must consist of simple primitives such as broadcast and convergecast, for efficient simulation to be possible.

\begin{claim}[Network decomposition of ball graph]\label{claim:ndAlgoBallGraph}
    Let $R \subseteq V$ and let $\{ball(v) \subseteq V : v \in R\}$ be a set of disjoint balls. Let $\mathcal{B}$ be a ball graph for $\{ball(v) \subseteq V : v \in R\}$, with nodes $R$ and an edge between $v,w \in R$ if $\ball(v)$ and $\ball(w)$ are adjacent in $G$. Assume that for each $v \in R$, there is a Steiner tree $T_v$ with weak diameter $O(r)$, with $\ball(v)$ as the terminal nodes, such that any edge in $E$ is in at most $O(\tau)$ trees. \Cref{thm:networkDecomp} can be simulated on $\mathcal{B}$ with communication network $G$ with an $O(r\cdot \tau)$ slowdown factor. 
\end{claim}
\begin{proof}
    The proof of \Cref{thm:networkDecomp} consists of iterating \Cref{lem:lowDegClustering} and \Cref{lem:isolatedClustering}. These are based on \cite[Theorem 4.1, Theorem 5.1]{GGH22}, respectively. In particular, the communication aspects remain unchanged. Communication in \cite[Theorem 4.1, Theorem 5.1]{GGH22} is based on the communication primitives of \cite[Lemma 4.7, 5.2]{GGH22}, respectively. We show that these communication primitives can be implemented for the ball graph with a slowdown factor of $O(r \cdot \tau)$, using the Steiner trees of the balls.

    \cite[Lemma 4.7]{GGH22}: Each node $v \in R$ in the ball graph starts a breadth first search after a given delay. The BFS token includes $\ID(v)$ and whether the delay of $v$ is still being decreased (whether $v \in V^{\text{active}}_i$ or not). Let $\kappa := \lceil 100 \log \log n\rceil$ be a constant. During the entire BFS process, each node $v \in R$ in the ball graph forwards at most $\kappa$ tokens to its neighbors in $\mathcal{B}$. During one \textit{step} of the BFS, a subset of at most $\kappa$ tokens are forwarded, with priority given for nodes $x \not\in V^{\text{active}}_i$. One step of the procedure can be implemented in the ball graph in $O(\kappa \cdot r \cdot \tau)$ \congest rounds, i.e., with an $O(r \cdot \tau)$ slowdown factor. A set of tokens sent from $v \in R$ is sent to nodes in $\ball(v)$ along $T_v$ by pipelining, which can be done in $O(\kappa \cdot r \cdot \tau)$ rounds. Each $w \in \ball(v)$ sends the set of tokens to its neighbors $w' \in N(w) \setminus \ball(v)$ in $O(\kappa)$ rounds. Incoming tokens are received by each $w \in \ball(v)$. A subset of at most $\tau$ tokens are forwarded to the root $v$ along $T_v$, with priority given for nodes $x \not\in V^{\text{active}}_i$, which can similarly be done in $O(\kappa \cdot r \cdot \tau)$ rounds.

    The paths taken by the tokens can be used to implement the required communication primitives. For each $v \in R$, the process defines a tree in the ball graph. The trees formed by the paths of the token contain all the necessary nodes (namely $M_{i-1}(v)$) in the ball graph, by the original analysis of the BFS procedure in \cite{GGH22}. The trees in the ball graph can be extended to $G$ by choosing one path inside the Steiner tree for each ball in the tree. Any edge in $E$ is part of at most $O(\kappa \cdot \tau)$ extended trees, since any edge is part of at most $O(\tau)$ Steiner trees of balls, and each node $v \in R$ in the ball graph forwards at most $\kappa$ tokens in total. Using the extended tree, each $v \in R$ can send one $O(\log n)$-bit message to nodes in $M_{i-1}(v)$. One step of this process runs in $O(\kappa \cdot \tau \cdot r)$ rounds, as in the BFS procedure. Reversing the direction, we can send an aggregation of messages from $M_{i-1}(v)$ to $v \in R$, where each step takes $O(\kappa \cdot \tau \cdot r)$ rounds. Hence the slowdown factor is $O(\tau \cdot r)$. 
    
    \cite[Lemma 5.2]{GGH22}: The lemma is also based on a BFS procedure, consisting of two phases. In the first phase, each center of a \textit{cluster} (subset of $R$) in the ball graph starts a BFS, propagated for at most $s$ hops in the ball graph ($s$ is the separation parameter in \cite[Theorem 5.1]{GGH22}). Each node $v \in R$ forwards at most $\kappa$ tokens per step to its neighbors in $\mathcal{B}$. In the second phase, balls that belong to a cluster propagate the set of received tokens along a tree of the cluster, from the root toward the leaves. By the clustering procedure (s-hop degree is at most $\kappa$ for clustered nodes), the number of tokens forwarded is at most $\kappa$. Both phases of this procedure can be implemented for the ball graph with an $O(\tau \cdot r)$ factor slowdown, with the same principle as for \cite[Lemma 4.7]{GGH22}. The process defines a tree in the ball graph, which is extended to $G$ by choosing one path inside the Steiner tree for each ball in the tree. The required communication primitives are similar in principle, including sending a single message from the cluster center, and convergecasting to the cluster center. These are implemented in the extended tree in the same way as above. 
\end{proof}

\section{Pseudocode for \Cref{sec:sparsGk}}
\label{app:pseudocodes}

\begin{algorithm}[h]
    \DontPrintSemicolon
    \SetAlgoLined
    \LinesNumbered
    \SetKwIF{If}{ElseIf}{Else}{if}{:}{else if}{else}{}
    \SetKwFor{For}{for}{:}{}
    \SetKwFor{ForEach}{foreach}{:}{}
    \KwIn{$k \ge 1$; Each $v \in V$ knows if it's in a set of initially active nodes $A \subseteq V$.}

    $Q_0 := A$\;

    \For{\upshape{iteration} $s= 1, \dots, k$}{
        $H_1 := Q_{s-1}$\;
        $\maxactivedeg^{(s)} := \Delta \text{ if } s=1 \text{ else } 72\Delta \log n$\;
        $r := \lfloor \log \maxactivedeg - \log \log n\rfloor-5$\;
        \For{\upshape{stage} $i= 1, \dots, r$}{
            Find $M_i \subseteq H$ using \Cref{lem:conditionalEVs} s.t. neither $\Phi_v$ nor $\Psi_v$ (\Cref{lem:simDetSparsification}) occur for any $v \in V$

            \ForEach{$v \in H_i$ in parallel}{
                \If{$v \in M_i$}{
                    Send flag \textit{sampled}, propagated to distance-$2s$ neighborhood (distance-2 neighborhood in $G^s$)\;
                }
                \If{$v \in M_i$ \upshape{or received} \textit{sampled}}{
                    Remove $v$ from $H_i$\;
                    Send broadcast (\textit{deactivated}, ID($v$)) to distance-$s$ neighborhood (distance-1 neighborhood in $G^s$) using \Cref{lem:comms}\;
                }
            }
            \ForEach{$v \in V$ in parallel}{
                Form knowledge of $N^s(v, H_{i+1}) \subseteq N^s(v, H_{i})$ (remove nodes $w \in H_i$ who sent \textit{deactivated})
            }
        }
        $M_{r+1} := H_{r+1}$\;
        $Q_s := \cup_{i=1}^{r+1} M_i$\;
        \ForEach{$v \in V$ in parallel}{
            Send $N^s(v, Q_s)$ to all neighbors $w \in N(v)$ to learn $N^{s+1}(v,Q_{s})$ and extend BFS trees (\Cref{lem:sendingIDs})\;
        }
    }
    \Return{$Q_k$}
    \caption{Sparsification for $G^k$}
    \label{algo:sparsification}
\end{algorithm}

\end{document}